\documentclass[12pt]{article}
\usepackage[utf8]{inputenc}
\usepackage[sectionbib]{natbib}

\title{Rate-Optimal Robust Estimation of High-Dimensional Vector Autoregressive Models}
\author{Di Wang and Ruey S. Tsay\\
\textit{{Booth School of Business, University of Chicago}}}

\usepackage[margin=1in]{geometry}
\usepackage{graphicx}
\usepackage{mathrsfs}
\usepackage{multirow}
\usepackage{amsfonts}
\usepackage{amsmath}
\usepackage{comment}
\usepackage{amsthm}
\usepackage{color}
\usepackage{mathtools}
\usepackage{algorithm}
\usepackage{sectsty}
\usepackage[colorlinks]{hyperref}
\usepackage{amsmath}
\usepackage{amssymb}
\usepackage[toc,page]{appendix}
\usepackage[mathscr]{euscript}
\usepackage[toc]{appendix}

\usepackage{chngcntr}
\usepackage{apptools}
\usepackage{booktabs}
\usepackage{lscape}
\usepackage{soul}
\usepackage{epstopdf}
\usepackage{array}
\usepackage{diagbox}

\newcolumntype{L}[1]{>{\raggedright\let\newline\\\arraybackslash\hspace{0pt}}m{#1}}
\newcolumntype{C}[1]{>{\centering\let\newline\\\arraybackslash\hspace{0pt}}m{#1}}
\newcolumntype{R}[1]{>{\raggedleft\let\newline\\\arraybackslash\hspace{0pt}}m{#1}}

\makeatletter
\newcommand*{\addFileDependency}[1]{
  \typeout{(#1)}
  \@addtofilelist{#1}
  \IfFileExists{#1}{}{\typeout{No file #1.}}
}
\makeatother

\usepackage{xr}

\renewcommand{\arraystretch}{1.3}

\newtheorem{assumption}{Assumption}

\newtheorem{example}{Example}
\newtheorem{lemma}{Lemma}
\newtheorem{proposition}{Proposition}
\newtheorem{theorem}{Theorem}

\newtheorem{remark}{Remark}

\hypersetup{
    colorlinks=true,
    linkcolor=blue,
    citecolor=blue,
    filecolor=magenta,
    urlcolor=cyan,
}

\DeclareMathOperator*{\argmin}{arg\,min}

\newcommand{\bm}{\mathbf}
\newcommand{\bbm}{\boldsymbol}

\mathtoolsset{showonlyrefs}

\begin{document}

\setlength{\parindent}{16pt}

\maketitle

\begin{abstract}
High-dimensional time series data appear in many scientific areas in the current data-rich environment. Analysis of such data poses new challenges to data analysts because of not only the complicated dynamic dependence between the series, but also the existence of aberrant observations, such as missing values, contaminated observations, and heavy-tailed distributions. For high-dimensional vector autoregressive (VAR) models, we introduce a unified estimation procedure that is robust to model misspecification, heavy-tailed noise contamination, and conditional heteroscedasticity. The proposed methodology enjoys both statistical optimality and computational efficiency, and can handle many popular high-dimensional models, such as sparse, reduced-rank, banded, and network-structured VAR models. With proper regularization and data truncation, the estimation convergence rates are shown to be almost optimal in the minimax sense under a bounded $(2+2\epsilon)$-th moment condition. When $\epsilon\geq1$, the rates of convergence match those obtained under the sub-Gaussian assumption. Consistency of the proposed estimators is also established for some $\epsilon\in(0,1)$, with minimax optimal convergence rates associated with $\epsilon$. The efficacy of the proposed estimation methods is demonstrated by simulation and a U.S. macroeconomic example.
  
\end{abstract}


\newpage
\section{Introduction}

\subsection{High-dimensional vector autoregression}
Vector autoregressive (VAR) models are arguably the most commonly used multivariate time series models in practice; see, e.g., \citet{lutkepohl2005new}, \citet{tsay2013multivariate}, and the references therein. Applications of the model can be found in a wide range of fields, such as economics and finance \citep{wu2016measuring}, time-course functional genomics \citep{michailidis2013autoregressive}, and neuroimaging \citep{gorrostieta2012investigating}.
Consider a $p$-dimensional zero-mean VAR model of order $d$, i.e., VAR($d$) model,
\begin{equation}\label{eq:VAR}
    \bm{y}_t=\bm{A}_1\bm{y}_{t-1}+\bm{A}_2\bm{y}_{t-2}+\cdots+\bm{A}_d\bm{y}_{t-d}+\bbm{\varepsilon}_t,~~~~t=1,2,\dots,T,
\end{equation}
where $\bm{y}_t=(y_{1t},\dots,y_{pt})^\top\in\mathbb{R}^p$ is the observed time series,  $\bm{A}_j\in\mathbb{R}^{p\times p}$ is the lag-$j$ coefficient matrix, and $\bbm{\varepsilon}_t=(\varepsilon_{1t},\dots,\varepsilon_{pt})^\top\in\mathbb{R}^{p}$ is a serially uncorrelated white noise innovation. We assume that all solutions of the determinant equation $|\bm{A}(z)| = 0$ are outside the unit circle, where $\bm{A}(z) = \bm{I}_p-\bm{A}_1 z-\cdots -\bm{A}_d z^d$ is referred to as the AR matrix polynomial in $z$. In modern applications, the dimension $p$ is often large. However, since the number of coefficient parameters is $p^2d$ and those coefficients are  often highly correlated, an unrestricted VAR($d$) model is likely to encounter the difficulty of over-parameterization and, hence, cannot provide reliable estimates nor accurate forecasts without further restrictions.

Estimation consistency of high-dimensional VAR models is achievable under some structural assumptions, provided that certain regularity conditions are satisfied. For example, if the coefficient matrices have an unobserved low-dimensional structure, such as sparsity or low-rankness, the structure-inducing regularization methods, including Lasso \citep{basu2015regularized}, Dantzig selector \citep{han2015direct}, and nuclear norm penalty \citep{negahban2011estimation}, give consistent estimates under the Gaussian assumption of the time series. Recently, \citet{zheng2019testing},  \citet{zheng2020finite}, and \citet{wang2021high} developed novel technical tools to relax the distributional assumption from Gaussian to sub-Gaussian.

However, in real applications, time series data often contain aberrant observations, which can occur in many ways, such as missing values, measurement error contamination, and heavy-tailed distribution. Those aberrant observations, if overlooked, can lead to biased estimates, erroneous inference, and sub-optimal forecasts. The situation can easily be further exacerbated when the dimension $p$ is large.  It is, therefore, important to study robust estimation of high-dimensional VAR models. In addition, the true data generating process is unlikely to follow a VAR model. Model uncertainty, including using a high-dimensional VAR($d$) model as an approximation to the true model, also deserves a careful investigation. 

Recently there have been emerging interests in studying high-dimensional VAR models with non-\textit{i.i.d.} and/or non-sub-Gaussian innovations. For example, \citet{wu2016performance} studied theoretical properties of Lasso and constrained $\ell_1$ minimization estimators for a VAR model with weakly correlated and heavy-tailed $\bbm{\varepsilon}_t$. \citet{wong2020lasso} investigated the estimation and prediction performance of Lasso for the sub-Weibull time series data under a $\beta$-mixing condition. Both theoretical and numerical results in the literature show that the performance of standard $\ell_1$ regularized estimators deteriorates substantially when the data have heavy tails. For robust estimation of high-dimensional heavy-tailed time series data, \citet{qiu2015robust} developed a quantile-based Dantzig selector for the class of elliptical VAR processes. \citet{han2020robust} proposed a robust estimation method for high-dimensional sparse generalized linear models with temporal dependent covariates. However, the existing literature on robust estimation for time series data focuses on sparse models. To the best of our knowledge, there is no unified solution to address the robust estimation problem for a large class of high-dimensional VAR models.

Our proposed robust estimation procedure is built on two key ingredients: the constrained Yule--Walker estimator and the robust autocovariance matrix estimator. The first ingredient, the constrained Yule--Walker estimator, provides a general and flexible estimation framework for two classes of high-dimensional models, namely the approximately low-dimensional VAR models and linear-restricted VAR models. The second ingredient, the robust autocovariance matrix estimator, is easy to implement by truncating the time series data. However, based on the specific model structure, we need to adapt the data truncation methods. With a large class of distributions having bounded second or higher-order moments, the proposed estimators are shown to be consistent under high-dimensional scaling. To be specific, we summarize the main contributions of this paper as follows:
\begin{itemize}
    \item[(i)] For various high-dimensional VAR models, the paper provides a simple and general estimation procedure robust to model misspecification, heavy-tailed noise contamination, and conditional heteroskedasticity. Our proposal can handle many popular high-dimensional VAR models, such as sparse, reduced-rank, banded, and network VAR models. An efficient and scalable alternating direction method of multipliers (ADMM) algorithm is developed.
    \item[(ii)] The proposed methodology enjoys statistical optimality in the minimax sense. Our theoretical framework deals with heavy-tailed distributions with bounded $(2+2\epsilon)$-th moments for any $\epsilon>0$. It results in a phase transition on the rates of convergence: for $\epsilon\geq1$, the estimator achieves the same convergence rates as those obtained under the Gaussian or sub-Gaussian distribution, while consistency is also established with a slower rate for $\epsilon\in(0,1)$. By establishing the matching minimax lower bounds, we show that the proposed estimators for high-dimensional VAR models and autocovariance matrices are rate-optimal.
\end{itemize}

\subsection{Related literature}

This work is related to a huge body of literature on the robust estimation of high-dimensional regression and covariance matrices. The early developments in robust statistics were pioneered by \citet{huber1964robust} and \citet{hampel1971general,hampel1974influence}; see also the overview by \citet{hampel2001robust} and the references therein. For time series data, robust estimation methods were developed for univariate and multivariate ARMA models \citep{martin1981robust,muler2009robust,muler2013robust}. For high-dimensional \textit{i.i.d.} data, inspired by \citet{catoni2012challenging}, the non-asymptotic deviation analysis for heavy-tailed variables and robust $M$-estimators for high-dimensional regression were proposed by \citet{fan2017estimation}, \citet{loh2017statistical}, \citet{sun2020adaptive}, \citet{wang2020tuning}, and many others. These recent robust estimation methods and their theoretical results were further extended to high-dimensional covariance and precision matrix estimation problems, such as \citet{avella2018robust}, \citet{minsker2018sub}, and \citet{zhang2021robust}. Another research line of robust regression is the least absolute deviation loss, and more generally, quantile loss, which have been studied extensively; see, e.g., \citet{wang2007robust}, \citet{belloni2011l1}, and \citet{wang2012quantile}. Recently, \citet{fan2016shrinkage} and \citet{ke2019user} independently proposed truncated variants of the sample covariance, which are easy to implement and motivate us to develop the robust estimators of high-dimensional autocovariance matrices.

The upper and lower bound development for high-dimensional estimation problems under heavy-tailed distributions is another emerging and important research topic. Under heavy-tailed distributions, the phase transition phenomenon in the rate of convergence was previously discovered by \citet{bubeck2013bandits}, \citet{avella2018robust}, \citet{sun2020adaptive}, \citet{tan2018robust}, and others. For the lower bound development, the phase transition phenomenon was first established by \citet{devroye2016sub} for univariate mean estimation problem, and was extended to the fixed and high-dimensional linear regression problems in \citet{sun2020adaptive}. Compared with the existing literature, our minimax lower bound results are established under a much more complicated setting. First, all existing lower bound results are developed for \textit{i.i.d.} data, but we allow weak serial dependency and establish the minimax lower bounds under strong mixing conditions. Second, for the lower bounds of the linear regression problem in \citet{sun2020adaptive}, a finite moment condition is imposed on the random error terms while the covariates are assumed to be light-tailed. Without any assumption on data generating process, we consider the finite moment condition on the observed time series data that serve as both predictors and responses in the autoregressive models, which is fundamentally different from the conventional linear regression problem.

\subsection{Notation and outline}

We start with some notations used in the paper.
Let $C$ denote a generic positive constant, which is independent of the dimension and sample size.
For any two real-valued sequences $x_k$ and $y_k$, $x_k\gtrsim y_k$ if there exists a $C>0$ such that $x_k\geq Cy_k$ for all $k$.
In addition, we write $x_k\asymp y_k$ if $x_k\gtrsim y_k$ and $y_k\gtrsim x_k$.
For any two real numbers $x$ and $y$, let $x\wedge y$ denote their minimum.
Throughout the paper, we use bold lowercase letters to denote vectors. For any vector $\bm{v}=(v_1,v_2,\dots,v_d)^\top$, denote its $\ell_q$ norm as $\|\bm{v}\|_q=(\sum_{i=1}^dv_i^q)^{1/q}$ for $1\leq q<\infty$, its maximum norm as $\|\bm{v}\|_\infty=\max_{1\leq i\leq d}|v_i|$, and its $\ell_0$ norm as $\|\bm{v}\|_0=\sum_{i=1}^d1\{v_i\neq0\}$. 
We use bold uppercase letters to denote matrices. 
For any matrix $\bm{M}=[\bm{m}_1,\dots,\bm{m}_{d_2}]\in\mathbb{R}^{d_1\times d_2}$, where $\bm{m}_j\in\mathbb{R}^{d_1}$ is the $j$-th column of $\bm{M}$, denote its $\ell_{p,q}$ norm as $\|\bm{M}\|_{p,q}=(\sum_{j=1}^{d_2}\|\bm{m}_{j}\|_p^q)^{1/q}$, for any $1\leq p,q\leq\infty$, and its $i$th largest singular value as $\sigma_i(\bm{M})$, for all $i=1,2,\dots,d_1\wedge d_2$. For a given matrix $\bm{M}\in\mathbb{R}^{d_1\times d_2}$, we let $\|\bm{M}\|_\textup{F}$, $\|\bm{M}\|_\textup{op}$, and $\|\bm{M}\|_\textup{nuc}$ denote its Frobenius norm, operator norm, and nuclear norm, respectively, where $\|\bm{M}\|_\textup{F}=\|\bm{M}\|_{2,2}$, $\|\bm{M}\|_\textup{op}=\sigma_1(\bm{M})$, and $\|\bm{M}\|_\textup{nuc}=\sum_{i=1}^{d_1\wedge d_2}\sigma_i(\bm{M})$. For any two matrices $\bm{M}_1$ and $\bm{M}_2$, their Kronecker product is $\bm{M}_1\otimes\bm{M}_2$. For any subspace $\mathcal{M}\subset\mathbb{R}^{d_1\times d_2}$ and any matrix $\bm{M}\in\mathbb{R}^{d_1\times d_2}$, denote by $\mathcal{M}^\perp$ and $\bm{M}_{\mathcal{M}}$ the orthogonal complement of $\mathcal{M}$ and the projection of $\bm{M}$ onto $\mathcal{M}$, respectively.

The rest of the paper is organized as follows. Section \ref{sec:2} studies constrained Yule--Walker estimators for two classes of VAR models: VAR with approximately low-dimensional structure and VAR with linear restrictions. Section \ref{sec:3} develops robust autocovariance estimators tailored for various low-dimensional structures. Minimax lower bounds for both high-dimensional VAR estimation and autocovariance matrix estimation are investigated in Section \ref{sec:lower_bound}. Computational algorithms and implementation details are discussed in Section \ref{sec:4}. Section \ref{sec:5} presents some simulation results, and Section \ref{sec:6} 
shows an empirical application of U.S. macroeconomic data. All technical proofs and detailed algorithms are relegated to the Appendices. The codes and data can be found at \href{https://github.com/diwangstat/RobustAR}{https://github.com/diwangstat/RobustAR}.

\section{Constrained Yule--Walker Estimation}\label{sec:2}

\subsection{Yule--Walker equation}\label{sec:2.1}

Consider a general mean-zero and covariance stationary process $\{\bm{y}_t\}_{t=1}^T$, where $\bm{y}_t\in\mathbb{R}^p$. In many applications, it is common to model the time series data and predict their future values using a linear VAR model.
The VAR($d$) model in \eqref{eq:VAR} can be rewritten as 
\begin{equation}\label{eq:VAR2}
    \bm{y}_t=\bm{A}\bm{x}_t+\bbm{\varepsilon}_t,
\end{equation}
where $\bm{x}_t=(\bm{y}_{t-1}^\top,\dots,\bm{y}_{t-d}^\top)^\top\in\mathbb{R}^{pd}$ is the predictor vector and $\bm{A}=[\bm{A}_1,\bm{A}_2,\cdots,\bm{A}_d]\in\mathbb{R}^{p\times pd}$ is the combined parameter matrix.
Throughout this paper, we consider that the lag order $d$ is fixed. 
Given the stationarity of $\bm{y}_t$, the parameter matrix of interest $\bm{A}^*\in\mathbb{R}^{p\times pd}$ is defined as the minimizer of the risk function
\begin{equation}
    \label{eq:Astar}
    \bm{A}^*:=\underset{\bm{A}\in\mathbb{R}^{p\times pd}}{\arg\min}~\mathbb{E}\left[\|\bm{y}_t-\bm{A}\bm{x}_t\|_2^2\right].
\end{equation}

For any integer $\ell$, denote the lag-$\ell$ autocovariance matrix of $\bm{y}_t$ by  $\bm{\Gamma}_\ell=\mathbb{E}[\bm{y}_t\bm{y}_{t-\ell}^\top]$.
By simple algebra,
\begin{equation}
    \begin{split}
        \|\bm{y}_t-\bm{A}\bm{x}_t\|_2^2=&\|\bm{y}_t-\textup{vec}(\bm{x}_t^\top\bm{A}^\top)\|_2^2 =\|\bm{y}_t-(\bm{I}_p\otimes\bm{x}_t^\top)\textup{vec}(\bm{A}^\top)\|_2^2\\
        =&~\bm{y}_t^\top\bm{y}_t-2\bm{y}_t^\top(\bm{I}_p\otimes\bm{x}_t^\top)\textup{vec}(\bm{A}^\top)+\textup{vec}(\bm{A}^\top)^\top(\bm{I}_p\otimes\bm{x}_t\bm{x}_t^\top)\textup{vec}(\bm{A}^\top)\\
        =&~\bm{y}_t^\top\bm{y}_t-2\textup{vec}(\bm{x}_t\bm{y}_t^\top)^\top\textup{vec}(\bm{A}^\top)+\textup{vec}(\bm{A}^\top)^\top(\bm{I}_p\otimes\bm{x}_t\bm{x}_t^\top)\textup{vec}(\bm{A}^\top),
    \end{split}
\end{equation}
which implies that $\bm{\Sigma}_1=\bm{A}^*\bm{\Sigma}_0$ or $\bm{A}^*=\bm{\Sigma}_1\bm{\Sigma}_0^{-1}$ if $\bm{\Sigma}_0$ is invertible, where $\bm{\Sigma}_1:=\mathbb{E}[\bm{y}_t\bm{x}_t^\top]=[\bm{\Gamma}_1,\bm{\Gamma}_2,\dots,\bm{\Gamma}_d]$ and
\begin{equation}
    \bm{\Sigma}_0:=\mathbb{E}[\bm{x}_t\bm{x}_t^\top]=
    \begin{bmatrix}
        \bm{\Gamma}_0 & \bm{\Gamma}_1 & \cdots & \bm{\Gamma}_{d-1}\\
        \bm{\Gamma}_1^\top & \bm{\Gamma}_0 & \cdots & \bm{\Gamma}_{d-2}\\
        \vdots & \vdots & \ddots & \vdots\\
        \bm{\Gamma}_{d-1}^\top & \bm{\Gamma}_{d-2}^\top & \cdots & \bm{\Gamma}_0
    \end{bmatrix}.
\end{equation}
This relationship between the parameter matrix $\bm{A}^*$ and autocovariance matrices is well known as the multivariate Yule--Walker equation for VAR($d$) models \citep{lutkepohl2005new,tsay2013multivariate}. 
When the data generating mechanism of $\bm{y}_t$ is not the VAR($d$) process, the Yule--Walker equation still holds if VAR is used as a running model.

The existence of aberrant observations, such as missing values or data contaminated by measurement errors, is ubiquitous in high-dimensional macroeconomic, environmental, and 
genetic data, among many other applications. In a VAR process, when some of the variables are removed as they cannot be observed or measured, the rest of the available variables generally would not follow any finite-order VAR model \citep{lutkepohl2005new}. In addition, if the true signal of the time series follows a VAR process and the observed time series $\bm{y}_t$ is a contaminated version of the true signal plus a measurement error, $\bm{y}_t$ no longer follows a VAR process.
Hence, the VAR model assumption commonly used in real applications is questionable, but the Yule--Walker equation which holds without any data generating mechanism assumption can be utilized to construct estimation methodology for the general covariance stationary time series.

In addition to the sub-Gaussian innovation assumption, another key limitation in the existing literature is the \textit{i.i.d.} assumption for the innovation series $\bbm{\varepsilon}_t$. The conditional heteroskedasticity is often observed in financial time series, and violation of the  homogeneous error assumption can affect the performance of standard estimation methods. Another major advantage of our moment-based methodology is that the \textit{i.i.d.} assumption for $\bbm{\varepsilon}_t$ can be relaxed to a serially uncorrelated, but weakly stationary condition. Specifically, our theoretical analysis allows some serial dependence in $\bbm{\varepsilon}_t$, dependence between $\bbm{\varepsilon}_t$ and $\bm{x}_t$, and conditional heteroskedasticity in $\bbm{\varepsilon}_t$. Note that this setting of innovations includes many important conditional heteroskedasticity models, such as GARCH models \citep{engle1982autoregressive,bollerslev1986generalized} and double AR models \citep{ling2004estimation,zhu2018linear}.

Based on the Yule--Walker equation, we propose a  constrained minimization estimation approach, named the constrained Yule--Walker estimation, for two classes of high-dimensional VAR($d$) models. In Section \ref{sec:2.2}, we consider the VAR($d$) model, where $\bm{A}^*$ in \eqref{eq:Astar} can be approximated by a matrix with some latent low-dimensional structure. In Section \ref{sec:2.3}, we consider the VAR($d$) model whose coefficients are subject to some linear restrictions. These two classes of models are shown to encompass many commonly-used high-dimensional VAR models in practice.

\subsection{Approximately low-dimensional VAR}\label{sec:2.2}

We first consider the situations where the dimension $p$ is relatively large compared to the sample size $T$ and the parameter matrix of interest $\bm{A}^*$ in \eqref{eq:Astar} can be well approximated by a matrix with certain types of low-dimensional structure. The approximately low-dimensional structure, such as weak sparsity and approximate low-rankness, is general and natural in high-dimensional time series modeling.

Based on the Yule--Walker equation, we consider a general class of constrained estimation:
\begin{equation}
    \label{eq:constrainedYW}
    \bm{\widehat{A}} = \underset{\bm{A}\in\mathbb{R}^{p\times pd}}{\arg\min}~\mathcal{R}(\bm{A})~~~\text{such that}~~\mathcal{R}^*(\bm{\widetilde{\Sigma}}_1-\bm{A}\bm{\widetilde{\Sigma}}_0)\leq\lambda,
\end{equation}
where $\mathcal{R}(\cdot)$ is a matrix norm as the regularization function, $\mathcal{R}^*(\cdot)$ is its dual norm as the constraint function, $\lambda$ is the constraint parameter, and $\bm{\widetilde{\Sigma}}_1$ and $\bm{\widetilde{\Sigma}}_0$ are robust autocovariance estimators to be specified later in Section \ref{sec:3}. When $\lambda$ is sufficiently small, the constraint function $\mathcal{R}^*(\cdot)$ guarantees that the sample version of the Yule--Walker equation holds roughly, and the regularizer $\mathcal{R}(\cdot)$ induces the low-dimensional structure and improves the estimation efficiency.

Following the framework for high-dimensional analysis in \citet{negahban2012unified}, we consider a decomposable regularizer $\mathcal{R}(\cdot)$. For a generic low-dimensional structure, define $\mathcal{M}\subseteq\overline{\mathcal{M}}\subset\mathbb{R}^{p\times pd}$, where $\mathcal{M}$ is referred to as the model subspace to represent the specific model constraints; for instance, it can be the subspace of low-rank matrices (see Example \ref{ex:reduced-rank}). The orthogonal complement of $\overline{\mathcal{M}}$, denoted by $\overline{\mathcal{M}}^\perp$, is the associated perturbation subspace to capture the deviation from the model subspace and is adopted to measure the error of approximation to the low-dimensional structure. We assume that the true value $\bm{A}^*$ can be decomposed into its projections onto $\mathcal{M}$ and $\overline{\mathcal{M}}^\perp$, i.e., $\bm{A}^*=\bm{A}^*_{\mathcal{M}}+\bm{A}^*_{\overline{\mathcal{M}}^\perp}$.

We call a regularizer $\mathcal{R}(\cdot)$ decomposable with respect to a pair of subspaces $(\mathcal{M},\overline{\mathcal{M}}^\perp)$, if for any $\bm{W}_1\in\mathcal{M}$ and $\bm{W}_2\in\overline{\mathcal{M}}^\perp$,
\begin{equation}
    \mathcal{R}(\bm{W}_1+\bm{W}_2) = \mathcal{R}(\bm{W}_1) + \mathcal{R}(\bm{W}_2).
\end{equation}
Many of the commonly-used convex regularizers, such as the nuclear norm for low-rankness, are shown to be decomposable; see \citet{negahban2012unified} for more details of decomposable regularizers. To measure the magnitude of the low-dimensional structure, define by $\phi(\overline{\mathcal{M}})$ the associated constant such that $\mathcal{R}(\bm{W})\leq \phi(\overline{\mathcal{M}})\mathcal{R}^*(\bm{W})$ for any $\bm{W}\in\overline{\mathcal{M}}$. In our theoretical analysis, we consider another auxiliary matrix norm $\mathcal{C}(\cdot)$ such that $(\mathcal{R}(\cdot),\mathcal{R}^*(\cdot),\mathcal{C}(\cdot))$ satisfy that
$\mathcal{R}^*(\bm{W}_1\bm{W}_2)\leq \mathcal{R}^*(\bm{W}_1)\mathcal{C}(\bm{W}_2)$ and $\mathcal{R}^*(\bm{W}_1\bm{W}_2)\leq \mathcal{R}(\bm{W}_1)\mathcal{R}^*(\bm{W}_2)$ for any symmetric matrix $\bm{W}_2$ and any compatible matrix $\bm{W}_1$.
The following example illustrates the suitable choice of $\mathcal{R}(\cdot)$, $\mathcal{R}^*(\cdot)$, and $(\mathcal{M},\overline{\mathcal{M}}^\perp)$ for a reduced-rank VAR model.

\begin{example}[Reduced-rank VAR]
    \label{ex:reduced-rank}
    The reduced-rank VAR model is an important approach to modeling high-dimensional time series by imposing a low-rank structure on $\bm{A}$; see also \citet{velu2013multivariate}, \citet{basu2019low}, and \citet{wang2019high}. To induce low-rankness, we use the nuclear norm as the regularizer and the operator norm as the constraint function; that is, $\mathcal{R}(\cdot)=\|\cdot\|_{\textup{nuc}}$ and $\mathcal{R}^*(\cdot)=\|\cdot\|_{\textup{op}}$. By the submultiplicative property of the operator norm, $\mathcal{C}(\cdot)=\|\cdot\|_{\textup{op}}$. For the approximately low-rank matrix $\bm{A}^*$, denote by $\mathcal{U}$ and $\mathcal{V}$ the subspace spanned by its leading left and right singular vectors, respectively. Define the model subspace
    \begin{equation}
        \mathcal{M}(\mathcal{U},\mathcal{V})=\{\bm{W}\in\mathbb{R}^{p\times pd}:\text{col}(\bm{W})\subseteq\mathcal{U},~\text{col}(\bm{W}^\top)\subseteq\mathcal{V}\},
    \end{equation}
    and the perturbation space
    \begin{equation}
        \overline{\mathcal{M}}^\perp(\mathcal{U},\mathcal{V})=\{\bm{W}\in\mathbb{R}^{p\times pd}:\text{col}(\bm{W})\perp\mathcal{U},~\text{col}(\bm{W}^\top)\perp\mathcal{V}\}.
    \end{equation}
\end{example}

For the general mean-zero and covariance stationary process $\{\bm{y}_t\}_{t=1}^T$, if we model the data by a VAR($d$) model in \eqref{eq:VAR} and apply the constrained Yule--Walker estimator in \eqref{eq:constrainedYW} with robust autocovariance estimators $\bm{\widetilde{\Sigma}}_0$ and $\bm{\widetilde{\Sigma}}_1$, we have the following upper bounds for the estimation error between the estimated matrix $\bm{\widehat{A}}$ and the true value $\bm{A}^*$ defined in \eqref{eq:Astar}.

\begin{proposition}
    \label{prop:1}
    For the constrained Yule--Walker estimator in \eqref{eq:constrainedYW},
    suppose that $\mathcal{R}^*(\bm{\widetilde{\Sigma}}_0-\bm{\Sigma}_0)\leq\zeta_0$, $\mathcal{R}^*(\bm{\widetilde{\Sigma}}_1-\bm{\Sigma}_1)\leq\zeta_1$, and $\lambda\geq \zeta_0\mathcal{R}(\bm{A}^*)+\zeta_1$. Then, for $(\mathcal{M},\overline{\mathcal{M}}^\perp)$,
    \begin{equation}
        \begin{split}
            \mathcal{R}^*(\bm{\widehat{A}}-\bm{A}^*) &\leq 2\hspace{0.01in}\mathcal{C}(\bm{\Sigma}_0^{-1})\lambda,\\
            \mathcal{R}(\bm{\widehat{A}}-\bm{A}^*) &\leq 4\hspace{0.01in}\mathcal{C}(\bm{\Sigma}_0^{-1})\phi(\overline{\mathcal{M}})\lambda + 2\mathcal{R}(\bm{A}^*_{\overline{\mathcal{M}}^\perp}),\\
            \text{and}~~\|\widehat{\bm{A}}-\bm{A}^*\|_\textup{F}^2 & \leq 8\hspace{0.01in}\mathcal{C}(\bm{\Sigma}_0^{-1})^2\phi(\overline{\mathcal{M}})\lambda^2+4\hspace{0.01in}\mathcal{C}(\bm{\Sigma}_0^{-1})\mathcal{R}(\bm{A}^*_{\overline{\mathcal{M}}^\perp})\lambda.
            \end{split}
    \end{equation}
\end{proposition}

This proposition presents the deterministic estimation error bounds in $\mathcal{R}^*(\cdot)$, $\mathcal{R}(\cdot)$, and $\|\cdot\|_\textup{F}^2$, given that the estimators $\bm{\widetilde{\Sigma}}_0$ and $\bm{\widetilde{\Sigma}}_1$ satisfy some regularity conditions. 
The estimation error bounds are related to $\mathcal{R}(\bm{A}^*)$ and $\mathcal{C}(\bm{\Sigma}_0^{-1})$, and these terms could be bounded or diverge slowly as the dimension $p$ increases; see more discussions for each specific model in Section \ref{sec:3}.
When the true parameter matrix $\bm{A}^*$ admits a low-dimensional approximation rather than an exact structure, $\mathcal{R}(\bm{A}^*_{\overline{\mathcal{M}}^\perp})$ is strictly positive, and the second terms in the second and third upper bounds are related to the approximation error.
In this case, the upper bounds in this proposition hold uniformly for a class of low-dimensional structures with respect to $(\mathcal{M},\overline{\mathcal{M}}^\perp)$, and the class of upper bounds can be minimized by balancing the estimation error on $\overline{\mathcal{M}}$ and the approximation error on $\overline{\mathcal{M}}^\perp$. 

Moreover, in some cases, it is of interest to consider the approximate low dimensionality in a more structured manner. In particular, based on the multi-response nature, the VAR model in \eqref{eq:VAR2} can be split into $p$ sub-models:
\begin{equation}\label{eq:subVAR}
    y_{it}=\bm{a}_i^\top\bm{x}_t+\varepsilon_{it}
\end{equation}
where $\bm{a}_i^\top$ is the $i$-th row of $\bm{A}$, for $i=1,\dots,p$. For these sub-models, as each $\bm{a}_i$ is a $pd$-dimensional vector, we may consider some low-dimensional structures, such as sparsity, on $\bm{a}_i$. Then, we consider the constrained estimation framework for each $\bm{a}_i$:
\begin{equation}\label{eq:split_YW}
    \widehat{\bm{a}}_i=\argmin_{\bm{a}\in\mathbb{R}^{pd}}\mathcal{R}(\bm{a})~~\text{such that}~~\mathcal{R}^*(\widetilde{\bbm{\sigma}}_{1i}-\widetilde{\bm{\Sigma}}_0\bm{a})\leq\lambda,
\end{equation}
where $\mathcal{R}(\cdot)$ is a vector norm, $\mathcal{R}^*(\cdot)$ is its dual norm, and $\widetilde{\bbm{\sigma}}_{1i}^\top$ is the $i$-th row of $\widetilde{\bm{\Sigma}}_1$. For simplicity, we consider the same constraint parameter $\lambda$ for all sub-problems. In this framework, these $p$
sub-problems can be handled in parallel and thus can be solved efficiently.

In an analogous fashion, we consider a pair of subspaces of $\mathbb{R}^{pd}$, $(\mathcal{M}_i,\overline{\mathcal{M}}_i^\perp)$, for each $\bm{a}_i$, where $\bm{a}^*_{i}=(\bm{a}^*_i)_{\mathcal{M}_i}+(\bm{a}^*_i)_{\overline{\mathcal{M}}_i^\perp}$ and $\mathcal{R}(\bm{v}_1+\bm{v}_2)=\mathcal{R}(\bm{v}_1)+\mathcal{R}(\bm{v}_2)$ for any $\bm{v}_1\in\mathcal{M}_i$ and $\bm{v}_2\in\overline{\mathcal{M}}_i^\perp$, for all $i=1,\dots,p$. In addition, define by $\phi(\overline{\mathcal{M}}_i)$ the constant such that $\mathcal{R}(\bm{v})\leq \phi(\overline{\mathcal{M}}_i)\mathcal{R}^*(\bm{v})$ for any $\bm{v}\in\overline{\mathcal{M}}_i$. Moreover, for theoretical analysis, we consider the auxiliary matrix norms $\overline{\mathcal{R}}^*(\cdot)$ and $\mathcal{C}(\cdot)$ satisfying that $\mathcal{R}^*(\bm{W}\bm{v})\leq\overline{\mathcal{R}}^*(\bm{W})\mathcal{R}(\bm{v})$ and $\mathcal{R}^*(\bm{W}^\top\bm{v})\leq\mathcal{C}(\bm{W})\mathcal{R}^*(\bm{v})$ for any compatible matrix $\bm{W}$ and vector $\bm{v}$. The following example shows that the $\ell_1$ regularization can be applied to the sparse VAR model.

\begin{example}[Sparse VAR]
    \label{ex:sparse}
    The sparse VAR model assumes that each $\bm{a}_i$ is strictly or weakly sparse, so that each response variable $y_{it}$ is only related to a subset of lagged values $y_{j,t-\ell}$. Though the overall sparsity in $\bm{A}$ has been considered in many existing literature, such as \citet{basu2015regularized} and \citet{kock2015oracle}, the sparsity structure in rows seems more natural in the time series context. For example, in the diagonal VAR(1) model, the total degree of sparsity grows with $p$, while the sparsity level in each row remains fixed. To induce sparsity in each row, as in \citet{han2015direct}, we consider $\mathcal{R}(\cdot)=\|\cdot\|_{1}$, $\mathcal{R}^*(\cdot)=\|\cdot\|_{\infty}$, $\overline{\mathcal{R}}^*(\cdot)=\|\cdot\|_{\infty}$, and $\mathcal{C}(\cdot)=\|\cdot\|_{1,\infty}$. For the weakly sparse $\bm{a}_i$, $i=1,\dots,p$, denote by $S_i\subset\{j:1\leq j\leq pd\}$ any index set corresponding to those coefficients significantly distant from zero in $\bm{a}_i$. For each $i$, define the model subspace associated with the chosen index set $S_i$ as
    \begin{equation}
        \mathcal{M}_i(S_i)=\{\bm{v}\in\mathbb{R}^{pd}:\bm{v}_{j}=0~\text{for all}~j\notin S_i\},
    \end{equation}
    and the perturbation subspace
    \begin{equation}
        \overline{\mathcal{M}}_i^\perp(S_i)=\{\bm{v}\in\mathbb{R}^{pd}:\bm{v}_{j}=0~\text{for all}~j\in S_i\}.
    \end{equation}
\end{example}

For the general mean-zero and covariance stationary process $\{\bm{y}_t\}_{t=1}^T$, if we separate the VAR($d$) model to $p$ sub-problems in \eqref{eq:subVAR} and apply the separated constrained Yule--Walker estimator in \eqref{eq:split_YW} with robust autocovariance estimators $\bm{\widetilde{\Sigma}}_0$ and $\widetilde{\bbm{\sigma}}_{1i}$, we have the following upper bounds for the estimation error between $\bm{\widehat{a}}_i$ and the true value $\bm{a}_i^*$.

\begin{proposition}
    \label{prop:2}
    For the separated Yule--Walker estimators in \eqref{eq:split_YW}, suppose that $\overline{\mathcal{R}}^*(\widetilde{\bm{\Sigma}}_0-\bm{\Sigma}_0)\leq\zeta_0$, $\mathcal{R}^*(\widetilde{\bbm{\sigma}}_{1i}-\bbm{\sigma}_{1i})\leq\zeta_{1i}$, and $\lambda\geq\max_{1\leq i\leq p}[\zeta_0\mathcal{R}(\bm{a}_i^*)+\zeta_{1i}]$. Then, for $(\mathcal{M}_i,\overline{\mathcal{M}}_i^\perp)$,
    \begin{equation}
        \begin{split}
            \mathcal{R}^*(\widehat{\bm{a}}_i-\bm{a}_i^*) & \leq 2\hspace{0.01in}\mathcal{C}(\bm{\Sigma}_0^{-1})\lambda,\\
            \mathcal{R}(\widehat{\bm{a}}_i-\bm{a}_i^*) & \leq 4\hspace{0.01in}\mathcal{C}(\bm{\Sigma}_0^{-1})\phi(\overline{\mathcal{M}}_i)\lambda+2\mathcal{R}(\bm{A}^*_{\overline{\mathcal{M}}_i^\perp}),\\
            \text{and}~~\|\widehat{\bm{a}}_i-\bm{a}_i^*\|_2^2 & \leq 8\hspace{0.01in}\mathcal{C}(\bm{\Sigma}_0^{-1})^2\phi(\overline{\mathcal{M}}_i)\lambda^2 + 4\hspace{0.01in}\mathcal{C}(\bm{\Sigma}_0^{-1})\mathcal{R}(\bm{A}^*_{\overline{\mathcal{M}}_i^\perp})\lambda
        \end{split}
    \end{equation}
    for $i=1,\dots,p$.
\end{proposition}

For the approximately low-rank and sparse VAR models discussed in Examples \ref{ex:reduced-rank} and \ref{ex:sparse}, we can obtain a sharp upper bound by appropriately selecting the matrix rank or the degree of sparsity. More concrete results are given in Section \ref{sec:3}.

\subsection{Linear-restricted VAR}\label{sec:2.3}

In modern high-dimensional time series modeling, domain knowledge and prior information are commonly available in addition to the time series data. In many such cases, prior information can be formulated as linear restrictions on the parameter matrix, and in this subsection, we apply the constrained Yule--Walker estimation approach to the high-dimensional VAR models with linear restrictions.

As discussed in \citet{tsay2013multivariate}, a general linear parameter restriction can be expressed as $\textup{vec}(\bm{A}^\top)=\bm{C}\bbm{\theta}+\bbm{\gamma}$,
where $\bm{C}\in\mathbb{R}^{p^2d\times r}$ is a prespecified constraint matrix, $\bbm{\gamma}\in\mathbb{R}^{p^2d}$ is a known constant vector, and $\bbm{\theta}\in\mathbb{R}^{r}$ is the unknown parameter vector.
Note that the linear constraint representation is not unique; that is, for any nonsingular matrix $\bm{O}\in\mathbb{R}^{r\times r}$, $\bm{C}\bbm{\theta}=(\bm{C}\bm{O})(\bm{O}^{-1}\bbm{\theta})$. Thus, we require that $\bm{C}$ is a tall orthonormal matrix, i.e., $\bm{C}^\top\bm{C}=\bm{I}_r$. Otherwise, we can apply singular value decomposition or QR decomposition to $\bm{C}$.

For simplicity, we consider the case where $\bbm{\gamma}=\bm{0}$. In this case, we can rewrite the VAR($d$) model in \eqref{eq:VAR2} to
\begin{equation}
    \label{eq:linearVAR}
    \bm{y}_t=(\bm{I}_p\otimes\bm{x}_t^\top)\bm{C}\bbm{\theta}+\bbm{\varepsilon}_t.
\end{equation}
Similarly to the discussions in Section \ref{sec:2.1}, for the general stationary time series $\bm{y}_t$ which may not strictly follow a linear-restricted VAR process, the parameter of interest $\bbm{\theta}^*$ is defined as
\begin{equation}
    \label{eq:thetaStar}
    \bbm{\theta}^*=\argmin_{\bbm{\theta}\in\mathbb{R}^r}\mathbb{E}\left[\|\bm{y}_t-(\bm{I}_p\otimes\bm{x}_t^\top)\bm{C}\bbm{\theta}\|_2^2\right].
\end{equation}
It implies the linear-restricted version of the Yule--Walker equation
\begin{equation}
    \bbm{\theta}^*=[\bm{C}^\top(\bm{I}_p\otimes\bm{\Sigma}_0)\bm{C}]^{-1}\bm{C}^\top\textup{vec}(\bm{\Sigma}_1^\top)~~\text{or}~~\bm{C}^\top\textup{vec}(\bm{\Sigma}_1^\top)=\bm{C}^\top(\bm{I}_p\otimes\bm{\Sigma}_0)\bm{C}\bbm{\theta}^*,
\end{equation}
provided that $\bm{C}^\top(\bm{I}_p\otimes\bm{\Sigma}_0)\bm{C}$ is invertible. 

Let $\bm{\Omega}=\bm{C}^\top(\bm{I}_p\otimes\bm{\Sigma}_0)\bm{C}\in\mathbb{R}^{r\times r}$ and $\bbm{\omega}=\bm{C}^\top\textup{vec}(\bm{\Sigma}_1^\top)\in\mathbb{R}^r$. For $\bm{\Omega}$ and $\bbm{\omega}$, suppose that we have robust estimators $\widetilde{\bm{\Omega}}$ and $\widetilde{\bbm{\omega}}$. Then, we consider the constrained Yule--Walker estimator
\begin{equation}
    \label{eq:YuleWalker2}
    \bbm{\widehat{\theta}}=\arg\min\|\bbm{\theta}\|_\infty
    ~~~\textup{subject to}~~\|\bm{\widetilde{\Omega}}\bbm{\theta}-\bbm{\widetilde{\omega}}\|_\infty\leq\lambda,
\end{equation}
where $\lambda$ is the constraint parameter and the $\ell_\infty$ norm provides both element-wise regularization on $\bbm{\theta}$ and element-wise constraint on the sample version of Yule--Walker equation. Decompose the restriction matrix $\bm{C}=[\bm{B}_1^\top,\dots,\bm{B}_p^\top]^\top$, where each $\bm{B}_i$ is a $pd\times r$ matrix. Let $\bm{B}=[\bm{B}_1,\dots,\bm{B}_p]\in\mathbb{R}^{pd\times pr}$ and then the estimator of $\bm{A}$ is $\bm{\widehat{A}}=(\bm{I}_p\otimes\bbm{\widehat{\theta}})^\top\bm{B}^\top$.

The linear-restricted model encompasses several important recent developments in high-dimensional vector autoregression.

\begin{example}[Banded VAR]
    \label{ex:banded}
    The banded VAR(1) proposed by \citet{guo2016high} has the following banded coefficient structure:
    \begin{equation}
        \bm{A}_{ij}=0,~~\text{for all }|i-j|>k_0,
    \end{equation}
    where $k_0$ is called the bandwidth parameter. 
    The banded structure is a special sparse structure, so we can formulate it to a linear constraint $\textup{vec}(\bm{A}^\top)=\bm{C}\bbm{\theta}$, where each column of $\bm{C}$ is a coordinate vector corresponding to one nonzero index in $\textup{vec}(\bm{A}^\top)$.
\end{example}

\begin{example}[Network VAR] \label{ex:network}
    \citet{zhu2017network} proposed the network VAR model for the network time series data $\bm{y}_t$, which has an observable network structure with the adjacency matrix $\bm{W}$. The network VAR model assumes the form
    \begin{equation}
        \bm{y}_t=\beta_1\bm{y}_{t-1}+\beta_2\bm{W}\bm{y}_{t-1}+\bbm{\varepsilon}_t.
    \end{equation}
    In other words, it is a VAR(1) model with $\bm{A}=\beta_1\bm{I}_p+\beta_2\bm{W}$, and the associated coefficient vector can be written as
    $\textup{vec}(\bm{A}^\top)=\bm{C}\bbm{\theta}$, where $\bm{C}=(p^{-1/2}\textup{vec}(\bm{I}_p),\|\bm{W}\|_\textup{F}^{-1}\textup{vec}(\bm{W}^\top))$, and $\bbm{\theta}=(p^{1/2}\beta_1,\|\bm{W}\|_{\textup{F}}\beta_2)^\top$. In this parameterization, $\bm{C}$ is an orthonormal matrix.
\end{example}

For a general zero-mean and covariance stationary time series $\{\bm{y}_t\}_{t=1}^T$, if we model the data by the linear-restricted VAR($d$) model in \eqref{eq:linearVAR} and apply the linear-restricted Yule--Walker estimator in \eqref{eq:YuleWalker2} with the robust estimators $\bm{\widetilde{\Omega}}$ and $\bbm{\widetilde{\omega}}$, we have the upper bounds for the estimation error between the estimated parameter vector $\bbm{\widehat{\theta}}$ and the true vector $\bbm{\theta}^*$ defined in \eqref{eq:thetaStar}.

\begin{proposition} 
    \label{prop:3}
    For the constrained Yule--Walker estimator in \eqref{eq:YuleWalker2}, suppose that $\|\bm{\widetilde{\Omega}}-\bm{\Omega}\|_{1,\infty}\leq\zeta_1$,
    $\|\bbm{\widetilde{\omega}}-\bbm{\omega}\|_\infty\leq\zeta_2$, and $\lambda\geq \zeta_1\|\bbm{\theta}^*\|_\infty+\zeta_2$. Then,
    \begin{equation}
        \|\bbm{\widehat{\theta}}-\bbm{\theta}^*\|_\infty\leq 2\lambda\|\bm{\Omega}^{-1}\|_{1,\infty}.
    \end{equation}
\end{proposition}

This proposition presents the $\ell_{\infty}$ norm estimation error bounds of the low-dimensional parameter $\bbm{\widehat{\theta}}$, which is a new result for the linear-restricted VAR models. Similarly to Proposition \ref{prop:1}, the upper bound is proportional to $\|\bm{\Omega}^{-1}\|_{1,\infty}$ that could be related to the dimension $p$, depending on the specific linear restrictions. The $\ell_\infty$ upper bound directly implies that $\|\widehat{\bbm{\theta}}-\bbm{\theta}^*\|_2\leq2\sqrt{r}\lambda\|\bm{\Omega}^{-1}\|_{1,\infty}$, and some sharper bounds may be obtained based on the the special structure of $\bm{C}$. More concrete results and discussions are provided in Section \ref{sec:3.3}.

\begin{remark}
  For the upper bound result in Proposition \ref{prop:3}, the terms $\|\bbm{\theta}^*\|_\infty$ and $\|\bm{\Omega}^{-1}\|_{1,\infty}$ may vary or even diverge with the dimension $p$. Indeed, the term $\|\bbm{\theta}^*\|_\infty$ for the network VAR model in Example \ref{ex:banded} naturally diverges at a rate $\sqrt{p}$. The explicit rate of $\lambda\|\bm{\Omega}^{-1}\|_{1,\infty}$ depends on the specific linear restriction structure and further assumptions. More discussions about these terms are given in Section \ref{sec:3.3} for some specific models.
\end{remark}

The deterministic upper bounds in Propositions \ref{prop:1}-\ref{prop:3} imply that the performance of parameter estimation of the constrained Yule--Walker estimators hinges on the accuracy of autocovariance matrix estimation. It suffices to find reliable robust autocovariance estimators, which will be introduced in the next section.

\section{Robust Autocovariance Estimation}\label{sec:3}

\subsection{Element truncation estimator}\label{sec:3.1}

\citet{fan2016shrinkage} and \citet{ke2019user} proposed a simple robust estimation approach via appropriate truncation on data and showed that the truncated covariance estimator can achieve the optimal rate as that under the sub-Gaussian distribution for independent samples. We first apply the truncation method to estimate the elements of autocovariance matrices.

We consider the element-wise truncated data $\bm{y}_t^\text{E}(\tau)=(y_{1t}(\tau),y_{2t}(\tau),\dots,y_{pt}(\tau))^\top$, where $y_{it}(\tau)=\text{sign}(y_{it})(\tau\wedge|y_{it}|)$, for $1\leq i\leq p$, and $\tau>0$ is the truncation parameter. Based on this truncation scheme, $\bm{\Gamma}_\ell$ can be estimated by $\bm{\widetilde{\Gamma}}_{\ell}^\text{E}(\tau)=T^{-1}\sum_{t=1}^T\bm{y}_t^\text{E}(\tau)\bm{y}_{t-\ell}^\text{E}(\tau)^\top$, for any integer $\ell\geq0$. The corresponding autocovariance estimators are
\begin{equation}
    \widetilde{\bm{\Sigma}}_1^{\text{E}}(\tau)=[\widetilde{\bm{\Gamma}}^{\text{E}}_1(\tau),\widetilde{\bm{\Gamma}}^{\text{E}}_2(\tau),\dots,\widetilde{\bm{\Gamma}}^{\text{E}}_d(\tau)],
\end{equation}and
\begin{equation}
    \widetilde{\bm{\Sigma}}_0^{\text{E}}(\tau) = 
    \begin{bmatrix}
        \widetilde{\bm{\Gamma}}^{\text{E}}_0(\tau) & \widetilde{\bm{\Gamma}}^{\text{E}}_1(\tau) & \cdots & \widetilde{\bm{\Gamma}}^{\text{E}}_{d-1}(\tau)\\
        \widetilde{\bm{\Gamma}}^{\text{E}}_{-1}(\tau) & \widetilde{\bm{\Gamma}}^{\text{E}}_0(\tau) & \cdots & \widetilde{\bm{\Gamma}}^{\text{E}}_{d-2}(\tau)\\
        \vdots & \vdots & \ddots & \vdots\\
        \widetilde{\bm{\Gamma}}^{\text{E}}_{-d+1}(\tau) & \widetilde{\bm{\Gamma}}^{\text{E}}_{-d}(\tau) & \cdots & \widetilde{\bm{\Gamma}}^{\text{E}}_0(\tau)
    \end{bmatrix},
\end{equation}
where $\widetilde{\bm{\Gamma}}^{\text{E}}_{-j}(\tau)=\widetilde{\bm{\Gamma}}^{\text{E}}_j(\tau)^\top$, for $j\geq1$. The element-wise truncation can control the deviation of $y_{it}y_{jt}$ or $y_{it}y_{j,t-\ell}$ and the truncation parameter $\tau$ allows us to balance the tradeoff between the truncation bias and robustness.

Note that no data generating mechanism assumption is imposed, and we adopt the $\alpha$-mixing condition to quantify the serial dependency. For any stochastic process $\{\bm{y}_t\}_{t=-\infty}^\infty$, the lag-$\ell$ $\alpha$-mixing dependence coefficient is defined as
$\alpha(\ell):=\sup_{s\in\mathbb{Z}}\alpha(\{\bm{y}_t\}_{t=-\infty}^s,\{\bm{y}_t\}_{t=s+\ell}^\infty)$, where
\begin{equation}
  \alpha(\{\bm{y}_t\}_{t=-\infty}^s,\{\bm{y}_t\}_{t=r}^\infty)=\sup|\mathbb{P}(A\cap B)-\mathbb{P}(A)\mathbb{P}(B)|,
\end{equation}
the supremum being taken over all events $A\in\sigma(\{\bm{y}_t\}_{t=-\infty}^s)$ and $B\in\sigma(\{\bm{y}_t\}_{t=r}^\infty)$, and $\sigma(\cdot)$ is the sigma field generated by the process. To derive the theoretical guarantees of the autocovariance estimation, we have the following assumptions.

\begin{assumption}\label{asmp:mixing}
    The process $\{\bm{y}_t\}$ is weakly stationary and $\alpha$-mixing with the mixing coefficients $\alpha(\ell)=O(r^\ell)$, where $r=r(p)$ is a sequence possibly depending on $p$ such that $0\leq r\leq \bar{r}$ for some constant $\bar{r}<1$.
\end{assumption}

\begin{assumption}\label{asmp:moment}
    For $1\leq i\leq p$, $\mathbb{E}[|y_{it}|^{2+2\epsilon}]\leq M_{2+2\epsilon}$, for some $\epsilon\in(0,1]$.
\end{assumption}

Assumption \ref{asmp:mixing} states the weak stationarity and geometrically decayed $\alpha$-mixing, rather than assuming that the true data generating process (DGP) of $\bm{y}_t$ is a VAR($d$) process. If $\bm{y}_t$ truly follows a VAR($d$) model in \eqref{eq:VAR}, Assumption \ref{asmp:mixing} holds by the stationarity and geometric ergodicity of the VAR process; see Proposition 2 in \citet{liebscher2005towards}. Assumption \ref{asmp:moment} relaxes the commonly-used sub-Gaussian condition in the existing literature to the bounded $(2+2\epsilon)$-th moment condition, where $M_{2+2\epsilon}$ may diverge to infinity with the dimension $p$. For the strong mixing time series, we denote $n_\text{eff}=T/\log(T)^2$ as the effective sample size representing the number of effectively independent samples from $T$ observations. Note that $n_\text{eff}$ is of almost the same rate as $T$ since $T^{1-\delta}=o(n_\text{eff})$ for any $\delta>0$.

\begin{proposition}\label{prop:element}
  Under Assumptions \ref{asmp:mixing} and \ref{asmp:moment}, if $T\gtrsim\log(p^2d)$ and
  \begin{equation}
    \tau\asymp\left[\frac{M_{2+2\epsilon}n_\textup{eff}}{\log(p^2d)}\right]^{\frac{1}{2+2\epsilon}},
  \end{equation}
    then, with probability at least $1-C\exp[-C\log(T)\log(p^2d)]$,
    \begin{equation}
        \begin{split}
            &\|\widetilde{\bm{\Sigma}}_0^{\textup{E}}(\tau)-\bm{\Sigma}_0\|_{\infty}\lesssim\left[\frac{M_{2+2\epsilon}^{1/\epsilon}\log(p^2d)}{n_\textup{eff}}\right]^{\frac{\epsilon}{1+\epsilon}},\\
            \text{and~~}&\|\widetilde{\bm{\Sigma}}_1^{\textup{E}}(\tau)-\bm{\Sigma}_1\|_{\infty}\lesssim \left[\frac{M_{2+2\epsilon}^{1/\epsilon}\log(p^2d)}{n_\textup{eff}}\right]^{\frac{\epsilon}{1+\epsilon}}.
        \end{split}
    \end{equation}
\end{proposition}

This proposition presents the $\ell_\infty$ bounds for the element-wise truncation autocovariance estimators $\bm{\widetilde{\Sigma}}^\text{E}_0(\tau)$ and $\bm{\widetilde{\Sigma}}^\text{E}_1(\tau)$. If the truncation parameter $\tau$ is chosen appropriately and the moment bound $M_{2+2\epsilon}$ is fixed, the convergence rate of autocovariance matrix estimation in the $\ell_\infty$ norm is $O([\log(p^2d)/n_\text{eff}]^{\epsilon/(1+\epsilon)})$. 
When $\epsilon=1$, the data has a bounded fourth moment and the convergence rates of the autocovariance matrices scale as $\sqrt{\log(p^2d)/n_\text{eff}}$. When $\epsilon\in(0,1)$, the convergence rates have a smooth phase transition phenomenon, decreasing from $[\log(p^2d)/n_\text{eff}]^{1/2}$ to $[\log(p^2d)/n_\text{eff}]^{\epsilon/(1+\epsilon)}$. 

\begin{remark}
  For robust covariance estimators based on data truncation or shrinkage, the existing theoretical analysis  in \citet{fan2016shrinkage} and \citet{ke2019user} focused on the data with a bounded fourth moment. Our results in Proposition \ref{prop:element} relax the fourth moment condition to the $(2+2\epsilon)$-th moment condition and can effectively handle a much larger class of distributions. In addition, \citet{avella2018robust} proposed rank-based and adaptive Huber regression methods for robust covariance estimation and obtained a similar phase transition in the upper bound with respect to $\epsilon$.
\end{remark}

For the weakly sparse VAR model in Example \ref{ex:sparse}, $\bm{\widetilde{\Sigma}}_0^\text{E}(\tau)$ and $\bm{\widetilde{\Sigma}}_1^\text{E}(\tau)$ can be used as robust autocovariance estimators in the constrained Yule--Walker estimation in \eqref{eq:split_YW} with $\mathcal{R}(\cdot)=\|\cdot\|_1$ and $\mathcal{R}^*(\cdot)=\|\cdot\|_\infty$.
The resulting estimator for the $i$-th row of $\bm{A}$ is denoted as $\widehat{\bm{a}}_i(\lambda,\tau)$ and the sparse estimator is $\widehat{\bm{A}}_\text{S}(\lambda,\tau)=[\widehat{\bm{a}}_1(\lambda,\tau),\dots,\widehat{\bm{a}}_p(\lambda,\tau)]^\top$.

Define an $\ell_q$-``ball'' with radius $s_q$ as $\mathbb{B}_q(s_q)=\{\bm{A}\in\mathbb{R}^{p\times pd}:\max_{1\leq i\leq p}\sum_{j=1}^{pd}|\bm{A}_{ij}|^q\leq s_q\}$. Note that when $q=0$, $\mathbb{B}_0(s_0)$ is the set of all $p$-by-$pd$ matrices whose rows are at most $s_0$-sparse. For $q>0$, $\mathbb{B}_q(s_q)$ requires that the absolute values of entries in $\bm{A}$ decay sufficiently fast, which is more general than the exact sparsity assumption. Given the weak sparsity in the rows of $\bm{A}^*$, it is natural to assume that $\|\bm{A}^{*\top}\|_{1,\infty}$ is bounded.

We are ready to present the theoretical properties of $\widehat{\bm{A}}_\text{S}(\lambda,\tau)$.

\begin{theorem}[Sparse VAR upper bounds]
    \label{thm:sparseAR}
    Suppose that $\bm{A}^*\in\mathbb{B}_q(s_q)$ for some $q\in[0,1)$, $\|\bm{A}^{*\top}\|_{1,\infty}\leq C<\infty$, and Assumptions \ref{asmp:mixing} and \ref{asmp:moment} hold. If $T\gtrsim\log(p^2d)$, 
  \begin{equation}
    \tau\asymp\left[\frac{M_{2+2\epsilon}n_\textup{eff}}{\log(p^2d)}\right]^{\frac{1}{2+2\epsilon}}\quad\text{and}\quad\lambda\asymp\left[\frac{M_{2+2\epsilon}^{1/\epsilon}\log(p^2d)}{n_\textup{eff}}\right]^{\frac{\epsilon}{1+\epsilon}},
  \end{equation}
  then, with probability at least $1-C\exp[-C\log(T)\log(p^2d)]$,
    \begin{equation}
        \begin{split}
            \|\bm{\widehat{A}}_\textup{S}(\lambda,\tau)-\bm{A}^*\|_{\infty}&\lesssim \|\bm{\Sigma}_0^{-1}\|_{1,\infty}\left[\frac{M_{2+2\epsilon}^{1/\epsilon}\log(p^2d)}{n_\textup{eff}}\right]^{\frac{\epsilon}{1+\epsilon}},\\
            \|\bm{\widehat{A}}^\top_\textup{S}(\lambda,\tau)-\bm{A}^{*\top}\|_{1,\infty}&\lesssim s_q\|\bm{\Sigma}_0^{-1}\|_{1,\infty}^{1-q} \left[\frac{M_{2+2\epsilon}^{1/\epsilon}\log(p^2d)}{n_\textup{eff}}\right]^{\frac{(1-q)\epsilon}{1+\epsilon}},\\
            \text{and}~~\|\bm{\widehat{A}}^\top_\textup{S}(\lambda,\tau)-\bm{A}^{*\top}\|_{2,\infty}&\lesssim \sqrt{s_q}\|\bm{\Sigma}_0^{-1}\|_{1,\infty}^{1-\frac{q}{2}} \left[\frac{M_{2+2\epsilon}^{1/\epsilon}\log(p^2d)}{n_\textup{eff}}\right]^{\left(1-\frac{q}{2}\right)\frac{\epsilon}{1+\epsilon}}.
        \end{split}
    \end{equation}
\end{theorem}

This theorem presents non-asymptotic estimation upper bounds in various matrix norms under the weakly sparse structure. If both $M_{2+2\epsilon}$ and $\|\bm{\Sigma}_0^{-1}\|_{1,\infty}$ are fixed, the $\|\cdot\|_{\infty}$, $\|\cdot\|_{1,\infty}$, and $\|\cdot\|_{2,\infty}$ upper bounds scale as $[\log(p^2d)/n_\text{eff}]^{\epsilon/(1+\epsilon)}$,  $s_q[\log(p^2d)/n_\text{eff}]^{(1-q)\epsilon/(1+\epsilon)}$, and $\sqrt{s_q}[\log(p^2d)/n_\text{eff}]^{(1-q/2)\epsilon/(1+\epsilon)}$, respectively.
If the data have a bounded fourth moment with $\epsilon=1$, the first two rates of convergence match those of Gaussian sparse VAR in \citet{han2015direct}. When the true model has a strict sparsity structure with $q=0$ and sparsity level $s_0$, based on the $\|\cdot\|_{2,\infty}$ bound, the sample size requirement is $n_\text{eff}\gtrsim s_0\log(p^2d)$. That is, if the sparsity level in each row of $\bm{A}^*$ is finite, the dimension $p$ is allowed to be exponentially large compared to $n_\text{eff}$. If the data do not have a bounded fourth moment but only a bounded $(2+2\epsilon)$-th moment for some $\epsilon\in(0,1)$, the proposed estimator is still consistent but the rate of convergence in the $\ell_{2,\infty}$ norm decreases from $\sqrt{s_q}[\log(p^2d)/n_\text{eff}]^{1/2}$ to $\sqrt{s_q}[\log(p^2d)/n_\text{eff}]^{\epsilon/(1+\epsilon)}$.

\begin{remark}\label{rmk:3.2}
  The boundedness of $\|\bm{\Sigma}_0^{-1}\|_{1,\infty}$ is widely assumed in the existing literature of the Dantzig selector for sparse VAR models; see, for example, \citet{han2015direct} and \citet{wu2016performance}. Indeed, if $\bm{y}_t$ follows a VAR(1) process in \eqref{eq:VAR}, the covariance of $\bm{y}_t$ has an explicit form $\bm{\Sigma}_0=\sum_{i=0}^\infty(\bm{A}^*)^i\bm{\Sigma}_{\bbm{\varepsilon}}(\bm{A}^{*\top})^i$, where $\bm{\Sigma}_{\bbm{\varepsilon}}$ is the covariance matrix of $\bbm{\varepsilon}_t$, and $\|\bm{\Sigma}_0^{-1}\|_{1,\infty}$ can be verified to be independent of $p$ given that $\bm{A}^*$ and $\bm{\Sigma}_{\bbm{\varepsilon}}$ satisfy certain structures. For example, if both $\bm{A}^*$ and $\bm{\Sigma}_{\bbm{\varepsilon}}$ are block diagonal with a fixed block size, $\bm{\Sigma}_0^{-1}$ is also block diagonal and its $\ell_{1,\infty}$ norm is independent of $p$. If $\bm{A}^*$ has a banded structure as in Example \ref{ex:banded}, $(\bm{\Sigma}_0^{-1})_{ij}$ shrinks to zero exponentially as $|i-j|$ increases to infinity, and hence the boundedness of $\|\bm{\Sigma}_0^{-1}\|_{1,\infty}$ can be expected.
\end{remark}

\begin{remark}
  The robust estimation of sparse linear regression has been investigated in \citet{fan2017estimation}, \citet{sun2020adaptive}, and \citet{fan2016shrinkage}, among many others. We highlight some main differences between conventional linear regression and vector autoregression in terms of robust estimation. First, for linear regression models, the heavy-tailed distribution condition is considered on the random errors; that is, the response is conditionally heavy-tailed given the predictors. In our analysis, as no data generating mechanism assumption is imposed, the heavy-tailed distributional assumption, i.e. the moment condition in Assumption \ref{asmp:moment}, is directly considered on the observed data $\bm{y}_t$. Second, when the design matrix of the linear regression model satisfies some regulatory conditions, the $(1+\epsilon)$-th moment condition is considered for the random errors in \citet{sun2020adaptive}. However, for the VAR models, as both predictor and response are  simultaneously heavy-tailed, a more stringent $(2+2\epsilon)$-th moment condition has to be imposed in our theoretical analysis. Though the distributional assumptions are different in these two problems, the similar phase transition phenomenon can be obtained for two regimes $\epsilon\in(0,1)$ and $\epsilon\geq1$.
\end{remark}

\subsection{Vector truncation estimator}\label{sec:3.2}

In this subsection, we propose and study an autcovariance estimator that is robust in the operator norm. Intuitively, in order to achieve the operator norm robustness, we need to control the spectrum of $\bm{x}_t\bm{x}_t^\top$ and $\bm{y}_t\bm{x}_t^\top$. Note that $\|\bm{x}_t\bm{x}_t^\top\|_\textup{op}=\|\bm{x}_t\|_2^2$ and $\|\bm{y}_t\bm{x}_t^\top\|_\textup{op}=\|\bm{y}_t\|_2\|\bm{x}_t\|_2$. Therefore, we propose the truncation method to the whole vector $\bm{y}_t$ in the $\ell_2$ norm. 
To construct robust estimators for $\bm{\Sigma}_0$ and $\bm{\Sigma}_1$ of a VAR($d$) model, we consider the vector-truncated responses and predictors $\bm{y}^\text{V}_t(\tau_1)=(\tau_1\wedge\|\bm{y}_t\|_2)\bm{y}_t/\|\bm{y}_t\|_2$ and $\bm{x}^\text{V}_t(\tau_2)=(\tau_2\wedge\|\bm{x}_t\|_2)\bm{x}_t/\|\bm{x}_t\|_2$, for $1\leq t\leq T$, where two truncation parameters $\tau_1$ and $\tau_2$ are adopted as $\bm{y}_t$ and $\bm{x}_t$ are of different dimension when $d>1$. For the case with $d=1$, we can use a single parameter by setting $\tau_1=\tau_2$.
Based on the vector truncation of the data, the corresponding truncation autocovariance estimators are defined as
\begin{equation}\label{eq:vectorAutoCov}
    \widetilde{\bm{\Sigma}}^{\text{V}}_0(\tau_2)=\frac{1}{T}\sum_{t=1}^T\bm{x}_t^\text{V}(\tau_2)\bm{x}_t^\text{V}(\tau_2)^\top,~~\text{and}~~\widetilde{\bm{\Sigma}}^{\text{V}}_1(\tau_1,\tau_2)=\frac{1}{T}\sum_{t=1}^T\bm{y}_t^\text{V}(\tau_1)\bm{x}_t^\text{V}(\tau_2)^\top.
\end{equation}

\begin{remark}
  \citet{fan2016shrinkage} proposed a robust covariance estimator based on the vector-wise truncation in the $\ell_4$ norm and investigated its theoretical properties under the finite fourth moment condition. However, as we are interested in the relaxed setting with the finite $(2+2\epsilon)$-th moment, the truncation in the $\ell_4$ norm is not appropriate in our setting. \citet{ke2019user} proposed the spectrum-wise truncation for robust estimation of covariance matrix, and our method shares the same idea.
\end{remark}

For the vector-wise truncation estimator, we adopt the $\beta$-mixing condition to quantify the serial dependency. Specifically, for any stochastic process $\{\bm{y}_t\}_{t=-\infty}^\infty$, the lag-$\ell$ $\beta$-mixing dependence coefficient is defined as
$\beta(\ell):=\sup_{s\in\mathbb{Z}}\beta(\{\bm{y}_t\}_{t=-\infty}^s,\{\bm{y}_t\}_{t=s+\ell}^\infty)$, where
\begin{equation}
  \beta(\{\bm{y}_t\}_{t=-\infty}^s,\{\bm{y}_t\}_{t=r}^\infty)=\sup\frac{1}{2}\sum_{i=1}^I\sum_{j=1}^J\left|\mathbb{P}(A_i\cap B_j)-\mathbb{P}(A_i)\mathbb{P}(B_j)\right|,
\end{equation}
the supremum being taken over all pairs of partitions $\{A_1,\dots,A_I\}$ and $\{B_1,\dots,B_J\}$ such that $A_i\in\sigma(\{\bm{y}_t\}_{t=-\infty}^s)$ and $B_j\in\sigma(\{\bm{y}_t\}_{t=r}^\infty)$
. To study the theoretical properties of the vector-wise truncation estimators, we need the following assumptions.

\begin{assumption}
    \label{asmp:mixing2}
    The process $\{\bm{y}_t\}$ is weakly stationary and $\beta$-mixing with the mixing coefficients $\beta(\ell)=O(r^{\ell})$, where $r=r(p)$ is a sequence possibly depending on $p$ such that $0\leq r\leq \bar{r}$ for some constant $\bar{r}<1$.
\end{assumption}

\begin{assumption}
    \label{asmp:moment2}
    For any $\bm{v}\in\mathbb{R}^{pd}$ such that $\|\bm{v}\|_2=1$ and some $\delta>0$, $\mathbb{E}[(\bm{x}_t^\top\bm{v})^{2+2\epsilon+\delta}]\leq M_{2+2\epsilon}$.
\end{assumption}

Note that the $\beta$-mixing condition in Assumption \ref{asmp:mixing2} is slightly stronger than the $\alpha$-mixing condition in Assumption \ref{asmp:mixing} for the element-wise truncation estimators. Here we adopt the $\beta$-mixing condition in order to apply the Bernstein-type inequality for $\beta$-mixing random matrices developed by \citet{banna2016bernstein}. For the VAR process with independent and identically distributed $\bbm{\varepsilon}_t$, the absolute regularity with geometrically decayed $\beta$-mixing coefficients is equivalent to the geometric ergodicity \citep{liebscher2005towards}. The bounded $(2+2\epsilon+\delta)$-th moment condition in Assumption \ref{asmp:moment2} is also slightly stronger than the bounded $(2+2\epsilon)$-th moment condition in Assumption \ref{asmp:moment}. Here $\delta$ is only for the technical purpose and can be arbitrarily small.

\begin{proposition}
    \label{prop:vector}
    Under Assumptions \ref{asmp:mixing2} and \ref{asmp:moment2}, if $T\gtrsim pd$,
  \begin{equation}
    \tau_1\asymp(p^\epsilon M_{2+2\epsilon}n_\textup{eff})^{\frac{1}{2+2\epsilon}}~~\text{and}~~\tau_2\asymp(p^\epsilon d^\epsilon M_{2+2\epsilon}n_\textup{eff})^{\frac{1}{2+2\epsilon}},
  \end{equation}
    then, with probability at least $1-C\exp[-C\log(T)]$,
    \begin{equation}
        \begin{split}
            \|\bm{\widetilde{\Sigma}}_0^\textup{V}(\tau_2)-\bm{\Sigma}_0\|_{\textup{op}}&\lesssim \left(\frac{pM_{2+2\epsilon}^{1/\epsilon}}{n_\textup{eff}}\right)^{\frac{\epsilon}{1+\epsilon}},\\
            \text{and}~~\|\bm{\widetilde{\Sigma}}_1^\textup{V}(\tau_1,\tau_2)-\bm{\Sigma}_1\|_{\textup{op}}&\lesssim \left(\frac{pM_{2+2\epsilon}^{1/\epsilon}}{n_\textup{eff}}\right)^{\frac{\epsilon}{1+\epsilon}}.
        \end{split}
    \end{equation}
\end{proposition}

This proposition presents the operator norm bounds for the vector truncation autocovariance estimators $\bm{\widetilde{\Sigma}}_0^\text{V}(\tau_2)$ and $\bm{\widetilde{\Sigma}}_1^\text{V}(\tau_1,\tau_2)$. 
Similarly to the element truncation estimator in Section \ref{sec:3.1}, we define $n_\text{eff}=T/\log(T)^2$. If the moment bound $M_{2+2\epsilon}$ is fixed, the operator norm convergence rates of both autocovariance matrix estimators scale as $(pd/n_\text{eff})^{\epsilon/(1+\epsilon)}$. When the time series data has a bounded fourth moment, that is $\epsilon=1$, our results are almost the same as those of covariance matrix estimators for \textit{i.i.d.} data in \citet{ke2019user} and \citet{fan2016shrinkage}. Smooth transition on the operator norm convergence rate is observed when the fourth moment condition is relaxed to the $(2+2\epsilon)$-th moment condition.

Based on the rates of the vector truncation autocovariance matrices, we derive the estimation rate of the nuclear norm constrained Yule--Walker estimator for the weakly low-rank VAR model. Define an $\ell_q$-``ball'' for the singular values with radius $r_q$ as $\widetilde{\mathbb{B}}_q(r_q)=\{\bm{A}\in\mathbb{R}^{p\times p}:\sum_{i=1}^p\sigma_i^q(\bm{A})\leq r_q\}$. When the singular values of $\bm{A}^*$ are weakly sparse, it is natural to assume that $\|\bm{A}^*\|_\textup{nuc}$, the sum of singular values of $\bm{A}^*$, is bounded. Denote by $\bm{\widehat{A}}_\text{RR}(\lambda,\tau_1,\tau_2)$ the reduced-rank constrained Yule--Walker estimator with the regularizer $\mathcal{R}(\cdot)=\|\cdot\|_\textup{nuc}$, the constraint function $\mathcal{R}^*(\cdot)=\|\cdot\|_\textup{op}$, the constraint parameter $\lambda$, and the vector truncation autocovariance estimators in \eqref{eq:vectorAutoCov} with truncation parameters $\tau_1$ and $\tau_2$.

\begin{theorem}[Reduced-rank VAR upper bounds]
    \label{thm:lowrankAR}
    Suppose that $\bm{A}^*\in\widetilde{\mathbb{B}}_q(r_q)$ for some $q\in[0,1)$, $\|\bm{A}^*\|_\textup{nuc}\leq C<\infty$, and Assumptions \ref{asmp:mixing2} and \ref{asmp:moment2} hold. If $T\gtrsim p$, 
  \begin{equation}
    \tau_1\asymp(p^\epsilon M_{2+2\epsilon}n_\textup{eff})^{\frac{1}{2+2\epsilon}},~~\tau_2\asymp(p^\epsilon d^\epsilon M_{2+2\epsilon}n_\textup{eff})^{\frac{1}{2+2\epsilon}},~~\text{and}~~\lambda\asymp\left(\frac{pM_{2+2\epsilon}^{1/\epsilon}}{n_\textup{eff}}\right)^{\frac{\epsilon}{1+\epsilon}},
  \end{equation}
  then, with probability at least $1-C\exp[-C\log(T)]$,
    \begin{equation}
        \begin{split}
            \|\bm{\widehat{A}}_\textup{RR}(\lambda,\tau_1,\tau_2)-\bm{A}^*\|_\textup{op}&\lesssim\left(\frac{pM_{2+2\epsilon}^{1/\epsilon}}{n_\textup{eff}}\right)^{\frac{\epsilon}{1+\epsilon}},\\
            \|\bm{\widehat{A}}_\textup{RR}(\lambda,\tau_1,\tau_2)-\bm{A}^*\|_\textup{nuc}&\lesssim r_q\|\bm{\Sigma}_0^{-1}\|^{1-q}_\textup{op}\left(\frac{pM_{2+2\epsilon}^{1/\epsilon}}{n_\textup{eff}}\right)^{\frac{(1-q)\epsilon}{1+\epsilon}},\\
            \text{and}~~\|\bm{\widehat{A}}_\textup{RR}(\lambda,\tau_1,\tau_2)-\bm{A}^*\|_\textup{F}&\lesssim \sqrt{r_q}\|\bm{\Sigma}_0^{-1}\|^{1-\frac{q}{2}}_\textup{op}\left(\frac{pM_{2+2\epsilon}^{1/\epsilon}}{n_\textup{eff}}\right)^{\left(1-\frac{q}{2}\right)\frac{\epsilon}{1+\epsilon}}.
        \end{split}
    \end{equation}
\end{theorem}
This theorem presents the non-asymptotic estimation upper bounds in the operator norm, nuclear norm, and Frobenius norm, respectively. If $M_{2+2\epsilon}$ is fixed, the convergence rates of $\widehat{\bm{A}}_\text{RR}$ in the operator norm and Frobenius norm scale as $(p/n_\text{eff})^{\epsilon/(1+\epsilon)}$ and $\sqrt{r_q}(p/n_\text{eff})^{(1-q/2)\epsilon/(1+\epsilon)}$. Specifically, when $\epsilon=1$, that is the time series have a bounded fourth moment, the Frobenius norm convergence rate of the robust estimator is nearly the same as those obtained by the standard nuclear norm penalized estimators for Gaussian VAR model in \citet{negahban2011estimation} and \citet{basu2019low}. When the time series only has a bounded $(2+2\epsilon)$-th moment for some $\epsilon\in(0,1)$, the estimation error rates decrease from $\sqrt{r_q}(p/n_\text{eff})^{1/2}$ to $(p/n_\text{eff})^{\epsilon/(1+\epsilon)}$. Note that though the estimation convergence rates decrease, if $r_q$ is fixed, the sample size requirement $n_\text{eff}\gtrsim pd$ for estimation consistency remains unchanged when the moment condition is relaxed.

\begin{remark}
  The boundedness of $\|\bm{\Sigma}_0^{-1}\|_\textup{op}$ is equivalent to that the smallest eigenvalue of $\bm{\Sigma}_0$ is bounded away from zero. If $\bm{y}_t$ follows a stationary VAR model in \eqref{eq:VAR}, this condition can be guaranteed if the smallest eigenvalue of $\bm{\Sigma}_{\bbm{\varepsilon}}$ is bounded away from zero, where $\bm{\Sigma}_{\bbm{\varepsilon}}$ is the covariance matrix of $\bbm{\varepsilon}_t$.
\end{remark}

\subsection{Linear-restricted truncation estimator}\label{sec:3.3}

For the linear-restricted VAR model in \eqref{eq:linearVAR}, the linear transformations of autocovariance matrices are defined as
\begin{equation}
  \bm{\Omega}=\mathbb{E}[\bm{C}^\top(\bm{I}_p\otimes\bm{x}_t)(\bm{I}_p\otimes\bm{x}_t^\top)\bm{C}]~~\text{and}~~\bbm{\omega}=\mathbb{E}[\bm{C}^\top(\bm{I}_p\otimes\bm{x}_t)\bm{y}_t].
\end{equation}
Motivated by the element and vector-wise truncation in the previous subsections, in order to robustly estimate each element of $\bm{\Omega}$, we apply the truncation to $(\bm{I}_p\otimes\bm{x}_t^\top)\bm{c}_i$, where each $\bm{c}_i\in\mathbb{R}^{p^2d}$ is the $i$-th column of $\bm{C}$. When each $\bm{c}_i$ is highly sparse, e.g., $\bm{c}_i$ is an coordinate vector for the banded VAR in Example \ref{ex:banded}, the vector $(\bm{I}_p\otimes\bm{x}_t^\top)\bm{c}_i$ is also highly sparse. In this case,let $S_i$ be the non-zero index set of $(\bm{I}_p\otimes\bm{x}_t)\bm{c}_i$, and  $\bm{w}_{it}=[(\bm{I}_p\otimes\bm{x}_t^\top)\bm{c}_i]_{S_i}$ and $\bm{z}_{it}=(\bm{y}_t)_{S_i}$ be the sub-vectors. It is obvious to check that $\bm{c}_i^\top(\bm{I}_p\otimes\bm{x}_t)(\bm{I}_p\otimes\bm{x}_t^\top)\bm{c}_j=\bm{w}_{it}^\top\bm{w}_{jt}$ and $\bm{c}_i^\top(\bm{I}_p\otimes\bm{x}_t)\bm{y}_t=\bm{w}_{it}^\top\bm{z}_{it}$, for $1\leq i,j\leq r$.

Therefore, we consider the vector-wise truncation $\widetilde{\bm{w}}_{it}(\tau_1)=(\tau_1\wedge\|\bm{w}_{it}\|_2)\bm{w}_{it}/\|\bm{w}_{it}\|_2$ and $\widetilde{\bm{z}}_{it}(\tau_2)=(\tau_2\wedge\|\bm{z}_{it}\|_2)\bm{z}_{it}/\|\bm{z}_{it}\|_2$, and the elements of $\bm{\Omega}$ and $\bbm{\omega}$ can be estimated by the linear-restricted truncated data
\begin{equation}
  \widetilde{\bm{\Omega}}_{ij}(\tau_1)=\frac{1}{T}\sum_{t=1}^T\widetilde{\bm{w}}_{it}(\tau_1)^\top\widetilde{\bm{w}}_{jt}(\tau_1)~~\text{and}~~\widetilde{\bbm{\omega}}_i(\tau_1,\tau_2)=\frac{1}{T}\sum_{t=1}^T\widetilde{\bm{w}}_{it}(\tau_1)^\top\widetilde{\bm{z}}_{it}(\tau_2),
\end{equation}
for $1\leq i,j\leq r$. 

\begin{remark}
  For the special matrix $\bm{C}$ whose columns are coordinate vectors corresponding to the non-zero entries in $\textup{vec}(\bm{A}^\top)$, the linear-restricted estimators with $\tau_1=\tau_2$ are equivalent to the linear transformations of element truncation estimators, namely $\widetilde{\bm{\Omega}}(\tau_1)=\bm{C}^\top(\bm{I}_p\otimes\widetilde{\bm{\Sigma}}_0^\text{E}(\tau_1))\bm{C}$ and $\widetilde{\bbm{\omega}}(\tau_1,\tau_1)=\bm{C}^\top\textup{vec}(\widetilde{\bm{\Sigma}}_1^\text{E}(\tau_1)^\top)$.
\end{remark}

For the linear-restricted models, we consider the following moment conditions.

\begin{assumption}\label{asmp:moment_linear}
  For $1\leq i\leq r$, $\mathbb{E}[\|\bm{w}_{it}\|_2^{2+2\epsilon}]\leq M_{1,2+2\epsilon}$ and $\mathbb{E}[\|\bm{z}_{it}\|_2^{2+2\epsilon}]\leq M_{2,2+2\epsilon}$.
\end{assumption}

The moment conditions in Assumption \ref{asmp:moment_linear} can be viewed as an extension of the element-wise moment condition in Assumption \ref{asmp:moment}. In addition, due to the normalization of the matrix $\bm{C}$, the moment bounds $M_{1,2+2\epsilon}$ and $M_{2,2+2\epsilon}$ are allowed to vary with the dimension $p$. More discussions on the moment bounds are presented below for each specific linear-restricted model. For the linear-restricted model in \eqref{eq:linearVAR} with $r>1$, the general results of the autocovariance estimators are given as follows.

\begin{proposition}
  \label{prop:linear}
  Under Assumptions \ref{asmp:mixing} and \ref{asmp:moment_linear}, if $T\gtrsim\log(r)$,
  \begin{equation}
    \tau_1\asymp\left[\frac{M_{1,2+2\epsilon}n_\text{eff}}{\log(r)}\right]^{\frac{1}{2+2\epsilon}},~~\text{and}~~\tau_2\asymp\left[\frac{M_{2,2+2\epsilon}n_\text{eff}}{\log(r)}\right]^{\frac{1}{2+2\epsilon}},
  \end{equation}
  then, with probability at least $1-C\exp[-C\log(r)\log(T)]$,
  \begin{equation}
    \begin{split}
      \|\widetilde{\bm{\Omega}}(\tau_1)-\bm{\Omega}\|_\infty & \lesssim\left[\frac{M_{1,2+2\epsilon}^{1/\epsilon}\log(r)}{n_\text{eff}}\right]^{\frac{\epsilon}{1+\epsilon}},\\
      \text{and}~~\|\widetilde{\bbm{\omega}}(\tau_1,\tau_2)-\bbm{\omega}\|_\infty & \lesssim\left[\frac{M_{1,2+2\epsilon}^{1/(2\epsilon)}M_{2,2+2\epsilon}^{1/(2\epsilon)}\log(r)}{n_\text{eff}}\right]^{\frac{\epsilon}{1+\epsilon}}.
    \end{split}
  \end{equation}
\end{proposition}

This proposition presents the estimation upper bounds for $\widetilde{\bm{\Omega}}(\tau_1)$ and $\widetilde{\bbm{\omega}}(\tau_2)$ in the $\ell_\infty$ norm. Following this proposition, we first consider the banded VAR model with the bandwidth $k_0$ in Example \ref{ex:banded}, denote by $\widehat{\bbm{\theta}}_\text{B}(\lambda,\tau)$ and $\widehat{\bm{A}}_\text{B}(\lambda,\tau)$ the linear-restricted constrained Yule--Walker estimators with robust autocovariance estimators $\widetilde{\bm{\Omega}}(\tau)$ and $\widetilde{\bbm{\omega}}(\tau,\tau)$. Note that as each $\bm{c}_i$ is a coordinate vector, both $\bm{w}_{it}$ and $\bm{z}_{it}$ are 1-dimensional, and the Assumption \ref{asmp:moment_linear} reduces to the element-wise moment bound in Assumption \ref{asmp:moment}. The following estimation upper bounds can be derived.

\begin{theorem}[Banded VAR upper bounds]\label{thm:banded}
  For the banded VAR model with a bounded bandwidth $k_0$ in Example \ref{ex:banded}, suppose that $\|\bbm{\theta}^*\|_\infty\leq C<\infty$, and Assumptions \ref{asmp:mixing} and \ref{asmp:moment} hold. If $T\gtrsim\log(p)$,
  \begin{equation}
    \tau\asymp\left[\frac{M_{2+2\epsilon}n_\textup{eff}}{\log(p)}\right]^{\frac{1}{1+\epsilon}},~~\text{and}~~\lambda\asymp\left[\frac{M_{2+2\epsilon}^{1/\epsilon}\log(p)}{n_\textup{eff}}\right]^{\frac{1}{1+\epsilon}},
  \end{equation}
  then, with probability at least $1-C\exp[-C\log(p)\log(T)]$,
  \begin{equation}
    \begin{split}
      \|\bbm{\widehat{\theta}}_\textup{B}(\lambda,\tau)-\bbm{\theta}^*\|_\infty &\lesssim \|\bm{\Omega}^{-1}\|_{1,\infty}\left[\frac{M_{2+2\epsilon}^{1/\epsilon}\log(p)}{n_\textup{eff}}\right]^{\frac{\epsilon}{1+\epsilon}},\\
      \|\bm{\widehat{A}}_\textup{B}(\lambda,\tau)-\bm{A}^*\|_\textup{op} &\lesssim \|\bm{\Omega}^{-1}\|_{1,\infty}\left[\frac{M_{2+2\epsilon}^{1/\epsilon}\log(p)}{n_\textup{eff}}\right]^{\frac{\epsilon}{1+\epsilon}},\\
      \text{and}~~\|\bbm{\widehat{\theta}}_\textup{B}(\lambda,\tau)-\bbm{\theta}^*\|_2= \|\bm{\widehat{A}}_\textup{B}(\lambda,\tau)-\bm{A}^*\|_\textup{F} &\lesssim
      \sqrt{p}\|\bm{\Omega}^{-1}\|_{1,\infty}\left[\frac{M_{2+2\epsilon}^{1/\epsilon}\log(p)}{n_\textup{eff}}\right]^{\frac{\epsilon}{1+\epsilon}}.
    \end{split}
  \end{equation}
\end{theorem}

When $M_{2+2\epsilon}$ and $\|\bm{\Omega}^{-1}\|_{1,\infty}$ are fixed, the upper bounds in this theorem imply that asymptotically the rates of convergence in the operator norm and Frobenius norm scale as $[\log(p)/n_\text{eff}]^{\epsilon/(1+\epsilon)}$ and $[p\log(p)/n_\text{eff}]^{\epsilon/(1+\epsilon)}$, respectively. In other words, the sample size requirement in the operator norm is $n_\text{eff}\gtrsim\log(p)$. When $\epsilon=1$, our results are comparable and consistent with those in \citet{guo2016high} up to a logarithm factor.
However, in \citet{guo2016high}, they imposed the \textit{i.i.d.} condition and sub-exponential tail condition for $\{\bbm{\varepsilon}_t\}$ when $p\gtrsim T$. Hence, under high-dimensional scaling, both the operator norm and Frobenius norm convergence rates of the robust constrained Yule--Walker estimator for heavy-tailed data with a bounded fourth moment are  almost the same as those of the standard ordinary least squares under the sub-exponential tail condition. In sum, the convergence rates obtained under the sub-exponential distribution can be achieved by the robust estimator under a much relaxed fourth moment condition, and we also establish the estimation consistency and the rates of convergence under the $(2+2\epsilon)$-th moment condition.

\begin{remark}
  The stationarity of the VAR(1) model requires that the largest eigenvalue of $\bm{A}^*$ in terms of absolute value is strictly smaller than one, and hence the absolute value of $\bbm{\theta}^*_i$ is expected to be bounded.
  As $\bm{\Omega}=\bm{C}^\top(\bm{I}_p\otimes\bm{\Sigma}_0)\bm{C}$ and the columns of $\bm{C}$ are all coordinate vectors, $\|\bm{\Omega}^{-1}\|_{1,\infty}$ is bounded by $\|\bm{\Sigma}_0^{-1}\|_{1,\infty}$, and the boundedness of $\|\bm{\Sigma}_0^{-1}\|_{1,\infty}$ has been discussed in Remark \ref{rmk:3.2}.
\end{remark}

Next, we consider another case of the linear-restricted VAR models, the network VAR model in Example \ref{ex:network}. In this model, $\bm{C}$ is a $p^2$-by-2 matrix with $\bm{c}_1=p^{-1/2}\textup{vec}(\bm{I}_p)$ and $\bm{c}_2=\|\bm{W}\|_\textup{F}^{-1}\textup{vec}(\bm{W}^\top)$ being the restrictions associated with the nodal effect and network effect, respectively. Accordingly, $\bbm{\theta}=(p^{1/2}\beta_1,\|\bm{W}\|_\textup{F}\beta_2)^\top$. As $\bm{C}$ needs to be normalized fro identification purpose, we assume that the unscaled parameters, $\beta_1$ and $\beta_2$ are constants independent of $p$. In addition, we assume that each node in the network is only connected to a fixed number of nodes, so the network adjacency matrix satisfies $\|\bm{W}\|_\textup{F}\asymp\sqrt{p}$. 

According to the linear restriction matrix $\bm{C}$ of the network VAR model, it can be verified that $\bm{w}_{1t}=p^{-1/2}\bm{y}_{t-1}$, $\bm{w}_{2t}=\|\bm{W}\|_\textup{F}^{-1}\bm{W}\bm{y}_{t-1}$, and $\bm{z}_{1t}=\bm{z}_{2t}=\bm{y}_t$. Hence, if $\max_{i=1,2}\mathbb{E}\|\bm{w}_{it}\|_2^{2+2\epsilon}=M_{2+2\epsilon}$ which is independent of $p$, then, $\mathbb{E}\|\bm{z}_{1t}\|_2^{2+2\epsilon}=\mathbb{E}\|\bm{z}_{2t}\|_2^{2+2\epsilon}\leq p^{1+\epsilon}M_{2+2\epsilon}$. For the network VAR model, denote by $\widehat{\bbm{\theta}}_\text{N}(\lambda,\tau_1,\tau_2)$ the constrained Yule--Walker estimator with the constraint parameter $\lambda$ and the autocovariance estimators $\widetilde{\bm{\Omega}}(\tau_1)$ and $\widetilde{\bbm{\omega}}(\tau_1,\tau_2)$. The estimation upper bounds of the network VAR model can be derived as follows.

\begin{theorem}[Network VAR upper bounds]
  \label{thm:network}
  For the network VAR model in Example \ref{ex:network} with $\|\bm{W}\|_\textup{F}\asymp\sqrt{p}$ and $\|\bbm{\theta}^*\|_\infty\asymp\sqrt{p}$, under Assumptions \ref{asmp:mixing} and \ref{asmp:moment_linear} with $M_{1,2+2\epsilon}=M_{2+2\epsilon}$ and $M_{2,2+2\epsilon}=p^{1+\epsilon}M_{2+2\epsilon}$, if
  \begin{equation}
    \tau_1\asymp(M_{2+2\epsilon}n_\textup{eff})^{\frac{1}{2+2\epsilon}},~~\tau_2\asymp\sqrt{p}(M_{2+2\epsilon}n_\textup{eff})^{\frac{1}{2+2\epsilon}},~~\text{and}~~\lambda\asymp\sqrt{p}\left(\frac{M_{2+2\epsilon}^{1/\epsilon}}{n_\textup{eff}}\right)^{\frac{\epsilon}{1+\epsilon}},
  \end{equation}
  then, with probability at least $1-C\exp[-C\log(T)]$,
  \begin{equation}
    \begin{split}
      \|\bbm{\widehat{\theta}}_\textup{N}(\lambda,\tau_1,\tau_2)-\bbm{\theta}^*\|_\infty&\lesssim\|\bm{\Omega}^{-1}\|_{1,\infty}\sqrt{p}\left(\frac{M_{2+2\epsilon}^{1/\epsilon}}{n_\textup{eff}}\right)^{\frac{\epsilon}{1+\epsilon}}, \\
      \text{and}~~\|\bm{\widehat{A}}_\textup{N}(\lambda,\tau_1,\tau_2)-\bm{A}^*\|_\textup{F}=\|\bbm{\widehat{\theta}}_\textup{N}(\lambda,\tau_1,\tau_2)-\bbm{\theta}^*\|_2&\lesssim \|\bm{\Omega}^{-1}\|_{1,\infty}\sqrt{p}\left(\frac{M_{2+2\epsilon}^{1/\epsilon}}{n_\textup{eff}}\right)^{\frac{\epsilon}{1+\epsilon}}.
    \end{split}
  \end{equation}
\end{theorem}
The upper bounds in this theorem are new for the network VAR model under the non-asymptotic scheme. 
Due to the normalization for the linear restriction matrix $\bm{C}$, both the true value $\|\bbm{\theta}^*\|_2$ and estimated value $\|\bbm{\widehat{\theta}}\|_2$ will diverge at a rate of $\sqrt{p}$ as the dimension $p$ increases to infinity. Hence, the $\ell_2$ rate of convergence for the normalized coefficients $p^{-1/2}\bbm{\theta}$ scales as $n_\text{eff}^{-\epsilon/(1+\epsilon)}$, independent of the dimension $p$, and our result with $\epsilon=1$ is comparable to that obtained under the sub-Gaussian condition in \citet{zheng2020finite}. The convergence rates under a bounded $(2+2\epsilon)$-th moment condition are established under a phase transition pattern.

\section{Minimax Lower Bounds}\label{sec:lower_bound}

In this section, we investigate theoretical properties on the lower bounds of estimation tasks for heavy-tailed high-dimensional time series data, including estimation of the high-dimensional VAR models and the autocovariance matrices. The lower bound analysis shows that the rates of convergence obtained in Section \ref{sec:3} are optimal in the minimax sense: there exists a distributional setting for the time series process for which the upper bounds obtained cannot be improved without further distributional assumptions.

\subsection{Lower bounds of VAR estimation}

Similarly to the upper bound analysis, no assumption on the data generating mechanism is imposed on the time series data in the lower bound analysis. We consider the VAR($d$) model and denote the distribution of $\{\bm{y}_t\}_{t=1}^T$ as $\mathbb{P}$. By the Yule--Walker equation, the true value of the VAR model coefficient matrix is defined as $\bm{A}^*(\mathbb{P})=\bm{\Sigma}_1(\mathbb{P})\bm{\Sigma}_0^{-1}(\mathbb{P})$, where $\bm{\Sigma}_0(\mathbb{P})=\mathbb{E}_{\mathbb{P}}[\bm{x}_t\bm{x}_t^\top]$ and $\bm{\Sigma}_1(\mathbb{P})=\mathbb{E}_{\mathbb{P}}[\bm{y}_t\bm{x}_t^\top]$.

The first case we consider is the sparse VAR model, where all row vectors of $\bm{A}^*$ are strictly sparse. For any $M>0$, $\epsilon\in(0,1]$ and $r\in(0,1)$, let $\mathcal{P}_\text{E}(M,\epsilon,r)$ denote the class of all joint distributions for the $\alpha$-mixing stochastic process $\bm{v}_t=(v_{1t},\dots,v_{pt})^\top$ with the element-wise moment condition, such that they satisfy $\max_{1\leq i\leq p}\mathbb{E}|v_{it}|^{2+2\epsilon}=M$ and $\alpha(\{\bm{v}_t\}_{t=-\infty}^{s},\{\bm{v}_t\}_{t=s+\ell}^\infty)=O(r^\ell)$ for any integer $s$ and lag order $\ell>0$. The moment parameter $M$ is imposed on the element of $\bm{v}_t$, which is consistent with Assumption \ref{asmp:moment}.

For $\bm{A}^{*}$ belonging to the sparse set $\mathbb{B}_0(s_0)$ defined in Section \ref{sec:3.1} and the process $\bm{y}_t$ following the distribution $\mathbb{P}\in\mathcal{P}_\text{E}(M,\epsilon,r)$, we have the minimax lower bound for the VAR estimation.

\begin{theorem}[Sparse VAR lower bound]
  \label{thm:LB_sparse_VAR}
  For any $M>0$, $\epsilon\in(0,1]$ and $r\in(0,1)$, suppose that $T/\log_{r^{-2}}T\geq2s_0$ and the joint distribution $\mathbb{P}$ of $\bm{y}_t$ belongs to $\mathcal{P}_\textup{E}(M,\epsilon,r)$. Then, for any estimator $\widehat{\bm{A}}:=\widehat{\bm{A}}(\{\bm{y}_t\}_{t=1}^T)$ which depends on the observations from $\bm{y}_1$ to $\bm{y}_T$,
  \begin{equation}
    \inf_{\widehat{\bm{A}}}\sup_{\substack{\mathbb{P}\in\mathcal{P}_\textup{E}(M,\epsilon,r),\\\bm{A}^*(\mathbb{P})\in\mathbb{B}_0(s_0)}}\mathbb{E}\left[\|\widehat{\bm{A}}^\top-\bm{A}^*(\mathbb{P})^\top\|_{2,\infty}\right]\gtrsim\sqrt{s_0}\|\bm{\Sigma}_0^{-1}(\mathbb{P})\|_{1,\infty}\left[\frac{M^{1/\epsilon}\log(T)}{T}\right]^{\frac{\epsilon}{1+\epsilon}}.
  \end{equation}
\end{theorem}

The minimax lower bound in this theorem matches the $\ell_{2,\infty}$ upper bound in Theorem \ref{thm:sparseAR} with $q=0$ up to a logarithm factor that is negligible compared with $T$. Hence, for the strictly sparse VAR model, the constrained Yule--Walker estimator $\widehat{\bm{A}}_\text{S}(\lambda,\tau)$ introduced in Section \ref{sec:3.1} is nearly rate-optimal when the regularization parameter $\lambda$ and truncation parameter $\tau$ are properly chosen. In addition, the effective sample size in the lower bound is $T/\log(T)$, where the $\log(T)$ factor can be referred to as the price paid for the serial  dependency in the data.

Moreover, it is clear that the banded VAR model with the bandwidth $k$ is a special case of the strictly sparse VAR model with the sparsity level $s_0=2k+1$. When $k$ is fixed, Theorem \ref{thm:LB_sparse_VAR} directly implies that the minimax lower bound of the banded VAR model is of rate $\|\bm{\Omega}^{-1}\|_{1,\infty}[M^{1/\epsilon}\log(T)/T]^{\epsilon/(1+\epsilon)}$ in the operator norm, matching the upper bound of the banded VAR in Theorem \ref{thm:banded}.

The second case considered is the reduced-rank VAR model where $\bm{A}^*$ is of low rank. For any $M>0$, $\epsilon\in(0,1]$ and $r\in(0,1)$, let $\mathcal{P}_\text{V}(M,\epsilon,r)$ denote the class of all joint distributions for the $\alpha$-mixing stochastic process $\bm{v}_t=(v_{1t},\dots,v_{pt})^\top$ with the vector-wise moment condition, such that $\sup_{\bm{u}\in\mathbb{R}^p,\|\bm{u}\|_2=1}\mathbb{E}|\bm{v}_t^\top\bm{u}|^{2+2\epsilon}=M$ and $\alpha(\{\bm{v}_t\}_{-\infty}^s,\{\bm{v}_t\}_{s+\ell}^\infty)=O(r^\ell)$ for any integer $s$ and lag order $\ell>0$. The moment condition defined on the whole vector is similar to Assumption \ref{asmp:mixing2} for the reduced-rank VAR model, but the technical term $\delta$ is omitted in the lower bound analysis.

For $\bm{A}^*$ belonging to the low-rank set $\widetilde{\mathbb{B}}_0(r_0)$ defined in Section \ref{sec:3.2} and the process $\bm{y}_t$ following the distribution $\mathbb{P}\in\mathcal{P}_\text{V}(M,\epsilon,r)$, the lower bound in terms of Frobenius norm is derived.

\begin{theorem}[Reduced-rank VAR lower bound]
  \label{thm:LB_lowrank_VAR}
  For any $M>0$, $\epsilon\in(0,1]$ and $r\in(0,1)$, suppose that $p\geq20$, $T>p/128$, and the joint distribution $\mathbb{P}$ of $\bm{y}_t$ belongs to $\mathcal{P}_\textup{V}(M,\epsilon,r)$. Then, for any estimator $\widehat{\bm{A}}:=\widehat{\bm{A}}(\{\bm{y}_t\}_{t=1}^T)$ which depends on the observations from $\bm{y}_1$ to $\bm{y}_T$,
  \begin{equation}
    \inf_{\widehat{\bm{A}}}\sup_{\substack{\mathbb{P}\in\mathcal{P}_\textup{V}(M,\epsilon,r),\\\bm{A}^*(\mathbb{P})\in\widetilde{\mathbb{B}}_0(r_0)}}\mathbb{E}\left[\|\widehat{\bm{A}}^\top-\bm{A}^*(\mathbb{P})^\top\|_\textup{F}\right]\gtrsim\sqrt{r_0}\|\bm{\Sigma}_0^{-1}(\mathbb{P})\|_\textup{op}\left(\frac{pM^{1/\epsilon}}{T}\right)^{\frac{\epsilon}{1+\epsilon}}.
  \end{equation}
\end{theorem}

This theorem presents the minimax lower bound of the reduced-rank VAR model with the exact low-rankness in terms of Frobenius norm, which matches the upper bound in Theorem \ref{thm:lowrankAR} up to a logarithm factor.

\begin{remark}
  The minimax lower bounds in Theorems \ref{thm:LB_sparse_VAR} and \ref{thm:LB_lowrank_VAR} are developed for the VAR models with an exact low-dimensional structure, such as the strict row-wise sparsity and exact low-rankness. The lower bounds for the high-dimensional regression and covariance estimation with the weakly sparse coefficients have been investigated by \citet{raskutti2011minimax}, \citet{cai2012optimal}, and others under the $\ell_q$ ball constraints. However, the theoretical techniques in their proofs rely heavily on the Gaussian distributional assumption, and hence cannot be applied to the heavy-tailed setting. The minimax lower bounds for the VAR estimation under the weak sparsity and approximate low-rankness are left for future research.
\end{remark}

\subsection{Lower bounds of autocovariance estimation}

For any process $\{\bm{y}_t\}$ following distribution $\mathbb{P}\in\mathcal{P}_\text{E}(M,\epsilon,r)$, let $\bm{x}_t=(\bm{y}_{t-1}^\top,\dots,\bm{y}_{t-d}^\top)^\top$, and a by-product of our lower bound analysis is the minimax lower bounds for the autocovariance estimation. First, we develop the minimax lower bounds in the $\ell_\infty$ norm for estimating the autocovariance matrices $\bm{\Sigma}_0(\mathbb{P})=\mathbb{E}_{\mathbb{P}}[\bm{x}_t\bm{x}_t^\top]$ and $\bm{\Sigma}_1(\mathbb{P})=\mathbb{E}_{\mathbb{P}}[\bm{y}_t\bm{x}_t^\top]$.

\begin{proposition}\label{prop:LB_maximum}
  For any $M>0$, $\epsilon\in(0,1]$ and $r\in[0,1)$, suppose that the joint distribution $\mathbb{P}$ of $\bm{y}_t$ belongs to $\mathcal{P}_\textup{E}(M,\epsilon,r)$. Then, for any covariance estimator $\widehat{\bm{\Sigma}}_k:=\widehat{\bm{\Sigma}}_k(\{\bm{y}_t\}_{t=1}^T)$ which depends on the observations from $\bm{y}_1$ to $\bm{y}_T$,
  \begin{equation}
    \inf_{\widehat{\bm{\Sigma}}_k}\sup_{\mathbb{P}\in\mathcal{P}_\textup{E}(M,\epsilon,r)}\mathbb{E}\left[\|\widehat{\bm{\Sigma}}_k-\bm{\Sigma}_k(\mathbb{P})\|_\infty\right]\gtrsim \left[\frac{M^{1/\epsilon}\log(T)}{T}\right]^{\frac{\epsilon}{1+\epsilon}},\quad k=0,1.
  \end{equation}
\end{proposition}

The minimax lower bound matches the upper bounds in Proposition \ref{prop:element} up to a logarithmic factor $[\log(p^2d)\log(T)]^{\epsilon/(1+\epsilon)}$, indicating the nearly minimax optimality of the proposed shrinkage estimator $\widetilde{\bm{\Sigma}}_k^\text{E}(\tau)$ with $\tau$ chosen properly. Compared with the upper and lower bounds of the covariance matrix estimation for the \textit{i.i.d.} data in \citet{avella2018robust} and \citet{devroye2016sub}, our minimax lower bound involves a factor $\log(T)^{\epsilon/(1+\epsilon)}$, which is due to the weak serial dependency under the geometrically decayed $\alpha$-mixing condition. It is also noteworthy that the $\log(T)$ factor can also be found in the upper bound analysis in \citet{zhang2021robust} under the functional dependence measures \citep{wu2005nonlinear}.

In addition, this minimax lower bound can easily be extended to the linear transformations of autocovariance matrices $\bm{\Omega}=\bm{C}^\top(\bm{I}_p\otimes\bm{\Sigma}_0)\bm{C}$ and $\bbm{\omega}=\bm{C}^\top\textup{vec}(\bm{\Sigma}_1)$, where $\bm{C}$ consists of coordinate vector columns. In other words, Proposition \ref{prop:LB_maximum} implies that the upper bounds in Proposition \ref{prop:linear} are rate-optimal up to a logarithm factor.

Moreover, the minimax lower bound in the operator norm is also established. 

\begin{proposition}\label{prop:LB_operator}
  For any $M>0$, $\epsilon\in(0,1]$ and $r\in[0,1)$, suppose that $p\geq20$, $T\geq p/128$, and the joint distribution $\mathbb{P}$ of $\bm{y}_t$ belongs to $\mathcal{P}_\textup{V}(M,\epsilon,r)$. Then, for any covariance estimator $\widehat{\bm{\Sigma}}_k:=\widehat{\bm{\Sigma}}_k(\{\bm{y}_t\}_{t=1}^T)$ which depends on the observations from $\bm{y}_1$ to $\bm{y}_T$,
  \begin{equation}
    \inf_{\widehat{\bm{\Sigma}}_k}\sup_{\mathbb{P}\in\mathcal{P}_\textup{V}(M,\epsilon,r)}\mathbb{E}\left[\|\widehat{\bm{\Sigma}}_k-\bm{\Sigma}_k(\mathbb{P})\|_\textup{op}\right]\gtrsim \left(\frac{pM^{1/\epsilon}}{T}\right)^{\frac{\epsilon}{1+\epsilon}},\quad k=0,1.
  \end{equation}
\end{proposition}

This lower bound matches the upper bound in Proposition \ref{prop:vector} up to a logarithm factor $\log(T)^{\epsilon/(1+\epsilon)}$, indicating that the vector-wise truncated estimators $\widetilde{\bm{\Sigma}}^\text{V}_k(\tau)$ introduced in Section \ref{sec:3.2} are nearly rate-optimal when $\tau$ is chosen properly.

\begin{remark}
  To the best of our knowledge, for covariance or autocovaraince estimation problems, the minimax lower bound in Proposition \ref{prop:LB_operator} is the first operator norm lower bound result with the phase transition phenomenon. This lower bound is sharper than that in \citet{fan2016shrinkage} by constructing a special multivariate discrete distribution with the $(2+2\epsilon)$-th moment condition; see Appendix \ref{append:B} for details.
\end{remark}

\section{Algorithm and Implementation}\label{sec:4}

\subsection{Linearized ADMM algorithm}
The constrained Yule--Walker estimators in \eqref{eq:constrainedYW}, \eqref{eq:split_YW} and \eqref{eq:YuleWalker2} are convex optimization problems as each of them consists of a convex objective function and a convex constraint function. For the constrained problem in \eqref{eq:constrainedYW}, define the constraint set $\mathbb{C}_{\mathcal{R}^*}(\lambda)=\{\bm{M}\in\mathbb{R}^{p\times pd}:\mathcal{R}^*(\bm{M})\leq \lambda\}$ and the optimization problem can be rewritten as
\begin{equation}
    \min_{\bm{A},\bm{D}}\mathcal{R}(\bm{A}),\quad\text{subject to }\bm{D}=\bm{A}\bm{\widetilde{\Sigma}}_0-\bm{\widetilde{\Sigma}}_1~\text{and}~\bm{D}\in\mathbb{C}_{\mathcal{R}^*}(\lambda).
\end{equation}
The augmented Lagrangian form is given as
\begin{equation}
    \mathcal{L}_\rho(\bm{A},\bm{D};\bm{W})=\mathcal{R}(\bm{A})-\langle\bm{W},\bm{A}\bm{\widetilde{\Sigma}}_0-\bm{\widetilde{\Sigma}}_1-\bm{D}\rangle+\frac{\rho}{2}\|\bm{A}\bm{\widetilde{\Sigma}}_0-\bm{\widetilde{\Sigma}}_1-\bm{D}\|_\textup{F}^2,
\end{equation}
subject to $\bm{D}\in\mathbb{C}(\lambda)$, where $\bm{W}$ is the Lagrangian multiplier and $\rho$ is the regularization parameter. The augmented Lagrangian form can be solved by the alternating direction method of multipliers \citep[ADMM]{boyd2011distributed} with the iterative updates of $(\bm{A},\bm{D},\bm{W})$:
\begin{equation}
    \begin{split}
        \bm{A}^{(j+1)}&=\underset{\bm{A}}{\argmin}~\mathcal{R}(\bm{A})+\frac{\rho}{2}\|\bm{A}\bm{\widetilde{\Sigma}}_0-\bm{\widetilde{\Sigma}}_1-\bm{D}^{(j)}-\bm{W}^{(j)}/\rho\|_\textup{F}^2,\\
        \bm{D}^{(j+1)}&=\underset{\bm{D}\in\mathbb{C}_{\mathcal{R}^*}(\lambda)}{\argmin}\|\bm{D}+\bm{\widetilde{\Sigma}}_1-\bm{A}^{(j+1)}\bm{\widetilde{\Sigma}}_0+\bm{W}^{(j)}/\rho\|_\textup{F}^2,\\
        \bm{W}^{(j+1)}&=\bm{W}^{(j)}+\rho(\bm{\widetilde{\Sigma}}_1-\bm{A}^{(j+1)}\bm{\widetilde{\Sigma}}_0+\bm{D}^{(j+1)}).
    \end{split}
\end{equation}
The efficiency of the ADMM algorithm depends largely on the complexity of the resulting subproblems.
Note that $\bm{A}$-update in the ADMM algorithm is a regularized least squares problem of $\bm{A}$, which might not have an explicit solution and can be computationally expensive. 

To alleviate the computational burden in $\bm{A}$-update, we use an inexact proximal method to minimize
\begin{equation}
    \begin{split}
        &\widetilde{\mathcal{L}}_\rho^{(j)}(\bm{A},\bm{D}^{(j)};\bm{W}^{(j)})\\
        =&\mathcal{R}(\bm{A})+\frac{\rho\mu}{2}\left\|\bm{A}-\left[\bm{A}^{(j)}-\frac{2}{\mu}(\bm{A}^{(j)}\bm{\widetilde{\Sigma}}_0-\bm{\widetilde{\Sigma}}_1-\bm{D}^{(j)}-\bm{W}^{(j)}/\rho)\bm{\widetilde{\Sigma}}_0\right]\right\|_\textup{F}^2
    \end{split}
\end{equation}
where $\mu$ is another prespecified parameter. 
As in \citet{wang2012linearized}, $\widetilde{\mathcal{L}}^{(j)}_\rho(\bm{A},\bm{D}^{(j)};\bm{W}^{(j)})$ can be viewed as a linearized approximation of $\mathcal{L}_\rho(\bm{A},\bm{D}^{(j)};\bm{W}^{(j)})$. 

For the separated constrained Yule--Walker estimation in \eqref{eq:split_YW}, the linearized ADMM algorithm can be developed in similar fashion. For the $\ell_1$ norm and nuclear norm regularized estimators, in $\bm{D}$-update and the proximal variant of $\bm{A}$-update, the soft-thresholding operator and truncation operator can be applied to the elements and singular values, respectively. The proposed ADMM algorithm can be illustrated as a particular application of the inexact Bregman ADMM algorithm \citep{wang2014bregman}. To guarantee the global convergence, $\mu/2$ needs to be greater than the largest eigenvalue of $\bm{\widetilde{\Sigma}}_0^2$.

In summary, the ADMM algorithm has explicit updates and can be solved efficiently. The details of the linearized ADMM algorithm for sparse and reduced-rank VAR models are presented in Appendix \ref{append:C}.

\subsection{Linear and semidefinite programmings}\label{sec:4.2}

Some specific cases of the constrained Yule--Walker estimator in \eqref{eq:constrainedYW} can be formulated as linear programmings (LP) or semidefinite programs (SDP), and can thus be efficiently solved by any of the standard LP or SDP solvers.

As discussed in Example \ref{ex:sparse}, for sparse VAR models, a separated constrained minimization framework in \eqref{eq:split_YW} is proposed with $\mathcal{R}(\cdot)=\|\cdot\|_1$ and $\mathcal{R}^*(\cdot)=\|\cdot\|_\infty$. Following the lead of Dantzig selector \citep{candes2007dantzig}, each of the sub-problems can be recast to an LP:
\begin{equation}
    \begin{split}
        \text{minimize} & ~\bm{1}_{pd}^\top\bm{u}\\
        \text{subject to} & ~-\bm{u}\leq \bm{a}_i\leq\bm{u}~\text{and}~-\lambda\cdot\bm{1}_{pd}\leq \widetilde{\bbm{\sigma}}_{1i}-\widetilde{\bm{\Sigma}}_0\bm{a}_i\leq\lambda\cdot\bm{1}_{pd},
    \end{split}
\end{equation}
where the optimization variables are $\bm{u}\in\mathbb{R}^{pd}$ and $\bm{a}_i\in\mathbb{R}^{pd}$. As discussed in \citet{candes2007dantzig}, there is a large class of efficient algorithms for solving such problems.

Similarly, the constrained Yule--Walker estimator in \eqref{eq:constrainedYW}, with $\mathcal{R}(\cdot)=\|\cdot\|_\textup{nuc}$ and $\mathcal{R}^*(\cdot)=\|\cdot\|_\textup{op}$, can be used to estimate the reduced-rank VAR models. As introduced by \citet{candes2011tight}, the nuclear norm minimization problem can be recast to a SDP with a linear matrix inequality constraint:
\begin{equation}
    \begin{split}
        \text{minimize}&~[\text{tr}(\bm{W}_1)+\text{tr}(\bm{W}_2)]/2\\
        \text{subject to}&~\begin{bmatrix}
            \bm{W}_1 & \bm{A} \\
            \bm{A}^\top & \bm{W}_2 \\
        \end{bmatrix} \succeq 0,~\text{and}~
        \begin{bmatrix}
            \lambda\bm{I}_p & \widetilde{\bm{\Sigma}}_1-\bm{A}\widetilde{\bm{\Sigma}}_0\\
            \widetilde{\bm{\Sigma}}_1^\top-\widetilde{\bm{\Sigma}}_0\bm{A}^\top & \lambda\bm{I}_{pd} 
        \end{bmatrix}\succeq 0,
    \end{split}
\end{equation}
with optimization variables $\bm{W}_1\in\mathbb{R}^{p\times p}$, $\bm{W}_2\in\mathbb{R}^{pd\times pd}$, and $\bm{A}\in\mathbb{R}^{p\times pd}$. Interior point methods or first-order methods, such as splitting cone solver, can be applied to efficiently solve large-scale SDP. 

For linear-restricted VAR models, the $\ell_\infty$ optimization problem of the constrained Yule--Walker estimator in \eqref{eq:YuleWalker2} can also be formulated as a LP:
\begin{equation}
    \begin{split}
        \text{minimize}&~t\\
        \text{subject to}&~ -t\cdot\bm{1}_r \leq \bbm{\theta} \leq t\cdot\bm{1}_r,~\text{and}
        ~-\lambda\cdot\bm{1}_r \leq \bm{\widetilde{\Omega}}\bbm{\theta}-\bbm{\widetilde{\omega}} \leq \lambda\cdot\bm{1}_r,
    \end{split}
\end{equation}
with optimization variables $t\in\mathbb{R}$ and $\bbm{\theta}\in\mathbb{R}^r$, and can be solved by standard LP solvers.

As the data truncation and robust autocovariance estimator calculation are computationally cheap, the operation time of ADMM, LP and SDP algorithms mainly depends on the dimension $p$ and AR order $d$. As shown in the simulation experiments in Section \ref{sec:5}, the proposed ADMM algorithm is much more efficient than the LP and SDP solvers, especially when $p$ is large.

\subsection{Tuning parameter selection}\label{sec:4.3}

For the constrained Yule--Walker estimators, both the robustification parameter(s) $\tau$ (or $\tau_1$ and $\tau_2$) and constraint parameter $\lambda$ need to adapt properly to the dimension $p$, sample size $T$ and moment condition bound $M_{2+2\epsilon}$ to achieve optimal trade-off between the truncation bias and tail robustness. However, based on the conditions and convergence rates in Section \ref{sec:3}, the optimal value of both parameters rely on some unknown parameters $\epsilon$ and $M_{2+2\epsilon}$. In the literature on robust estimation of \textit{i.i.d.} data, cross-validation is an intuitive data-driven method to select robustification and regularization parameters. To accommodate the intrinsically ordered nature of time series data, we use a rolling forecasting validation, one of the standard approaches to tuning parameter selection for time series data.

To simultaneously select robustification and constraint parameters, a two-dimensional or multidimensional grid of tuning parameters is constructed.
The search of the constraint parameter $\lambda$ starts from a predetermined $\lambda_{\max}$ and decreases in log-linear increments. For the robustification parameter $\tau$, grid search on the data-driven interval $[\tau_{\min},\tau_{\max}]$ is suggested, where $\tau_{\min}$ and $\tau_{\max}$ are empirical quantiles of the data. For instance, to tune $\tau$ for the element truncation estimators, we may set $\tau_{\min}$ and $\tau_{\max}$ to be the median and maximum of $|y_{it}|$, for all $1\leq i\leq p$ and $1\leq t\leq T$, respectively. If at least one of $p$ and $T$ is very large, we may calculate the quantiles of some randomly sampled $|y_{it}|$.

For the dataset with time series observations from $t=1$ to $T$, we determine a validation period $[T_\text{val},T]$ whose length $T-T_\text{val}$ is supposed to be slightly smaller than $T$. Given the constraint and robustification parameters $(\lambda,\tau)$, at any time point $t\in[T_\text{val},T-1]$, we use all the historical data to calculate the constrained Yule--Walker estimates $\bm{\widehat{A}}(\lambda,\tau)$ and the corresponding one-step-ahead forecast $\bm{\widehat{A}}(\lambda,\tau)\bm{x}_t$. The performance evaluation is based on the mean squared forecast error (MSFE)
\begin{equation}
    \text{MSFE}(\lambda,\tau; T_\text{val}, T)=\frac{1}{T-T_\text{val}}\sum_{t=T_\text{val}+1}^{T}\|\bm{y}_t-\bm{\widehat{A}}(\lambda,\tau)\bm{x}_t\|_2^2,
\end{equation}
and the tuning parameters are selected by minimizing the $\text{MSFE}(\lambda,\tau;T_\text{val},T)$ over the grid points.

\section{Simulation Study}\label{sec:5}

\subsection{VAR estimation}\label{sec:5.1}

In this subsection, we compare the performance of the robust procedure and two standard estimation procedures for sparse, reduced-rank, and banded VAR models in finite samples. For each model, we consider three innovation settings: standardized $t_{2.1}$-distributed innovations, standardized log-normal innovations, and standard Gaussian innovations. They represent heavy-tailed symmetric distributions with finite second moments, heavy-tailed asymmetric distributions with finite fourth moments, and light-tailed symmetric distribution, respectively. The data are drawn from the VAR(1) process $\bm{y}_t=\bm{A}^*\bm{y}_{t-1}+\bbm{\varepsilon}_t$ with the following structures of $\bm{A}^*$.

DGP-1: Sparse VAR with $\bm{A}^*$ specified as $\bm{A}^*_{ij}=0.5$ if $i=j$, $\bm{A}^*_{ij}=0.4$ if $i-j=1$, $\bm{A}^*_{ij}=-0.4$ if $i-j=-1$, and $\bm{A}^*_{ij}=0$ otherwise. We consider $p\in\{30,50\}$ and $T\in\{100,150,200,250,300\}$.

DGP-2: Reduced-rank VAR with $\bm{A}^*$ having a low-rank singular value decomposition $\bm{A}^*=\bm{U}\bm{D}\bm{V}^\top$. Specifically, we consider a fixed $\bm{D}=\text{diag}(1.5,0.8)$, and generate random orthonormal $\bm{U}\in\mathbb{R}^{p\times 2}$ and $\bm{V}\in\mathbb{R}^{p\times 2}$ from the first two leading singular vectors of a random $p$-by-$p$ Gaussian matrix. We consider $p\in\{30,50\}$ and $T\in\{300,350,400,450,500\}$.

DGP-3: Banded VAR with a banded $\bm{A}^*$ having bandwidth $k=2$. Specifically, $\bm{A}^*_{ij}=0.5$ if $i=j$; $\bm{A}^*_{ij}=0.3$ if $i-j=1$; $\bm{A}^*_{ij}=-0.3$ if $i-j=-1$; $\bm{A}^*_{ij}=0.2$ if $|i-j|=2$; and $\bm{A}^*_{ij}=0$ for all $|i-j|>2$. We consider $p\in\{50,100\}$ and $T\in\{100,150,200,250,300\}$.

We first specify the robust estimation procedure, denoted as ROB, for these three DGPs. For DGP-1, as discussed in Example \ref{ex:sparse}, we apply the split estimation method in \eqref{eq:split_YW} with $\mathcal{R}(\cdot)=\|\cdot\|_1$, $\mathcal{R}(\cdot)=\|\cdot\|_\infty$, and the element-wise truncation autocovariance estimators $\widetilde{\bm{\Sigma}}_0^\text{E}(\tau)$ and $\widetilde{\bm{\Sigma}}_1^\text{E}(\tau)$ specified in Section \ref{sec:3.1}. For DGP-2, as discussed in Example \ref{ex:reduced-rank}, we apply the estimation method in \eqref{eq:constrainedYW} with $\mathcal{R}(\cdot)=\|\cdot\|_\textup{nuc}$, $\mathcal{R}^*(\cdot)=\|\cdot\|_\textup{op}$, and the vector-wise truncation autocovariance estimators $\widetilde{\bm{\Sigma}}_0^\text{V}(\tau)$ and $\widetilde{\bm{\Sigma}}_1^\text{V}(\tau)$ specified in Section \ref{sec:3.2}. For DGP-3, as discussed in Example \ref{ex:banded}, we apply the estimation method in \eqref{eq:YuleWalker2} with the element-wise truncation autocovariance estimators $\widetilde{\bm{\Omega}}^\text{E}(\tau)$ and $\widetilde{\bbm{\omega}}^\text{E}(\tau)$ in Section \ref{sec:3.3}. The tuning parameters $\lambda$ and $\tau$ are selected simultaneously by the MSFE in Section \ref{sec:4.3}.

Alternatively, if the possibility of 
heavy-tailed distribution is ignored, we consider the sample autocovariance estimators
\begin{equation}
  \widehat{\bm{\Sigma}}_0=\frac{1}{T}\sum_{t=1}^T\bm{y}_t\bm{y}_t^\top,~~\widehat{\bm{\Sigma}}_1=\frac{1}{T-1}\sum_{t=2}^T\bm{y}_t\bm{y}_{t-1}^\top,~~\widehat{\bm{\Omega}}=\bm{C}^\top(\bm{I}_p\otimes\widetilde{\bm{\Sigma}}_0)\bm{C},~~\text{and}~~\bm{C}^\top\textup{vec}(\widehat{\bm{\Sigma}}_1^\top),
\end{equation}
and plug them in the Yule--Walker estimators. We denote this estimation procedure as the regular Yule--Walker (RYW) method. In the high-dimensional VAR literature, another popular estimation procedure is the regularized least squares (RLS) method. For DGP-1 and 2, we apply the regularized method
\begin{equation}
  \widehat{\bm{A}}_\text{RLS}=\argmin_{\bm{A}}\frac{1}{T}\sum_{t=1}^T\|\bm{y}_t-\bm{A}\bm{y}_{t-1}\|_2^2+\lambda\mathcal{R}(\bm{A})
\end{equation}
with $\mathcal{R}(\cdot)=\|\cdot\|_1$ and $\|\cdot\|_\textup{nuc}$, respectively. For DGP-3, we consider the least squares (LS) estimator
\begin{equation}
  \widehat{\bbm{\theta}}_\text{LS}=\argmin_{\bbm{\theta}}\frac{1}{T}\|\bm{y}_t-(\bm{I}_p\otimes\bm{y}_{t-1}^\top)\bm{C}\bbm{\theta}\|_2^2.
\end{equation}

To evaluate the estimation performance and to confirm the theoretical results in Theorems \ref{thm:sparseAR}-\ref{thm:banded}, for $\widehat{\bm{A}}-\bm{A}^*$, we calculate the $\ell_{2,\infty}$ norm for sparse VAR models, the Frobenius norm for reduced-rank VAR models, and the operator norm and $\ell_\infty$ norm for banded VAR models. For the three DGPs, the estimation errors are calculated by averaging 500 replications for each setting and presented in Figures \ref{fig:sim1}, \ref{fig:sim2} and \ref{fig:sim3}, respectively.

\begin{figure}[!htp]
  \begin{center}
    \includegraphics[width=0.9\textwidth]{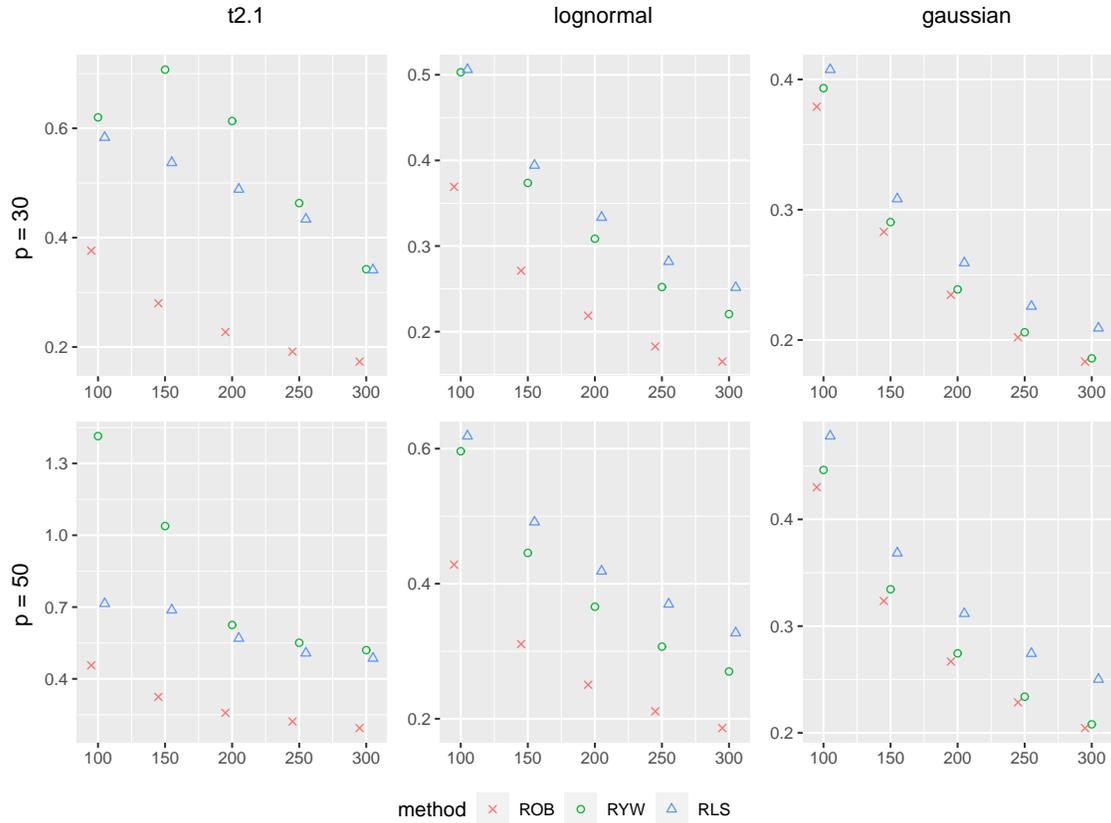}
    \vspace{-0.5cm}
    \caption{Estimation errors of robust (ROB), regularized Yule--Walker (RYW), and regularized least squares (RLS) methods in the $\ell_{2,\infty}$ norm v.s. sample size $T$ for various dimension and innovation distribution settings in DGP-1.}
    \label{fig:sim1}
  \end{center}
\end{figure}

\begin{figure}[!htp]
    \begin{center}
        \includegraphics[width=0.9\textwidth]{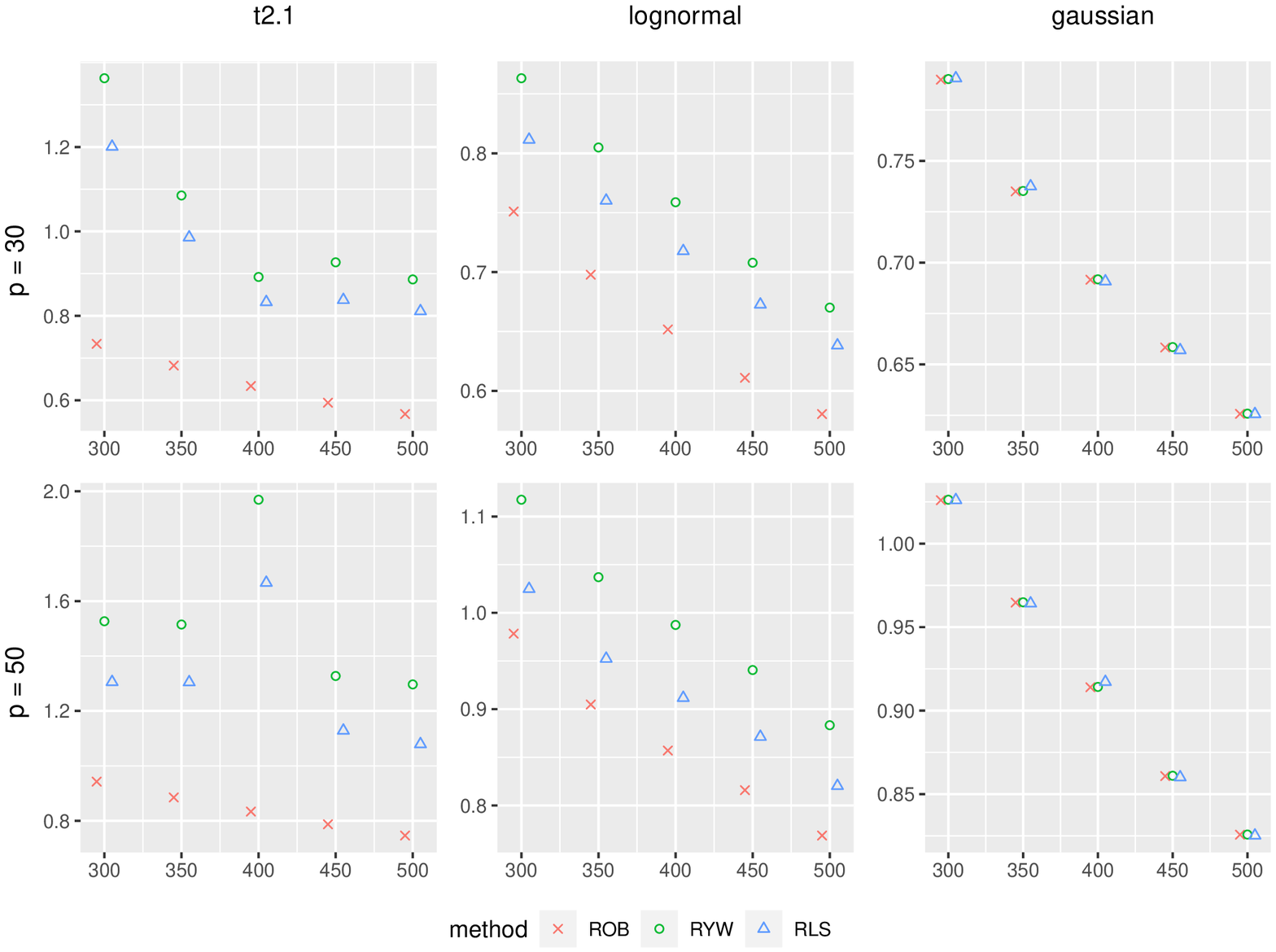}
    \vspace{-0.5cm}
        \caption{Estimation errors of robust (ROB), regularized Yule--Walker (RYW), and regularized least squares (RLS) methods in the Frobenius norm v.s. sample size $T$ for various dimension and innovation distribution settings in DGP-2.}
        \label{fig:sim2}
    \end{center}
\end{figure}

\begin{figure}[!htp]
    \begin{center}
        \includegraphics[width=0.9\textwidth]{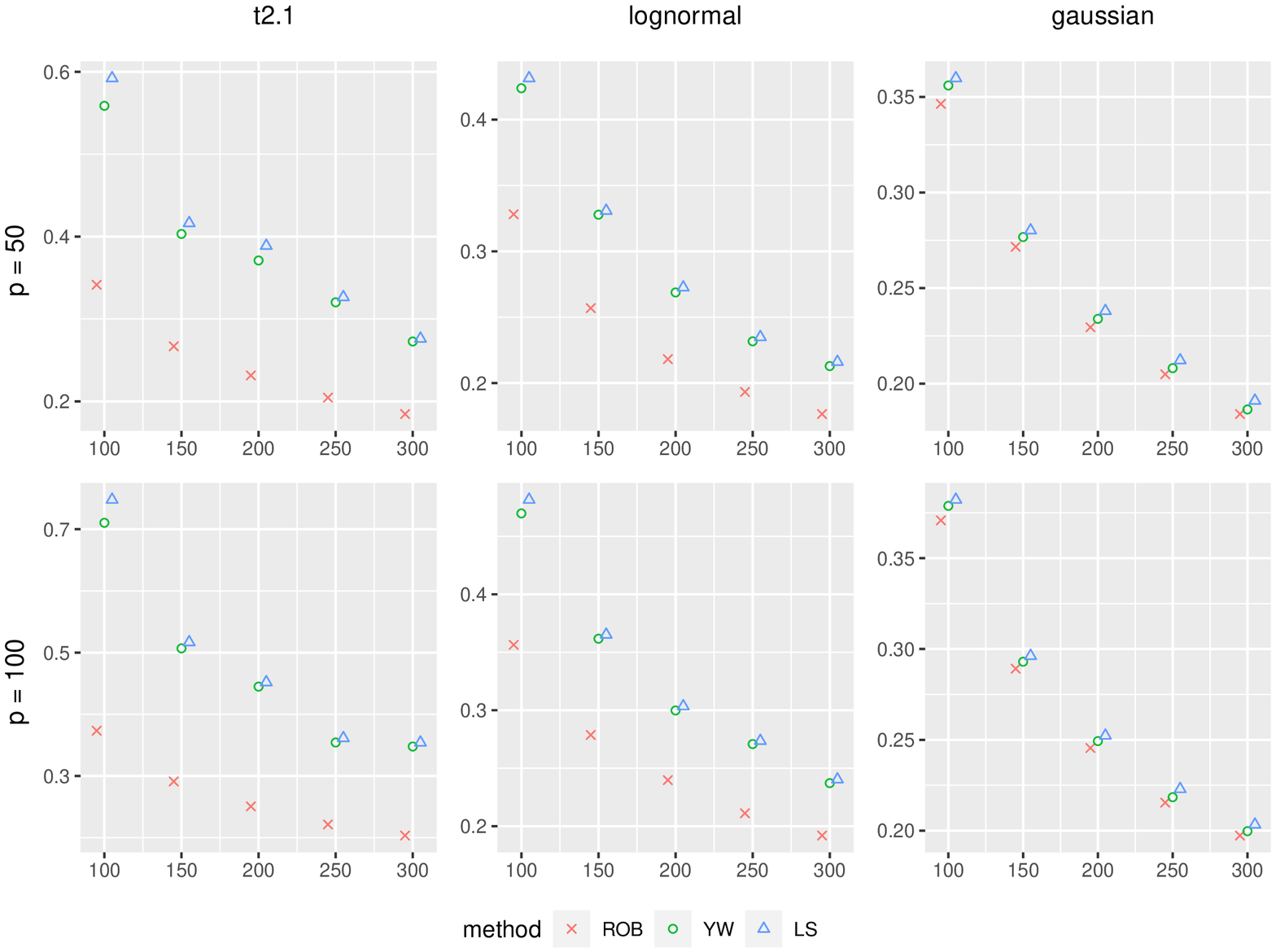}
    \vspace{-0.5cm}
        \caption{Estimation errors of robust (ROB), regularized Yule--Walker (RYW), and least squares (LS) methods in the operator norm v.s. sample size $T$ for various dimension and innovation distribution settings in DGP-3.}
        \label{fig:sim3}
    \end{center}
\end{figure}

For the three types of VAR models considered, we can observe from Figures \ref{fig:sim1}-\ref{fig:sim3} that the proposed robust estimation method yields much small statistical errors than both standard estimation methods under the heavy-tailed innovations, i.e. the standardized $t_{2.1}$ distributed innovation and log-normal innovation. In particular, under the $t_{2.1}$ innovation setting, the estimation performance of both standard estimation methods is not stable.
Under the setting of Gaussian innovation, the proposed robust estimator produces almost the same or even slightly smaller estimation errors than the standard methods, indicating that it does not hurt to use the robust procedure under the light-tailed innovation setting. In sum, the simulation results generally confirm the theoretical convergence rates and demonstrate the robustness of the proposed estimators against heavy-tailed distributions.

We also compare the performance of two computational algorithms discussed in Section \ref{sec:4}, using the simulated data with a fixed $T=300$ and varying dimension $p\in\{25,50,75,100,125\}$ from the Gaussian settings of DGP-1 and 2, respectively. We obtain almost identical estimates from the two algorithms under the same choices of $\lambda$ and $\tau$. Regarding computational time, we record the averaged time consumed over 100 replications for each algorithm. Figure \ref{fig:com_time} demonstrates that the ADMM is much more efficient than the LP and SDP solvers, especially when the dimension $p$ is large.

\begin{figure}[!htp]
  \begin{center}
    \includegraphics[width=0.65\textwidth]{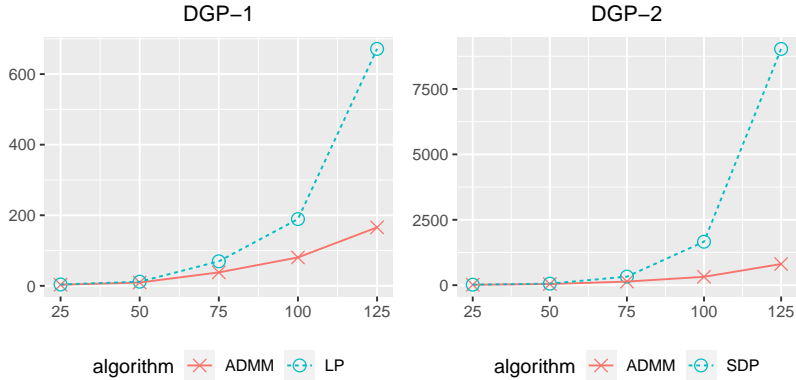}
    \vspace{-0.5cm}
    \caption{Computational time (in seconds) of ADMM and LP/SDP algorithms v.s. dimension $p$ in DGP-1 and DGP-2.}
    \label{fig:com_time}
  \end{center}
\end{figure}

\subsection{Autocovariance matrix estimation}

We also conduct a simulation experiment to compare the truncation-based autocovariance matrix estimator with the standard sample autocovariance estimator. Two data generating processes, DGP-1 and DGP-2 in Subsection \ref{sec:5.1}, are employed. Similarly to the previous experiment, we also consider three innovation settings for each DGP. The element and vector truncation estimators are considered for these two processes. Based on 1000 replications, the average estimation errors of autocovariance matrices in terms of $\ell_{\infty}$ and operator norms are presented in Figures \ref{fig:sim4} and \ref{fig:sim5} for the two processes, respectively.

\begin{figure}
  \begin{center}
    \includegraphics[width=0.9\textwidth]{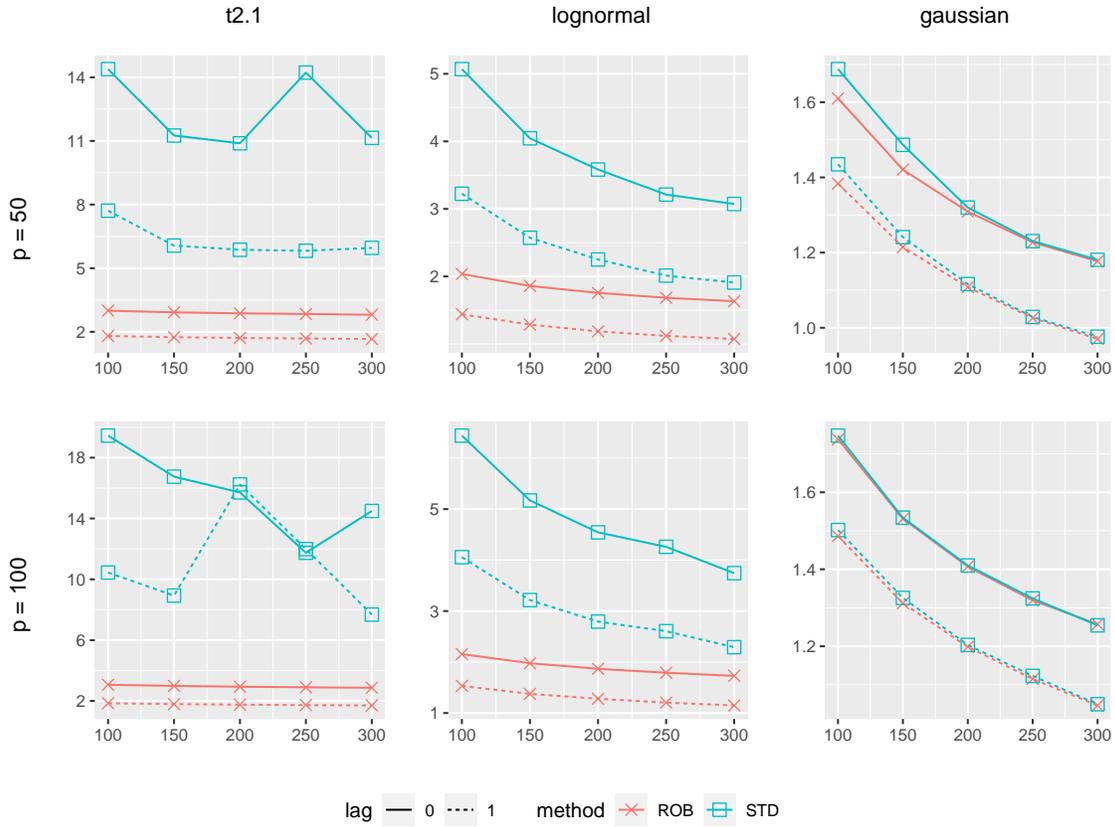}
    \vspace{-0.2cm}
    \caption{Estimation errors of $\widetilde{\bm{\Sigma}}_0^\textup{E}(\tau)$ and $\widetilde{\bm{\Sigma}}_1^\textup{E}(\tau)$ in the $\ell_\infty$ norm v.s. sample size $T$ for different dimension and innovation distribution settings in DGP-1.}
    \label{fig:sim4}
  \end{center}
\end{figure}

\begin{figure}
  \begin{center}
    \includegraphics[width=0.9\textwidth]{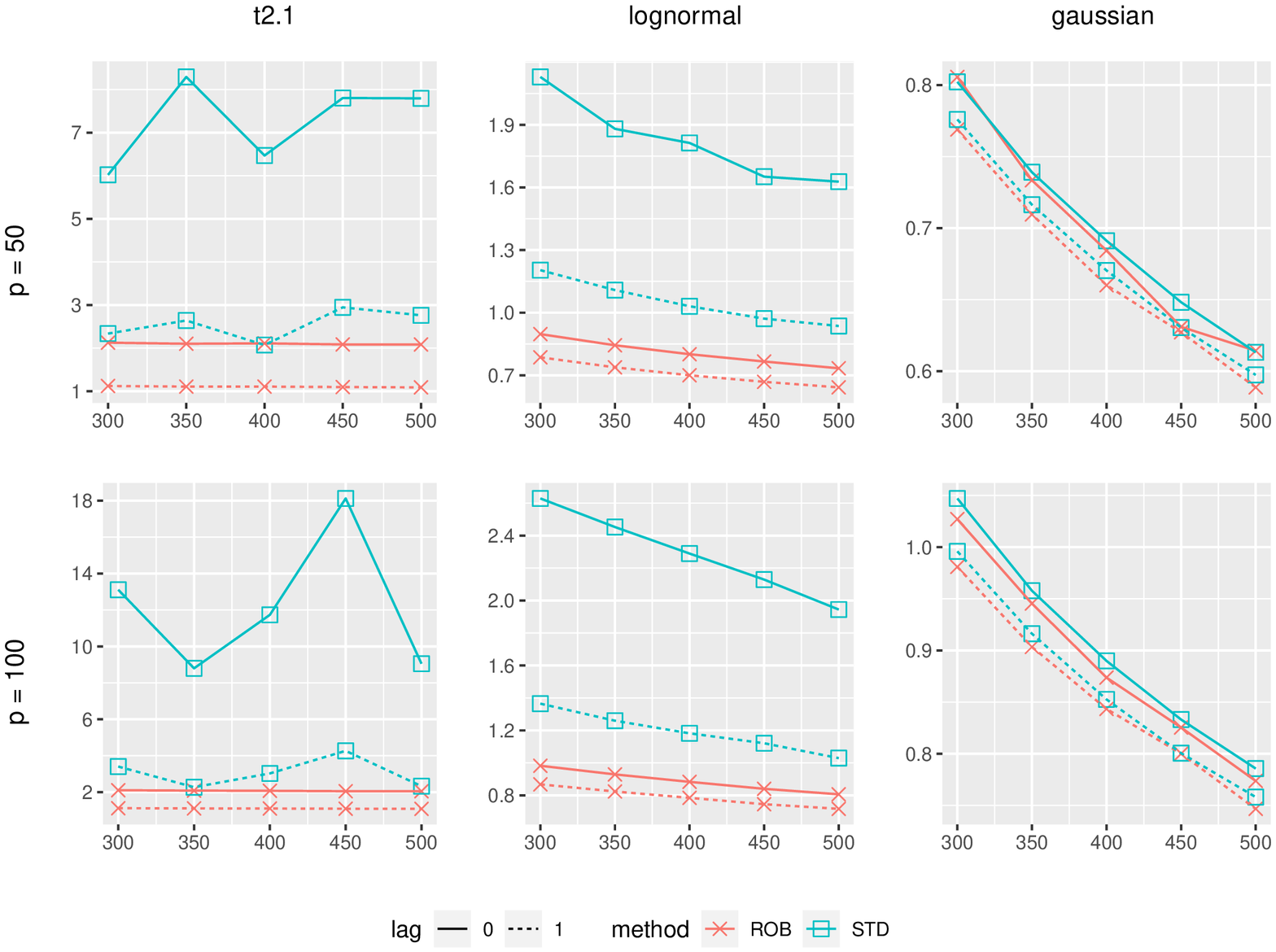}
    \vspace{-0.2cm}
    \caption{Estimation errors of $\widetilde{\bm{\Sigma}}_0^\textup{E}(\tau)$ and $\widetilde{\bm{\Sigma}}_1^\textup{E}(\tau)$ in the operator norm v.s. sample size $T$ for different dimension and innovation distribution settings in DGP-2.}
    \label{fig:sim5}
  \end{center}
\end{figure}

For these two DGPs, the truncation-based autocovariance matrix estimators yield much smaller estimation errors than the sample autocovariance estimator under the heavy-tailed innovation settings. Similarly to the previous experiment, under the $t_{2.1}$ innovation setting, the performance of the sample autocovariance estimator is not stable, but the smoothly decreasing pattern of the proposed robust estimator can be observed. Under the Gaussian innovation setting, the performance of both estimators are nearly the same. The numerical results in this experiment confirm the rates of convergence in Propositions \ref{prop:element} and \ref{prop:vector} and verify the robustness of the proposed autocovariance estimators.

\section{Real Data Example}\label{sec:6}

In this section, we apply the proposed robust estimation procedure to a real data set consisting of 40 quarterly macroeconomic variables of the United States from the third quarter of 1959 to fourth quarter of 2007, with 194 observations for each variable. Following \citet{stock2009forecasting}, all variables are seasonally adjusted except for financial variables, transformed by differencing or log differencing to be stationary, and standardized to zero mean and unit standard deviation. These macroeconomic variables capture many aspects of the U.S. economy, including GDP, National Association of Purchasing Manager indices, industrial production, pricing, employment, credit and interest rate. Both factor model and VAR model have been applied to these series in empirical econometric analyses for structural analysis and forecasting; see \citet{stock2009forecasting}, \citet{koop2013forecasting} and \citet{wang2019high}. The kurtosis of the standardized variables are plotted in Figure \ref{fig:KurtHist}, with the blue and red dashed lines indicating the kurtosis of normal distribution and $t_6$ distribution, respectively. Among 40 variables, 36 and 11 variables have larger kurtosis than standard normal distribution and $t_6$ distribution, providing some evidence of the existence of heavy-tailed distributions in this macroeconomic dataset.

\begin{figure}[!htp]
    \centering
    \includegraphics[width=0.8\textwidth]{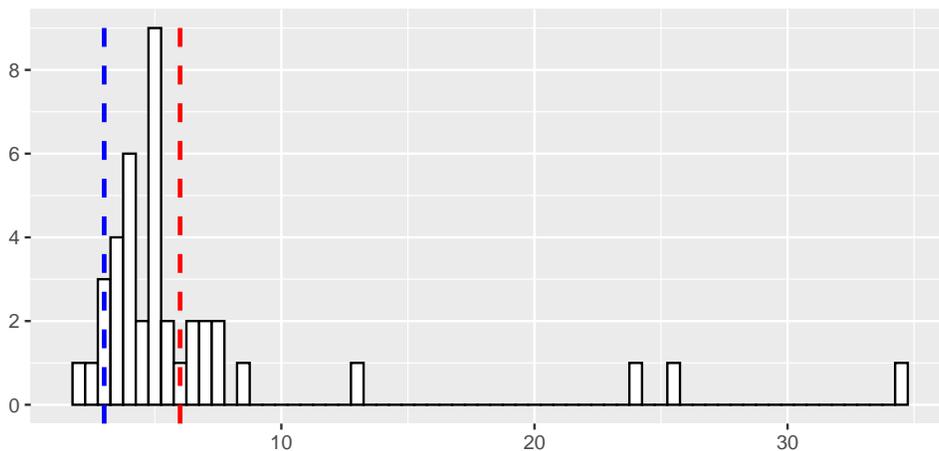}
    \vspace{-0.5cm}
    \caption{Histogram of kurtosis of 40 macroeconomic variables. Blue and red dashed line mark the theoretical kurtosis of standard normal distribution and $t_6$ distribution.}
    \label{fig:KurtHist}
\end{figure}

Following \citet{koop2013forecasting}, we apply a VAR(4) model to these macroeconomic time series. If no structural assumption is imposed on the parameter matrix, we can directly use the ordinary least squares estimator by minimizing the least squares loss function, namely
\begin{equation}
    \bm{\widehat{A}}_\text{OLS}=\argmin_{\bm{A}\in\mathbb{R}^{40\times 160}}\frac{1}{T}\sum_{t=1}^T\|\bm{y}_t-\bm{A}\bm{x}_{t}\|_2^2,
\end{equation}
where $\bm{x}_t=(\bm{y}_{t-1}^\top,\cdots,\bm{y}_{t-4}^\top)^\top$. If the low-rank structure on $\bm{A}$ is considered and the rank $r$ is pre-specified, the rank-constrained least squares estimator is formulated as
\begin{equation}
    \bm{\widehat{A}}_\text{RRR}=\argmin_{\bm{A}\in\mathbb{R}^{40\times 160},\text{rank}(\bm{A})\leq r}\frac{1}{T}\sum_{t=1}^T\|\bm{y}_t-\bm{A}\bm{x}_{t}\|_2^2.
\end{equation}
As shown in Section \ref{sec:5}, if the low-rank structure is not directly imposed, we can use the nuclear norm regularized least squares (NN-RLS), nuclear norm regularized Yule--Walker (NN-RYW), and the proposed nuclear norm robust (NN-ROB) methods. For the sparse VAR model, the $\ell_1$ regularized least squares (L1-RLS), $\ell_1$ regularized Yule--Walker (L1-RYW), and $\ell_1$ robust (L1-ROB) methods can be applied as discussed in Section \ref{sec:5}.

The performance of these eight methods are compared via out-of-sample forecasting errors. From the first quarter of 1993 ($t=135$) to the fourth quarter 2007 ($t=194$), we fit the VAR model by different methods utilizing all the historical data available until time $t-1$ and obtain the one-step-head forecast $\bm{\widehat{A}}\bm{x}_t$. Then, we calculate the rolling forecasting errors $\widehat{\bbm{\delta}}_t=\bm{\widehat{A}}\bm{x}_t-\bm{y}_t$ in the $\ell_2$ and $\ell_\infty$ norms, respectively, and summarized the their means and medians in Table \ref{tbl:2}.

\begin{table}
    \centering
  \renewcommand{\arraystretch}{1.4}
  \caption{Mean and median of out-of-sample rolling forecasting errors of various methods for the full, sparse, and reduced-rank VAR(4) models.}
    \scriptsize{\begin{tabular}{l|ccccccccc}
      \hline
      Model & Full & \multicolumn{3}{c}{Sparse} && \multicolumn{4}{c}{Reduced-rank}\\
      \cline{3-5} \cline{7-10}
      Method & OLS & L1-RLS & L1-RYW & L1-ROB && RRR & NN-RLS & NN-RYW & NN-ROB \\
      \hline
      mean of $\|\widehat{\bbm{\delta}}_t\|_2$ & 22.63 & 5.64 & 5.82 & \textbf{4.17} && 19.67 & 8.82 & 8.91 & 6.52 \\
      median of $\|\widehat{\bbm{\delta}}_t\|_2$ & 14.57 & 5.21 & 5.20 & \textbf{3.81} && 9.62 & 5.39 & 5.34 & 3.97 \\
      mean of $\|\widehat{\bbm{\delta}}_t\|_\infty$ & 9.29 & 2.16 & 2.42 & \textbf{1.66} && 6.42 & 3.44 & 3.27 & 2.49 \\
      median of $\|\widehat{\bbm{\delta}}_t\|_\infty$ & 6.42 & 1.96 & 2.01 & \textbf{1.43} && 3.06 & 2.16 & 2.09 & 1.80 \\
      \hline
    \end{tabular}}
    \vspace{0.3cm}
    \label{tbl:2}
\end{table}

As shown in Table \ref{tbl:2}, the OLS estimator produces the largest out-of-sample prediction errors as the full VAR(4) model is highly over-parameterized, while low-dimensional and sparse structure can alleviate the problem of overfitting and significantly improve the forecasting performance. Both L1-ROB and NN-ROB estimators have smaller out-of-sample forecasting errors than their standard counterparts, indicating the necessity of robust estimation for large-scale time series data and the satisfactory performance of the proposed methodology. The L1-ROB estimator performs best among all estimators as it produces a parsimonious sparse parameter matrix and prevents overfitting effectively.




\bibliographystyle{imsart-nameyear}
\bibliography{mybib.bib}

\clearpage
\begin{center}
{\Large\textbf{Supplementary material for \\ ``Rate-Optimal Robust Estimation of High-Dimensional Vector Autoregressive Models"}}\\~\\
{Di Wang and Ruey S. Tsay
\\\textit{Booth School of Business, University of Chicago}}
\end{center}
\vspace{0.3cm}

This supplementary material provides all technical proofs of the theoretical results in the main paper, some auxiliary lemmas, and details of the linearized ADMM algorithm. 
Specifically, Appendices \ref{append:A} and \ref{append:B} present the proofs of upper and lower bound results, respectively, and some related auxiliary lemmas.
Appendix \ref{append:C} shows the ADMM algorithms for the $\ell_1$ regularized sparse VAR model and the nuclear norm regularized reduced-rank VAR model.

\begin{appendix}

\section{Proofs of Upper Bound Results}\label{append:A}

 We present the proofs of the deterministic upper bounds (Propositions \ref{prop:1}--\ref{prop:3}) in Section \ref{sec:A.1}, the error bounds of the robust autocovariance estimators (Propositions \ref{prop:element}--\ref{prop:linear}) in Section \ref{sec:A.2}, and the error bounds of the robust VAR estimators (Theorems \ref{thm:sparseAR}--\ref{thm:network}) in Section \ref{sec:A.3}, respectively. The auxiliary lemmas are relegated to Section \ref{sec:A.4}.
 
 We start with some notation used in the Appendix. Denote by $\mathbb{S}^{p}=\{\bm{v}\in\mathbb{R}^{p}:\|\bm{v}\|_2=1\}$ the $p$-dimensional sphere with unit radius in the Euclidean norm. For any symmetric matrix $\bm{M}$, let  $\lambda_{\max}(\bm{M})$ and $\lambda_{\min}(\bm{M})$ be its largest and smallest eigenvalues.

\subsection{Proof of Propositions \ref{prop:1}--\ref{prop:3}}\label{sec:A.1}

\begin{proof}[\textbf{Proof of Proposition \ref{prop:1}}]

The general idea of the proof follows from that of Dantzig selector \citep{candes2007dantzig} and constrained minimization estimation \citep{cai2011constrained,han2015direct}.

By the conditions in Proposition \ref{prop:1}, $\mathcal{R}^*(\widetilde{\bm{\Sigma}}_0-\bm{\Sigma}_0)\leq\zeta_0$ and $\mathcal{R}^*(\widetilde{\bm{\Sigma}}_1-\bm{\Sigma}_1)\leq\zeta_1$. Let $\bm{\Delta}=\widehat{\bm{A}}-\bm{A}^*$ be the estimation error. We first show that the true value $\bm{A}^*$ is a feasible solution to the optimization problem in \eqref{eq:constrainedYW}. Note that as $\mathcal{R}^*(\cdot)$ is convex,
\begin{equation}
  \begin{split}
    &\mathcal{R}^*(\bm{A}^*\bm{\widetilde{\Sigma}}_0-\bm{\widetilde{\Sigma}}_1)
    =\mathcal{R}^*(\bm{A}^*\bm{\widetilde{\Sigma}}_0-\bm{\Sigma}_1+\bm{\Sigma}_1-\bm{\widetilde{\Sigma}}_1)\\
    =&\mathcal{R}^*(\bm{A}^{*}\bm{\widetilde{\Sigma}}_0-\bm{A}^*\bm{\Sigma}_0+\bm{\Sigma}_1-\bm{\widetilde{\Sigma}}_1)\\
    \leq&\mathcal{R}^*(\bm{A}^*[\bm{\widetilde{\Sigma}}_0-\bm{\Sigma}_0])+\mathcal{R}^*(\bm{\Sigma}_1-\bm{\widetilde{\Sigma}}_1)\leq\mathcal{R}(\bm{A}^*)\mathcal{R}^*(\bm{\widetilde{\Sigma}}_0-\bm{\Sigma}_0)+\zeta_1\\
    \leq&\zeta_0\mathcal{R}(\bm{A}^*)+\zeta_1\leq\lambda,
  \end{split}
\end{equation}
where the first inequality follows from the triangle inequality, and the second inequality follows from the properties of $\mathcal{R}^*(\cdot)$ and $\mathcal{R}(\cdot)$.

Therefore, $\bm{A}^*$ is feasible in the optimization equation, and hence $\mathcal{R}(\widehat{\bm{A}})\leq \mathcal{R}(\bm{A}^*)$. Then, by the triangle inequality, we have
\begin{equation}
  \begin{split}
    & \mathcal{R}^*(\bm{\Delta}) = \mathcal{R}^*(\widehat{\bm{A}}-\bm{\Sigma}_1\bm{\Sigma}_0^{-1})\\
    = & \mathcal{R}^*([\widehat{\bm{A}}\bm{\Sigma}_0-\bm{\Sigma}_1]\bm{\Sigma}_0^{-1})\\
    = & \mathcal{R}^*([\widehat{\bm{A}}\bm{\Sigma}_0-\widehat{\bm{A}}\widetilde{\bm{\Sigma}}_0+\widehat{\bm{A}}\widetilde{\bm{\Sigma}}_0-\widetilde{\bm{\Sigma}}_1+\widetilde{\bm{\Sigma}}_1-\bm{\Sigma}_1]\bm{\Sigma}_0^{-1})\\
    \leq & \{\mathcal{R}^*(\widehat{\bm{A}}[\bm{\Sigma}_0-\widetilde{\bm{\Sigma}}_0]) + \mathcal{R}^*(\widehat{\bm{A}}\widetilde{\bm{\Sigma}}_0-\widetilde{\bm{\Sigma}}_1) + \mathcal{R}^*(\widetilde{\bm{\Sigma}}_1-\bm{\Sigma}_1)\}\cdot\mathcal{C}(\bm{\Sigma}_0^{-1})\\
    \leq & [\zeta_0\mathcal{R}(\widehat{\bm{A}})+\lambda+\zeta_1]\cdot\mathcal{C}(\bm{\Sigma}_0^{-1}) \leq [\zeta_0\mathcal{R}(\bm{A}^*)+\lambda+\zeta_1]\cdot\mathcal{C}(\bm{\Sigma}_0^{-1}) \leq 2\lambda\hspace{0.01in}\mathcal{C}(\bm{\Sigma}_0^{-1}).
  \end{split} 
\end{equation}

By the definitions of model subspace and perturbation subspace, we can decompose $\bm{A}^*=\bm{A}^*_{\mathcal{M}}+\bm{A}^*_{\overline{\mathcal{M}}^\perp}$. Since $\bm{A}^*$ is feasible with respect to the constraint, we have
\begin{equation}
  \label{eq:tech1}
  \mathcal{R}(\widehat{\bm{A}}) \leq \mathcal{R}(\bm{A}^*) = \mathcal{R}(\bm{A}^*_{\mathcal{M}} + \bm{A}^*_{\overline{\mathcal{M}}^\perp}) = \mathcal{R}(\bm{A}^*_{\mathcal{M}}) + \mathcal{R}(\bm{A}^*_{\overline{\mathcal{M}}^\perp}),
\end{equation}
where the last equality follows from the decomposability of $\mathcal{R}(\cdot)$.

Moreover, we have
\begin{equation}
  \begin{split}
    \mathcal{R}(\bm{\widehat{A}})=&\mathcal{R}(\bm{A}^*+\bm{\Delta})=\mathcal{R}(\bm{A}^*_{\mathcal{M}}+\bm{\Delta}_{\overline{\mathcal{M}}^\perp}+\bm{A}^*_{\overline{\mathcal{M}}^\perp}+\bm{\Delta}_{\overline{\mathcal{M}}})\\
    \geq& \mathcal{R}(\bm{A}^*_{\mathcal{M}}+\bm{\Delta}_{\overline{\mathcal{M}}^\perp})-\mathcal{R}(\bm{A}^*_{\overline{\mathcal{M}}^\perp}+\bm{\Delta}_{\overline{\mathcal{M}}})\\
    \geq &\mathcal{R}(\bm{A}^*_{\mathcal{M}}) + \mathcal{R}(\bm{\Delta}_{\overline{\mathcal{M}}^\perp}) - \mathcal{R}(\bm{A}^*_{\overline{\mathcal{M}}^\perp}) -  \mathcal{R}(\bm{\Delta}_{\overline{\mathcal{M}}}),
  \end{split}
\end{equation}
where the first inequality follows from the triangle inequality and the second inequality follows from the decomposibility of $\mathcal{R}(\cdot)$ and the triangle inequality. Thus, together with \eqref{eq:tech1}, we have
$\mathcal{R}(\bm{\Delta}_{\overline{\mathcal{M}}^\perp}) \leq \mathcal{R}(\bm{\Delta}_{\overline{\mathcal{M}}}) + 
2\mathcal{R}(\bm{A}^*_{\overline{\mathcal{M}}^\perp})$.

Hence, the upper bound in terms of $\mathcal{R}(\cdot)$ can be found as
\begin{equation}
  \begin{split}
    &\mathcal{R}(\bm{\Delta}) = \mathcal{R}(\bm{\Delta}_{\overline{\mathcal{M}}}+\bm{\Delta}_{\overline{\mathcal{M}}^\perp})\\
    \leq & \mathcal{R}(\bm{\Delta}_{\overline{\mathcal{M}}}) + \mathcal{R}(\bm{\Delta}_{\overline{\mathcal{M}}^\perp})
    \leq 2\mathcal{R}(\bm{\Delta}_{\overline{\mathcal{M}}}) + 2\mathcal{R}(\bm{A}^*_{\overline{\mathcal{M}}^\perp})\\
    \leq & 2\phi(\overline{\mathcal{M}})\mathcal{R}^*(\bm{\Delta}_{\overline{\mathcal{M}}}) + 2\mathcal{R}(\bm{A}^*_{\overline{\mathcal{M}}^\perp})\\ 
    \leq & 2\phi(\overline{\mathcal{M}})\mathcal{R}^*(\bm{\Delta}) + 2\mathcal{R}(\bm{A}^*_{\overline{\mathcal{M}}^\perp}) \\
    \leq & 4\phi(\overline{\mathcal{M}})\lambda\hspace{0.01in}\mathcal{C}(\bm{\Sigma}_0^{-1}) + 2\mathcal{R}(\bm{A}^*_{\overline{\mathcal{M}}^\perp}).
  \end{split}
\end{equation}

Finally, by the duality between $\mathcal{R}(\cdot)$ and $\mathcal{R}^*(\cdot)$, the upper bound in the squared Frobenius norm can be derived as
\begin{equation}
  \|\bm{\Delta}\|_\textup{F}^2 \leq \mathcal{R}(\bm{\Delta})\mathcal{R}^*(\bm{\Delta})\leq 8\phi(\overline{\mathcal{M}})\lambda^2\hspace{0.01in}\mathcal{C}^2(\bm{\Sigma}_0^{-1})+4\lambda\mathcal{R}(\bm{A}^*_{\overline{\mathcal{M}}^\perp})\hspace{0.01in}\mathcal{C}(\bm{\Sigma}_0^{-1}).
\end{equation}
\end{proof}

\begin{proof}[\textbf{Proof of Proposition \ref{prop:2}}]

The proof of Proposition \ref{prop:2} generally follows that of Proposition \ref{prop:1}. Note that we have $\overline{\mathcal{R}}^*(\widetilde{\bm{\Sigma}}_0-\bm{\Sigma}_0)\leq\zeta_0$ and $\mathcal{R}^*(\widetilde{\bbm{\sigma}}_{1i}-\bbm{\sigma}_{1i})\leq\zeta_{1i}$. Let $\bbm{\delta}_i=\widehat{\bm{a}}_i-\bm{a}_i^*$. We first show that the true value $\bm{a}_i^*$ is a feasible solution to the optimization sub-problem in \eqref{eq:split_YW}. Note that
\begin{equation}
  \begin{split}
    &\mathcal{R}^*(\bm{\widetilde{\Sigma}}_0\bm{a}_i^*-\bbm{\widetilde{\sigma}}_{1i})
    =\mathcal{R}^*(\bm{\widetilde{\Sigma}}_0\bm{a}_i^*-\bbm{\sigma}_{1i}+\bbm{\sigma}_{1i}-\bbm{\widetilde{\sigma}}_{1i})\\
    =&\mathcal{R}^*(\bm{\widetilde{\Sigma}}_0\bm{a}_{i}^{*}-\bm{\Sigma}_0\bm{a}_i^*+\bbm{\sigma}_{1i}-\bbm{\widetilde{\sigma}}_{1i})\\
    \leq&\mathcal{R}^*([\bm{\widetilde{\Sigma}}_0-\bm{\Sigma}_0]\bm{a}_i^*)+\mathcal{R}^*(\bbm{\sigma}_{1i}-\bbm{\widetilde{\sigma}}_{1i})\leq\overline{\mathcal{R}}^*(\bm{\widetilde{\Sigma}}_0-\bm{\Sigma}_0)\mathcal{R}(\bm{a}_i^*)+\zeta_{2i}\\
    \leq&\zeta_0\mathcal{R}(\bm{a}_i^*)+\zeta_{1i}\leq\lambda.
  \end{split}
\end{equation}
Therefore, $\bm{a}_i^*$ is feasible in the optimization equation, and $\mathcal{R}(\widehat{\bm{a}}_i)\leq \mathcal{R}(\bm{a}_i^*)$. Then, by triangle inequality, we have
\begin{equation}
  \begin{split}
    & \mathcal{R}^*(\bbm{\delta}_i) = \mathcal{R}^*(\widehat{\bm{a}}_i-\bm{\Sigma}_0^{-1}\bbm{\sigma}_{1i})\\
    = & \mathcal{R}^*(\bm{\Sigma}_0^{-1}[\bm{\Sigma}_0\widehat{\bm{a}}_i-\bbm{\sigma}_{1i}])\\
    = & \mathcal{R}^*(\bm{\Sigma}_0^{-1}[\bm{\Sigma}_0\widehat{\bm{a}}_i-\widetilde{\bm{\Sigma}}_0\widehat{\bm{a}}_i+\widetilde{\bm{\Sigma}}_0\widehat{\bm{a}}_i-\widetilde{\bbm{\sigma}}_{1i}+\widetilde{\bbm{\sigma}}_{1i}-\bbm{\sigma}_{1i}])\\
    \leq & \mathcal{C}(\bm{\Sigma}_0^{-1})\cdot \mathcal{R}^*([\bm{\Sigma}_0-\widetilde{\bm{\Sigma}}_0]\widehat{\bm{a}}_i) + \mathcal{R}_0(\bm{\Sigma}_0^{-1})\cdot\mathcal{R}^*(\widetilde{\bm{\Sigma}}_0\widehat{\bm{a}}_i-\widetilde{\bbm{\sigma}}_1) + \mathcal{C}(\bm{\Sigma}_0^{-1})\cdot\mathcal{R}^*(\widetilde{\bbm{\sigma}}_1-\bbm{\sigma}_1)\\
    \leq & \mathcal{C}(\bm{\Sigma}_0^{-1})\cdot[\zeta_1\mathcal{R}(\widehat{\bm{a}}_i)+\lambda+\zeta_{2i}] \leq \hspace{0.01in}\mathcal{C}(\bm{\Sigma}_0^{-1})\cdot[\zeta_1\mathcal{R}(\bm{a}_i^*)+\lambda+\zeta_{2i}]\leq 2\lambda\hspace{0.01in}\mathcal{C}(\bm{\Sigma}_0^{-1}).
  \end{split} 
\end{equation}

By the definitions of model subspace and perturbation subspace, we can decompose $\bm{a}_i^*=(\bm{a}_i^*)_{\mathcal{M}_i}+(\bm{a}^*_i)_{\overline{\mathcal{M}}^\perp_i}$. Since $\bm{a}_i^*$ is feasible, by decomposability of $\mathcal{R}(\cdot)$, we have
\begin{equation}
  \label{eq:tech2}
  \mathcal{R}(\widehat{\bm{a}}_i) \leq \mathcal{R}(\bm{a}^*) = \mathcal{R}((\bm{a}^*_i)_{\mathcal{M}_i} + (\bm{a}_i^*)_{\overline{\mathcal{M}}_i^\perp}) = \mathcal{R}((\bm{a}^*_i)_{\mathcal{M}_i}) + \mathcal{R}((\bm{a}^*_i)_{\overline{\mathcal{M}}_i^\perp}).
\end{equation}
Moreover, we have
\begin{equation}
  \begin{split}
    \mathcal{R}(\bm{\widehat{a}}_i)=&\mathcal{R}(\bm{a}_i^*+\bbm{\delta}_i)=\mathcal{R}((\bm{a}^*_i)_{\mathcal{M}_i}+(\bbm{\delta}_i)_{\overline{\mathcal{M}}_i^\perp}+(\bm{a}_i^*)_{\overline{\mathcal{M}}_i^\perp}+(\bbm{\delta}_i)_{\overline{\mathcal{M}}_i})\\
    \geq& \mathcal{R}((\bm{a}^*_i)_{\mathcal{M}_i}+(\bbm{\delta}_i)_{\overline{\mathcal{M}}_i^\perp})-\mathcal{R}((\bm{a}^*_i)_{\overline{\mathcal{M}}_i^\perp}+(\bbm{\delta}_i)_{\overline{\mathcal{M}}_i})\\
    \geq &\mathcal{R}((\bm{a}^*_i)_{\mathcal{M}_i}) + \mathcal{R}((\bbm{\delta}_i)_{\overline{\mathcal{M}}_i^\perp}) - \mathcal{R}((\bm{a}^*_i)_{\overline{\mathcal{M}}_i^\perp}) -  \mathcal{R}((\bbm{\delta}_i)_{\overline{\mathcal{M}}_i}),
  \end{split}
\end{equation}
where the first inequality follows from the triangle inequality and the second one follows from the decomposability of $\mathcal{R}(\cdot)$.
Thus, together with \eqref{eq:tech2}, we have
$\mathcal{R}((\bbm{\delta}_i)_{\overline{\mathcal{M}}_i^\perp}) \leq \mathcal{R}((\bbm{\delta}_i)_{\overline{\mathcal{M}}_i}) + 
2\mathcal{R}((\bm{a}^*_i)_{\overline{\mathcal{M}}_i^\perp})$.

Hence, the upper bound in terms of $\mathcal{R}(\cdot)$ can be found as
\begin{equation}
  \begin{split}
    &\mathcal{R}(\bbm{\delta}_i) = \mathcal{R}((\bbm{\delta}_i)_{\overline{\mathcal{M}}_i}+(\bbm{\delta}_i)_{\overline{\mathcal{M}}_i^\perp})\\
    \leq & \mathcal{R}((\bbm{\delta}_i)_{\overline{\mathcal{M}}_i}) + \mathcal{R}((\bbm{\delta}_i)_{\overline{\mathcal{M}}_i^\perp})
    \leq 2\mathcal{R}((\bbm{\delta}_i)_{\overline{\mathcal{M}}_i}) + 2\mathcal{R}((\bm{a}_i^*)_{\overline{\mathcal{M}}_i^\perp})\\
    \leq & 2\phi(\overline{\mathcal{M}}_i)\mathcal{R}^*((\bbm{\delta}_i)_{\overline{\mathcal{M}}_i}) + 2\mathcal{R}(\bm{A}^*_{\overline{\mathcal{M}}_i^\perp})\\ 
    \leq & 2\phi(\overline{\mathcal{M}}_i)\mathcal{R}^*(\bbm{\delta}_i) + 2\mathcal{R}((\bm{a}_i^*)_{\overline{\mathcal{M}}_i^\perp}) \\
    \leq & 4\phi(\overline{\mathcal{M}}_i)\lambda\hspace{0.01in}\mathcal{C}(\bm{\Sigma}_0^{-1}) + 2\mathcal{R}(\bm{A}^*_{\overline{\mathcal{M}}_i^\perp}),
  \end{split}
\end{equation}
for all $i=1,2,\dots,p$.

Finally, by the duality between $\mathcal{R}(\cdot)$ and $\mathcal{R}^*(\cdot)$, we have
\begin{equation}
  \begin{split}
    \|\bbm{\delta}_i\|_2^2\leq \mathcal{R}(\bbm{\delta}_i)\mathcal{R}^*(\bbm{\delta}_i) \leq 8\phi(\overline{\mathcal{M}}_i)\lambda^2\hspace{0.01in}\mathcal{C}^2(\bm{\Sigma}_0^{-1}) + 4\lambda\mathcal{R}(\bm{A}^*_{\overline{\mathcal{M}}_i^\perp})\hspace{0.01in}\mathcal{C}(\bm{\Sigma}_0^{-1}).
  \end{split}
\end{equation}
\end{proof}

\begin{proof}[\textbf{Proof of Proposition \ref{prop:3}}]

Similarly to Proposition \ref{prop:1}, we first show that the true value $\bbm{\theta}^*$ is a feasible solution to the optimization problem. For simplicity, we assume that $\bbm{\gamma}=\bm{0}$. Note that
\begin{equation}
  \begin{split}
    &\|\bm{\widetilde{\Omega}}\bbm{\theta}^*-\bbm{\widetilde{\omega}}\|_\infty=\|\bm{\widetilde{\Omega}}\bbm{\theta}^*-\bm{\Omega}\bbm{\theta}^*+\bbm{\omega}-\bbm{\widetilde{\omega}}\|_\infty\\
    \leq&\|(\bm{\widetilde{\Omega}}-\bm{\Omega})\bbm{\theta}^*\|_\infty+\|\bbm{\omega}-\bbm{\widetilde{\omega}}\|_\infty\\
    \leq&\|\bm{\widetilde{\Omega}}-\bm{\Omega}\|_{1,\infty}\|\bbm{\theta}^*\|_\infty+\|\bbm{\omega}-\bbm{\widetilde{\omega}}\|_\infty\\
    \leq&\zeta_1\|\bbm{\theta}^*\|_\infty+\zeta_2\leq\lambda,
  \end{split}
\end{equation}
where the third last inequality follows from the facts that $\bm{\widetilde{\Omega}}-\bm{\Omega}$ is a symmetric matrix and $\|\bm{M}\bm{v}\|_\infty\leq\|\bm{M}^\top\|_{1,\infty}\|\bm{v}\|_\infty$, for any compatible matrix $\bm{M}$ and vector $\bm{v}$. Therefore, the true value $\bbm{\theta}^*$ is feasible and $\|\bbm{\widehat{\theta}}\|_\infty\leq\|\bbm{\theta}^*\|_\infty$. Denote $\bbm{\delta}=\bbm{\widehat{\theta}}-\bbm{\theta}^*$, and we have that
\begin{equation}
  \begin{split}
    &\|\bbm{\delta}\|_\infty=\|\bm{\Omega}^{-1}\bm{\Omega}\bbm{\delta}\|_\infty\\
    \leq&\|\bm{\Omega}^{-1}\|_{1,\infty}\|\bm{\Omega}(\bbm{\widehat{\theta}}-\bbm{\theta}^*)\|_\infty\\
    =&\|\bm{\Omega}^{-1}\|_{1,\infty}\|\bm{\Omega}\bbm{\widehat{\theta}}-\bbm{\omega}\|_\infty\\
    =&\|\bm{\Omega}^{-1}\|_{1,\infty}\Big\{\|(\bm{\Omega}-\bm{\widetilde{\Omega}})\bbm{\widehat{\theta}}+\bm{\widetilde{\Omega}}\bbm{\widehat{\theta}}-\bbm{\widetilde{\omega}}+\bbm{\widetilde{\omega}}-\bbm{\omega}\|_\infty\Big\}\\
    \leq&\|\bm{\Omega}^{-1}\|_{1,\infty}\Big\{\|(\bm{\Omega}-\bm{\widetilde{\Omega}})\|_{1,\infty}\|\bbm{\widehat{\theta}}\|_\infty+\|\bm{\widetilde{\Omega}}\bbm{\widehat{\theta}}-\bbm{\widetilde{\omega}}\|_\infty+\|\bbm{\widetilde{\omega}}-\bbm{\omega}\|_\infty\Big\}\\
    \leq&\|\bm{\Omega}^{-1}\|_{1,\infty}\Big\{\|\bm{\Omega}-\bm{\widetilde{\Omega}}\|_{1,\infty}\|\bbm{\theta}^*\|_\infty+\|\bm{\widetilde{\Omega}}\bbm{\widehat{\theta}}-\bbm{\widetilde{\omega}}\|_\infty+\|\bbm{\widetilde{\omega}}-\bbm{\omega}\|_\infty\Big\}\\
    \leq&\|\bm{\Omega}^{-1}\|_{1,\infty}(\zeta_1\|\bbm{\theta}^*\|_\infty+\lambda+\zeta_2)\leq2\lambda\|\bm{\Omega}^{-1}\|_{1,\infty},
  \end{split}
\end{equation}
where the second equality follows from the linear restricted Yule--Walker equation $\bm{\Omega}\bbm{\theta}^*=\bbm{\omega}$, the third last inequality follows from the feasibility condition $\|\bbm{\widehat{\theta}}\|_\infty\leq \|\bbm{\theta}^*\|_\infty$, and the second last inequality follows from the condition $\lambda\geq\zeta_1\|\bbm{\theta}^*\|_\infty+\zeta_2$.
\end{proof}

\subsection{Proofs of Propositions \ref{prop:element}--\ref{prop:linear}}\label{sec:A.2}

\begin{proof}[\textbf{Proof of Proposition \ref{prop:element}}]
In this proof, we simplify the notation $y_{it}^\text{E}(\tau)$ to $y_{it}(\tau)$ and consider the finite $(2+2\epsilon)$-th moment condition for $\bm{y}_t$, for any $\epsilon\in(0,1]$, namely $\max_{1\leq i\leq p}\mathbb{E}[|y_{it}|^{2+2\epsilon}]= M_{2+2\epsilon}$.

For $1\leq i,j\leq p$ and $\ell=0,1,\dots,d$, we can bound the bias introduced by the data truncation as 
\begin{equation}
  \label{eq:bias_1}
  \begin{split}
    &\mathbb{E}[y_{it}y_{j,t+\ell}]-\mathbb{E}[\widetilde{y}_{it}(\tau)\widetilde{y}_{j,t+\ell}(\tau)]\\
    \leq& \mathbb{E}[|y_{it}y_{j,t+\ell}|(1\{|y_{it}|\geq \tau\}+1\{|y_{j,t+\ell}|\geq \tau\})]\\
    \overset{(1)}{\leq}&\left[(\mathbb{E}|y_{it}|^{2+2\epsilon})(\mathbb{E}|y_{j,t+\ell}|^{2+2\epsilon})\right]^{1/(2+2\epsilon)}\left[\mathbb{P}(|y_{it}|\geq \tau)^{\epsilon/(1+\epsilon)}+\mathbb{P}(|y_{j,t+\ell}|\geq \tau)^{\epsilon/(1+\epsilon)}\right]\\
    \overset{(2)}{\leq}& M_{2+2\epsilon}^{1/(1+\epsilon)}\left[\left(\frac{\mathbb{E}|y_{it}|^{2+2\epsilon}}{\tau^{2+2\epsilon}}\right)^{\epsilon/(1+\epsilon)} + \left(\frac{\mathbb{E}|y_{j,t+\ell}|^{2+2\epsilon}}{\tau^{2+2\epsilon}}\right)^{\epsilon/(1+\epsilon)}\right]\\
    \leq& \frac{2M_{2+2\epsilon}}{\tau^{2\epsilon}} \overset{(3)}{\asymp} \left[\frac{M_{2+2\epsilon}^{1/\epsilon}\log(p^2d)}{n_\text{eff}}\right]^{\epsilon/(1+\epsilon)},
  \end{split}
\end{equation}
where the steps (1) to (3) follow from H\"{o}lder's inequality, Markov's inequality, and the choice of the truncation parameter
\begin{equation}
  \tau\asymp\left[\frac{M_{2+2\epsilon}n_\text{eff}}{\log(p^2d)}\right]^{\frac{1}{2+2\epsilon}},
\end{equation}
respectively.

For the truncated $\widetilde{\bm{y}}_{t}(\tau)$, by the element-wise truncation threshold $\tau$ and the $(2+2\epsilon)$-th moment bound $M_{2+2\epsilon}$, it can be checked that 
\begin{equation}
  \label{eq:variance_eq}
  \begin{split}
    &\mathbb{E}\left[|y_{it}(\tau)y_{j,t+\ell}(\tau)|^2\right]\leq \tau^{2-2\epsilon}\cdot\mathbb{E}\left[|y_{it}(\tau)y_{j,t+\ell}(\tau)|^{1+\epsilon}\right]\\
    \overset{(1)}{\leq} & \tau^{2-2\epsilon}\cdot\sqrt{\mathbb{E}[|y_{it}(\tau)|^{2+2\epsilon}]\mathbb{E}[|y_{j,t+\ell}|^{2+2\epsilon}]}\leq \tau^{2-2\epsilon}M_{2+2\epsilon},
  \end{split}
\end{equation}
where the inequality (1) follows from the Cauchy--Schwarz inequality. Also, for any $k=3,4,\dots$, the higher-order moments satisfy that
\begin{equation}
  \label{eq:moment_eq}
    \mathbb{E}\left[|y_{it}(\tau)y_{j,t+\ell}(\tau)|^k\right]
    \leq \tau^{2(k-2)} \mathbb{E}\left[|y_{it}(\tau)y_{j,t+\ell}(\tau)|^{2}\right].
\end{equation}

Since $\bm{\widetilde{y}}_t(\tau)$ is a measurable function of $\bm{y}_t$, by Lemma \ref{lemma:mixing}, $\widetilde{\bm{y}}_t(\tau)$ is also $\alpha$-mixing with the mixing coefficient smaller than or equal to $\alpha(\ell)\leq Cr^\ell$ for all $\ell\geq0$. As the lag order $d$ is fixed, by Lemma \ref{lemma:lag_mixing}, the lagged values $\widetilde{\bm{x}}_t=(\widetilde{\bm{y}}_{t-1}^\top,\dots,\widetilde{\bm{y}}_{t-d}^\top)^\top$ is also $\alpha$-mixing with the mixing coefficient decaying  geometrically.
Therefore, by \eqref{eq:variance_eq}, \eqref{eq:moment_eq} and the Bernstein-type inequality for the $\alpha$-mixing sequence in Lemma \ref{lemma:alpha-mixing},
\begin{equation}
  \begin{split}
    &\mathbb{P}\left(\left|\frac{1}{T}\sum_{t=1}^{T}y_{it}(\tau)y_{j,t+\ell}(\tau)-\mathbb{E}[y_{it}(\tau)y_{j,t+\ell}(\tau)]\right|\geq\varepsilon\right)\\
    \leq& C\left[\mu(\varepsilon)+\log(T)\right]\exp\left[-C\frac{\mu(\varepsilon)T}{\log(T)}\right],
  \end{split}
\end{equation}
where
\begin{equation}
  \mu(\varepsilon)=\frac{\varepsilon^2}{\tau^2\varepsilon+\tau^{2-2\epsilon}M_{2+2\epsilon}}.
\end{equation}
Let $\varepsilon=CM_{2+2\epsilon}^{1/(1+\epsilon)}[\log(p)/n_\text{eff}]^{\epsilon/(1+\epsilon)}$. Then, $\mu(\varepsilon)\ll\log(T)$ and 
\begin{equation}
  \begin{split}
    &\mathbb{P}\left(\left|\frac{1}{T}\sum_{t=1}^{T}y_{j,t}(\tau)y_{k,t+\ell}(\tau)-\mathbb{E}[y_{j,t}(\tau)y_{k,t+\ell}(\tau)]\right|\geq C\left[\frac{M_{2+2\epsilon}^{1/\epsilon}\log(p^2d)}{n_\text{eff}}\right]^{\frac{\epsilon}{1+\epsilon}}\right)\\
    \leq & C\log(T)\exp\left[-Cn_\text{eff}\log(T)\cdot\frac{\log(p^2d)}{n_\text{eff}}\right]= C\exp\left[\log\log(T)-C\log(p^2d)\log(T)\right]\\
    \leq & C\exp[-C\log(p^2d)\log(T)]
  \end{split}
\end{equation}
For all $1\leq i,j\leq p$ and $\ell=0,\dots,d$,
\begin{equation}
  \label{eq:deviation_1}
  \begin{split}
    & \mathbb{P}\left(\max_{\substack{1\leq j,k\leq p\\ \ell=0,\dots,d}}\left|\frac{1}{T}\sum_{t=1}^Ty_{it}(\tau)y_{j,t+\ell}(\tau)-\mathbb{E}[y_{it}(\tau)y_{j,t+\ell}(\tau)]\right| \geq C\left[\frac{M_{2+2\epsilon}^{1/\epsilon}\log(p^2d)}{n_\text{eff}}\right]^{\frac{\epsilon}{1+\epsilon}}\right) \\
    \leq&\sum_{\substack{1\leq j,k\leq p\\ \ell=0,\dots,d}}\mathbb{P}\left(\left|\frac{1}{T}\sum_{t=1}^Ty_{it}(\tau)y_{j,t+\ell}(\tau)-\mathbb{E}[y_{it}(\tau)y_{j,t+\ell}(\tau)]\right| \geq C\left[\frac{M_{2+2\epsilon}^{1/\epsilon}\log(p^2d)}{n_\text{eff}}\right]^{\frac{\epsilon}{1+\epsilon}}\right)\\
    \leq&Cp^2d\exp[-C\log(p^2d)\log(T)]= C\exp[\log(p^2d)-C\log(p^2d)\log(T)]\\
    \leq & C\exp[-C\log(p^2d)\log(T)]
  \end{split}
\end{equation}

Finally, the $\ell_\infty$ norm bounds for $\widetilde{\bm{\Sigma}}_0$ and $\widetilde{\bm{\Sigma}}_1$ can be obtained by combining the high probability deviation bound in \eqref{eq:bias_1} and the bound for the truncation bias in \eqref{eq:deviation_1}.
\end{proof}

\begin{proof}[\textbf{Proof of Proposition \ref{prop:vector}}]

Throughout this proof, we simplify the notations $\bm{y}_t^\text{V}(\tau_1)$ and $\bm{x}_t^\text{V}(\tau_2)$ to $\bm{y}_t(\tau_1)$ and $\bm{x}_t(\tau_2)$, respectively. For any vector $\bm{v}=(v_1,v_2,\dots,v_p)^\top\in\mathbb{R}^p$ and $\epsilon\in(0,1]$, by H\"{o}lder's inequality,
\begin{equation}
  \begin{split}
    \|\bm{v}\|_2^2=\sum_{i=1}^pv_i^2\leq\left(\sum_{i=1}^pv_i^{2+2\epsilon}\right)^{\frac{1}{1+\epsilon}}p^{\frac{\epsilon}{1+\epsilon}}=p^{\frac{\epsilon}{1+\epsilon}}\|\bm{v}\|_{2+2\epsilon}^2.
  \end{split}
\end{equation}
and
\begin{equation}
  \begin{split}
    \|\bm{v}\|_2^2=\sum_{i=1}^pv_i^2\leq\left(\sum_{i=1}^pv_i^{2+2\epsilon+\delta}\right)^{\frac{1}{1+\epsilon+\delta/2}}p^{\frac{\epsilon+\delta/2}{1+\epsilon+\delta/2}}\|\bm{v}\|_{2+2\epsilon+\delta}^2.
  \end{split}
\end{equation}

First, we prove the operator norm of
$\bm{\widetilde{\Sigma}}_0^\text{V}(\tau_2)-\bm{\Sigma}_0$. 
For any $\bm{v}\in\mathbb{R}^{pd}$ such that $\|\bm{v}\|_2=1$, it holds that
\begin{equation}
  \label{eq:bias2}
  \begin{split}
    &\mathbb{E}[\bm{v}^\top(\bm{x}_t\bm{x}_{t}^\top-\bm{x}_t(\tau_2)\bm{x}_t^\top(\tau_2))\bm{v}]\\
    =&\mathbb{E}[((\bm{v}^\top\bm{x}_t)^2-(\bm{v}^\top\bm{x}_t(\tau_2))^2)1\{\|\bm{x}_t\|_{2}>\tau_2\}]\\
    \leq&\mathbb{E}[(\bm{v}^\top\bm{x}_t)^21\{\|\bm{x}_t\|_{2}>\tau_2\}]\\
    \overset{(1)}{\leq}&[\mathbb{E}(\bm{v}^\top\bm{x}_t)^{2+2\epsilon}]^{1/(1+\epsilon)}\cdot\mathbb{P}[\|\bm{x}_t\|_{2}>\tau_2]^{\epsilon/(1+\epsilon)}\\
    \overset{(2)}{\leq}& M_{2+2\epsilon}^{1/(1+\epsilon)}\left(\frac{\mathbb{E}\|\bm{x}_t\|_{2}^{2+2\epsilon}}{\tau_2^{2+2\epsilon}}\right)^{\frac{\epsilon}{1+\epsilon}} \leq  M_{2+2\epsilon} \cdot (\tau_2^{-2}pd)^{\epsilon}\\
    \overset{(3)}{\asymp} &\left[\frac{pdM_{2+2\epsilon}^{1/\epsilon}}{n_\text{eff}}\right]^{\epsilon/(1+\epsilon)},
  \end{split}
\end{equation}
where the steps (1) to (3) follow from H\"{o}lder's inequality, Markov's inequality, and the choice $\tau_2\asymp \left[(pd)^\epsilon M_{2+2\epsilon}n_\text{eff}\right]^{1/(2+2\epsilon)}$, respectively

Denote $\bm{\widetilde{X}}_t=\bm{x}_t(\tau_2)\bm{x}_t(\tau_2)^\top$, where $\tau_2$ is omitted for simplicity. By H\"{o}lder's inequality,
\begin{equation}
  \begin{split}
    & \|\widetilde{\bm{X}}_t-\mathbb{E}\widetilde{\bm{X}}_t\|_\textup{op} \leq \|\widetilde{\bm{X}}_t\|_\textup{op}+\|\mathbb{E}\widetilde{\bm{X}}_t\|_\textup{op}\\
    = & \|\bm{x}_{t}(\tau_2)\|_2^2 + \sup_{\bm{v}\in\mathbb{S}^{pd}}\mathbb{E}\left[(\bm{v}^\top\bm{x}_t(\tau_2))^2\right]\\
    \leq & \tau_2^2 + \sup_{\bm{v}\in\mathbb{S}^{pd}}\mathbb{E}\left[(\bm{v}^\top\bm{x}_t(\tau))^{2+2\epsilon}\right]^{\frac{1}{1+\epsilon}} \leq \tau_2^2 + M_{2+2\epsilon}^{1/(1+\epsilon)}\\
    \asymp & (pd)^{\epsilon/(1+\epsilon)}M_{2+2\epsilon}^{1/(1+\epsilon)}n_\text{eff}^{1/(1+\epsilon)}:=B.
  \end{split}
\end{equation}
For any $\bm{v}\in\mathbb{S}^{pd}$, also by H\"{o}lder's inequality,
\begin{equation}
  \begin{split}
    &\mathbb{E}(\bm{v}^\top\widetilde{\bm{X}}_t\widetilde{\bm{X}}_t\bm{v}) = \mathbb{E}\left[\|\bm{x}_t(\tau_2)\|_2^2(\bm{v}^\top\bm{x}_t(\tau_2))^2\right]\\
    \leq & \mathbb{E}\left[\|\bm{x}_t(\tau_2)\|_2^{2-2\epsilon}\cdot \|\bm{x}_t(\tau_2)\|_2^{2\epsilon}(\bm{v}^\top\bm{x}_t(\tau_2))^2\right]\\
    \leq & \tau_2^{2-2\epsilon}\cdot\mathbb{E}[\|\bm{x}_t\|_2^{2+2\epsilon}]^{\frac{\epsilon}{1+\epsilon}}\cdot\mathbb{E}[(\bm{v}^\top\bm{x}_t)^{2+2\epsilon}]^{\frac{1}{1+\epsilon}}\\
    \leq & \tau_2^{2-2\epsilon}\cdot (pd)^\epsilon\cdot M_{2+2\epsilon}.
  \end{split}
\end{equation}
It follows that $\|\mathbb{E}\widetilde{\bm{X}}_t\widetilde{\bm{X}}_t\|_\textup{op}\leq \tau_2^{2-2\epsilon}\cdot (pd)^\epsilon \cdot M_{2+2\epsilon}$.

Since $\|\mathbb{E}(\widetilde{\bm{X}}_t)\mathbb{E}(\widetilde{\bm{X}}_t)\|_\textup{op}\leq \|\mathbb{E}\bm{x}_t(\tau_2)\bm{x}_t(\tau_2)^\top\|_\textup{op}^2=[\sup_{\bm{v}\in\mathbb{S}^{pd}}\mathbb{E}(\bm{v}^\top\bm{x}_t(\tau_2))^2]^2\leq M_{2+2\epsilon}^{2/(1+\epsilon)}$, we have that 
\begin{equation}
  \begin{split}
    & \|\mathbb{E}[(\widetilde{\bm{X}}_t-\mathbb{E}\widetilde{\bm{X}}_t)(\widetilde{\bm{X}}_t-\mathbb{E}\widetilde{\bm{X}}_t)]\|_\textup{op} \leq \|\mathbb{E}\widetilde{\bm{X}}_t\widetilde{\bm{X}}_t\|_\textup{op} + \|\mathbb{E}(\widetilde{\bm{X}}_t)\mathbb{E}(\widetilde{\bm{X}}_t)\|_\textup{op}\\
    \leq & \tau_2^{2-2\epsilon}\cdot (pd)^\epsilon\cdot M_{2+2\epsilon} + M_{2+2\epsilon}^{2/(1+\epsilon)}\\
    \asymp & (pd)^{\frac{2\epsilon}{1+\epsilon}}\cdot M_{2+2\epsilon}^{2/(1+\epsilon)}\cdot n_\text{eff}^{(1-\epsilon)/(1+\epsilon)} := v^2.
  \end{split}
\end{equation}

By Lemmas \ref{lemma:mixing} and \ref{lemma:lag_mixing}, as $\widetilde{\bm{X}}_t$ can be viewed as a function of $\bm{x}_t$, $\widetilde{\bm{X}}_t$ is a $\beta$-mixing sequence with mixing coefficients bounded by those of $\bm{y}_t$, where the definition of the $\beta$-mixing condition for a time series of random vectors can readily be extended to a time series of random matrices.
For any $\bm{v}\in\mathbb{S}^{pd}$, $\bm{u}\in\mathbb{S}^{pd}$, $\ell\geq 1$, by Lemma \ref{lemma:covariance},
\begin{equation}
  \begin{split}
    &\mathbb{E}[\bm{v}^\top\widetilde{\bm{X}}_t\widetilde{\bm{X}}_{t+\ell}\bm{u}]\\
    =&\mathbb{E}\left[(\bm{v}^\top\bm{x}_t(\tau_2))(\bm{u}^\top\bm{x}_{t+\ell}(\tau_2))\bm{x}_t(\tau_2)^\top\bm{x}_{t+\ell}(\tau_2)\right]\\
    \leq & \tau_2^{2-2\epsilon}\cdot\mathbb{E}\left[(\bm{v}^\top\bm{x}_t(\tau_2))(\bm{u}^\top\bm{x}_{t+\ell}(\tau))[\bm{x}_t(\tau_2)^\top\bm{x}_{t+\ell}(\tau_2)]^\epsilon\right]\\
    \leq & \tau_2^{2-2\epsilon}\cdot\beta(\ell)^{\delta/(4+\delta)}\cdot \mathbb{E}[\|\bm{x}_t\|_2^{2+2\epsilon+\delta}]^{\frac{\epsilon}{2+2\epsilon+\delta}}
    \cdot\mathbb{E}[\|\bm{x}_{t+\ell}\|_2^{2+2\epsilon+\delta}]^{\frac{\epsilon}{2+2\epsilon+\delta}}\cdot\\
    &\mathbb{E}[(\bm{v}^\top\bm{x}_t)^{2+2\epsilon}]^{\frac{1}{2+2\epsilon+\delta}}\cdot
    \mathbb{E}[(\bm{v}^\top\bm{x}_{t+\ell})^{2+2\epsilon}]^{\frac{1}{2+2\epsilon+\delta}}\\
    \leq & \tau_2^{2-2\epsilon}\cdot (pd)^\epsilon\cdot \beta(\ell)^{\delta/(2+2\epsilon+\delta)}\cdot M_{2+2\epsilon}.
  \end{split}
\end{equation}
Thus, we have that
\begin{equation}
  \|\mathbb{E}[(\widetilde{\bm{X}}_t-\mathbb{E}\widetilde{\bm{X}}_t)(\widetilde{\bm{X}}_{t+\ell}-\mathbb{E}\widetilde{\bm{X}}_t)]\|_\textup{op}\leq C\cdot \tau_2^{2-2\epsilon} \cdot M_{2+2\epsilon}\cdot p^\epsilon d^\epsilon\cdot \exp[-(\delta/(2+2\epsilon+\delta))\ell]. 
\end{equation} 

For any $K\subset\{1,\dots, T\}$,
\begin{equation}
  \begin{split}
    &\frac{1}{\text{Card}(K)}\lambda_{\max}\left\{\mathbb{E}\left(\sum_{t\in K}(\widetilde{\bm{X}}_t-\mathbb{E}\widetilde{\bm{X}}_t)\right)^2\right\} \\
    \leq&\frac{1}{\text{Card}(K)}\left\|\sum_{i,j\in K}\mathbb{E}(\bm{\widetilde{X}}_i-\mathbb{E}\bm{\widetilde{X}}_i)(\bm{\widetilde{X}}_j-\mathbb{E}\bm{\widetilde{X}}_j)\right\|_\textup{op}\\
    \leq&\frac{1}{\text{Card}(K)}\sum_{i,j\in K}\left\|\mathbb{E}(\bm{\widetilde{X}}_i-\mathbb{E}\bm{\widetilde{X}}_i)(\bm{\widetilde{X}}_j-\mathbb{E}\bm{\widetilde{X}}_j)\right\|_\textup{op}\\
    =&\left\|\mathbb{E}(\bm{\widetilde{X}}_t-\mathbb{E}\bm{\widetilde{X}}_t)^2\right\|_\textup{op}+\frac{2}{\text{Card}(K)}\sum_{i>j,i,j\in K}\left\|\mathbb{E}(\bm{\widetilde{X}}_i-\mathbb{E}\bm{\widetilde{X}}_i)(\bm{\widetilde{X}}_j-\mathbb{E}\bm{\widetilde{X}}_j)\right\|_\textup{op}\\
    \leq&C\cdot \tau_2^{2-2\epsilon}\cdot M_{2+2\epsilon}\cdot p^\epsilon d^\epsilon\cdot[1-\exp(-\delta/(2+2\epsilon+\delta))]^{-1} \asymp v^2.
  \end{split}
\end{equation}
Thus, by the $\beta$-mixing matrix Bernstein-type inequality in Lemma \ref{lemma:beta-mixing}, we have that for any $\varepsilon>0$,
\begin{equation}
  \begin{split}
    &\mathbb{P}\left(\left\|\frac{1}{T}\sum_{t=1}^T\widetilde{\bm{X}}_t-\mathbb{E}\widetilde{\bm{X}}_t\right\|_\textup{op} \geq \varepsilon\right)\\
    \leq & 2\mathbb{P}\left(\lambda_{\max}\left(\sum_{t=1}^T\widetilde{\bm{X}}_t-\mathbb{E}\widetilde{\bm{X}}_t\right) \geq T\varepsilon\right)\\
    \leq & 2pd\exp\left(-\frac{CT^2\varepsilon^2}{v^2 T+B^2+T\varepsilon\log(T)^2B}\right).
  \end{split}
\end{equation}
Letting $\varepsilon=C(pd M_{2+2\epsilon}^{1/\epsilon}n_\text{eff}^{-1})^{\epsilon/(1+\epsilon)}$, if $T\gtrsim pd$, we have
\begin{equation}
  \begin{split}
    &\mathbb{P}\left(\left\|\frac{1}{T}\sum_{t=1}^T\widetilde{\bm{X}}_t-\mathbb{E}\widetilde{\bm{X}}_t\right\|_\textup{op} \geq C\left[\frac{pdM_{2+2\epsilon}^{1/\epsilon}}{n_\text{eff}}\right]^{\epsilon/(1+\epsilon)}\right)\\
    \leq & 2\exp[\log(pd)-C\log(T)]\leq 2\exp[-C\log(T)].
  \end{split}
\end{equation}
Combining this tail probability and the bound of truncation bias in \eqref{eq:bias2}, we have that with probability at least $1-2\exp[C\log(T)]$,
\begin{equation}
  \|\widetilde{\bm{\Sigma}}_0^\text{V}(\tau_2)-\bm{\Sigma}_0\|_\textup{op}\lesssim \left[\frac{pdM_{2+2\epsilon}^{1/\epsilon}}{n_\text{eff}}\right]^{\epsilon/(1+\epsilon)}.
\end{equation}

Next we prove the upper bound for the operator norm of $\widetilde{\bm{\Sigma}}_1^\text{V}(\tau_1,\tau_2)-\bm{\Sigma}_1$. For any $\bm{v}\in\mathbb{S}^{p}$ and $\bm{u}\in\mathbb{S}^{pd}$, similarly to $\widetilde{\bm{\Sigma}}_0^\text{V}(\tau_2)$, it holds that
\begin{equation}
  \begin{split}
    & \mathbb{E}[\bm{v}^\top(\bm{y}_t\bm{x}_{t}^\top-\bm{y}_t(\tau_1)\bm{x}_t(\tau_2)^\top)\bm{u}]\\
    \leq & \mathbb{E}[|(\bm{v}^\top\bm{y}_t)(\bm{u}^\top\bm{x}_{t})|\cdot(1\{\|\bm{y}_t\|_2>\tau_1\}+1\{\|\bm{x}_{t}\|_2>\tau_2\})]\\
    \leq & \mathbb{E}[|(\bm{v}^\top\bm{y}_t)(\bm{u}^\top\bm{x}_{t})|^{1+\epsilon}]^{1/(1+\epsilon)}\cdot \mathbb{P}[\|\bm{y}_t\|_2>\tau_1]^{\epsilon/(1+\epsilon)}\\
    + & \mathbb{E}[|(\bm{v}^\top\bm{y}_t)(\bm{u}^\top\bm{x}_{t})|^{1+\epsilon}]^{1/(1+\epsilon)}\cdot \mathbb{P}[\|\bm{x}_t\|_2>\tau_2]^{\epsilon/(1+\epsilon)}\\
    \leq & M_{2+2\epsilon}^{1/(1+\epsilon)}\left(\frac{\mathbb{E}\|\bm{y}_t\|_2^{2+2\epsilon}}{\tau_1^{2+2\epsilon}}\right)^{\epsilon/(1+\epsilon)} + M_{2+2\epsilon}^{1/(1+\epsilon)} \left(\frac{\mathbb{E}\|\bm{x}_t\|_2^{2+2\epsilon}}{\tau_2^{2+2\epsilon}}\right)^{\epsilon/(1+\epsilon)} \\
    \leq & M_{2+2\epsilon}\cdot\tau_1^{-2\epsilon}\cdot p^\epsilon + M_{2+2\epsilon}\cdot\tau_2^{-2\epsilon}\cdot p^\epsilon d^\epsilon\\
    \lesssim & \left[\frac{pdM_{2+2\epsilon}^{1/\epsilon}}{n_\text{eff}}\right]^{\epsilon/(1+\epsilon)},
  \end{split}
\end{equation}
with $\tau_1\asymp [p^\epsilon M_{2+2\epsilon}n_\text{eff}]^{1/(2+2\epsilon)}$ and $\tau_2\asymp \left[(pd)^\epsilon M_{2+2\epsilon}n_\text{eff}\right]^{1/(2+2\epsilon)}$.

Denote the symmetrized version of $\bm{y}_t(\tau_1)\bm{x}_{t}(\tau_2)^\top$ by $\widetilde{\bm{Z}}_t$ with the parameters $\tau_1$ and $\tau_2$ omitted
\begin{equation}
  \widetilde{\bm{Z}}_t=\begin{pmatrix}
    \bm{O}_{p\times p} & \bm{y}_t(\tau_1)\bm{x}_{t}(\tau_2)^\top\\
    \bm{x}_t(\tau_2)\bm{y}_{t}(\tau_1)^\top & \bm{O}_{pd\times pd}
  \end{pmatrix}.
\end{equation}
Note that $\|\bm{y}_t(\tau_1)\bm{x}_t(\tau_2)^\top\|_\textup{op}=\lambda_{\max}(\widetilde{\bm{Z}}_t)$ and
\begin{equation}
  \begin{split}
    &\|\widetilde{\bm{Z}}_t-\mathbb{E}\widetilde{\bm{Z}}_t\|_\textup{op} \leq \|\widetilde{\bm{Z}}_t\|_\textup{op} + \|\mathbb{E}\widetilde{\bm{Z}}_t\|_\textup{op}\\
    \leq&\|\bm{y}_t(\tau_1)\|_2\|\bm{x}_t(\tau_2)\|_2+\sup_{\bm{v}\in\mathbb{S}^{p},\bm{u}\in\mathbb{S}^{p}}\mathbb{E}[(\bm{v}^\top\widetilde{\bm{y}}_t(\tau))(\bm{u}^\top\widetilde{\bm{y}}_{t-1}(\tau))]\\
    \leq & \tau_1\tau_2 + \sup_{\bm{v}\in\mathbb{S}^{p}}\mathbb{E}[(\bm{v}^\top\bm{y}_t(\tau_1))^{2+2\epsilon}]^{1/(2+2\epsilon)}\sup_{\bm{u}\in\mathbb{S}^{pd}}\mathbb{E}[(\bm{u}^\top\bm{x}_{t}(\tau_2))^{2+2\epsilon}]^{1/(2+2\epsilon)}\\
    \leq & \tau_1\tau_2 + M_{2+2\epsilon}^{1/(1+\epsilon)} \asymp B.
  \end{split}
\end{equation}

For any $\bm{w}\in\mathbb{S}^{p(d+1)}$, split it to $\bm{w}=(\bm{w}_1^\top,\bm{w}_2^\top)^\top$, where $\bm{w}_1\in\mathbb{R}^p$ and $\bm{w}_2\in\mathbb{R}^{pd}$, and similarly to $\widetilde{\bm{X}}_t\widetilde{\bm{X}}_t$,
\begin{equation}
  \begin{split}
    &\mathbb{E}(\bm{w}^\top\widetilde{\bm{Z}}_t\widetilde{\bm{Z}}_t\bm{w})\\
    =&\mathbb{E}(\bm{w}_1^\top\bm{y}_t(\tau_1)\bm{x}_t(\tau_2)^\top\bm{x}_t(\tau_2)\bm{y}_{t}(\tau_1)^\top\bm{w}_1) + \mathbb{E}(\bm{w}_2^\top\bm{x}_{t}(\tau_2)\bm{y}_{t}(\tau_1)^\top\bm{y}_{t}(\tau_1)\bm{y}_{t}(\tau_2)^\top\bm{w}_2)\\
    \leq & \mathbb{E}[\|\bm{x}_t(\tau_2)\|_2^{2-2\epsilon}\cdot(\|\bm{x}_t(\tau_2)\|_2^{2\epsilon}(\bm{w}_1^\top\bm{y}_t)^2)] + 
    \mathbb{E}[\|\bm{y}_{t}(\tau_1)\|_2^{2-2\epsilon}\cdot(\|\bm{y}_{t}(\tau_1)\|_2^{2\epsilon}(\bm{w}_2^\top\bm{x}_{t})^2)]\\
    \leq & \tau_2^{2-2\epsilon}\cdot (pd)^\epsilon\cdot M_{2+2\epsilon}\cdot\|\bm{w}_1\|_2^2 + \tau_1^{2-2\epsilon}\cdot p^\epsilon\cdot M_{2+2\epsilon}\cdot\|\bm{w}_2\|_2^2\\
    \leq & (\tau_1^{2-2\epsilon}+\tau_2^{2-2\epsilon})\cdot (pd)^\epsilon\cdot M_{2+2\epsilon}.
  \end{split}
\end{equation}
Since $\|\mathbb{E}(\widetilde{\bm{Z}}_t)\mathbb{E}(\widetilde{\bm{Z}}_t)\|_\textup{op}\leq \|\mathbb{E}(\widetilde{\bm{Z}}_t)\|_\textup{op}^2\leq M_{2+2\epsilon}^{2/(1+\epsilon)}$, we have that
\begin{equation}
  \|\mathbb{E}[(\widetilde{\bm{Z}}_t-\mathbb{E}\widetilde{\bm{Z}}_t)(\widetilde{\bm{Z}}_t-\mathbb{E}\widetilde{\bm{Z}}_t)]\|_\textup{op} \leq (\tau_1^{2-2\epsilon}+\tau_2^{2-2\epsilon})\cdot (pd)^\epsilon\cdot M_{2+2\epsilon} + M_{2+2\epsilon}^{2/(1+\epsilon)}.
\end{equation}

For any $\bm{v}=(\bm{v}_1^\top,\bm{v}_2^\top)^\top\in\mathbb{S}^{p(d+1)}$, $\bm{u}=(\bm{u}_1^\top,\bm{u}_2^\top)^\top\in\mathbb{S}^{p(d+1)}$, $\ell\geq 1$, by Lemma \ref{lemma:covariance} and similar techniques used for $\widetilde{\bm{X}}_t\widetilde{\bm{X}}_t$,
\begin{equation}
  \begin{split}
    &\mathbb{E}[\bm{v}^\top\widetilde{\bm{Z}}_t\widetilde{\bm{Z}}_{t+\ell}\bm{u}]\\
    \leq & \mathbb{E}[\bm{v}_1^\top\bm{y}_{t}(\tau_1)\bm{x}_{t}(\tau_2)^\top\bm{x}_{t+\ell}(\tau_2)\bm{y}_{t+\ell}(\tau_1)^\top\bm{u}_1] + \mathbb{E}[\bm{v}_2^\top\bm{x}_{t}(\tau_2)\bm{y}_t(\tau_1)^\top\bm{y}_{t+\ell}(\tau_1)\bm{x}_{t+\ell}(\tau_2)^\top\bm{u}_2]\\
    \leq & (\tau_1^{2-2\epsilon}+\tau_2^{2-2\epsilon})\cdot (pd)^\epsilon\cdot\beta(\ell)^{\delta/(2+2\epsilon+\delta)}\cdot M_{2+2\epsilon},
  \end{split}
\end{equation}
We can also show that with probability at least $1-2\exp[-C\log(T)]$,
\begin{equation}
  \left\|\frac{1}{T}\sum_{t=1}^T\widetilde{\bm{Z}}_t-\mathbb{E}\widetilde{\bm{Z}}_t\right\|_\textup{op}\lesssim \left[\frac{pdM_{2+2\epsilon}^{1/\epsilon}}{n_\text{eff}}\right]^{\epsilon/(1+\epsilon)}
\end{equation}
and hence,
\begin{equation}
  \|\widetilde{\bm{\Sigma}}_1^\text{V}(\tau_1,\tau_2)-\bm{\Sigma}_1\|_\textup{op}\lesssim\left[\frac{pdM_{2+2\epsilon}^{1/\epsilon}}{n_\text{eff}}\right]^{\epsilon/(1+\epsilon)}.
\end{equation}

\end{proof}

\begin{proof}[\textbf{Proof of Proposition \ref{prop:linear}}]
  By Assumption \ref{asmp:moment_linear}, for $i=1,\dots,r$, $\mathbb{E}[\|\bm{w}_{it}\|_2^{2+2\epsilon}]\leq M_{1,2+2\epsilon}$, where the moment bound $M_{1,2+2\epsilon}$ possibly depends on the dimension $p$.

  First, by the Cauchy--Schwarz inequality and Markov's inequailty, we can bound the bias from data truncation, for $1\leq i,j\leq r$,
  \begin{equation}
    \begin{split}
      & |\mathbb{E}[\bm{w}_{it}^\top\bm{w}_{jt}]-\mathbb{E}[\widetilde{\bm{w}}_{it}(\tau_1)^\top\widetilde{\bm{w}}_{jt}(\tau_1)]|\\
      \leq & \mathbb{E}[|\bm{w}_{it}^\top\bm{w}_{jt}|^{1+\epsilon}]^{1/(1+\epsilon)}\cdot[\mathbb{P}(\|\bm{w}_{it}\|_2\geq\tau_1)^{\epsilon/(1+\epsilon)}+\mathbb{P}(\|\bm{w}_{jt}\|_2\geq\tau_1)^{\epsilon/(1+\epsilon)}]\\
    \leq & \mathbb{E}[\|\bm{w}_{it}\|_2^{2+2\epsilon}]^{1/(2+2\epsilon)}\mathbb{E}[\|\bm{w}_{jt}\|_2^{2+2\epsilon}]^{1/(2+2\epsilon)}\left[\left(\frac{\mathbb{E}\|\bm{w}_{it}\|_2^{2+2\epsilon}}{\tau_1^{2+2\epsilon}}\right)^{\frac{\epsilon}{1+\epsilon}}+\left(\frac{\mathbb{E}\|\bm{w}_{jt}\|_2^{2+2\epsilon}}{\tau_1^{2+2\epsilon}}\right)^{\frac{\epsilon}{1+\epsilon}}\right]\\
    \leq & \frac{2M_{1,2+2\epsilon}}{\tau_1^{2+2\epsilon}} \lesssim \left[\frac{M_{1,2+2\epsilon}^{1/\epsilon}\log(r)}{n_\text{eff}}\right]^{\frac{\epsilon}{1+\epsilon}},
    \end{split}
  \end{equation}
  where $\tau_1\asymp[M_{1,2+2\epsilon}n_\text{eff}/\log(r)]^{1/(2+2\epsilon)}$.

  For the truncated $\bm{w}_{it}(\tau_1)$, based on the truncation level $\tau_1$ on $\|\bm{w}_{it}(\tau_1)\|_2$, it can be checked that, for any $k=3,4,\dots$
  \begin{equation}
    \begin{split}
      & \mathbb{E}[|\bm{w}_{it}(\tau_1)^\top\bm{w}_{jt}(\tau_1)|^k]\leq  \mathbb{E}[\|\bm{w}_{it}(\tau_1)\|_2^k\|\bm{w}_{jt}(\tau_1)\|_2^k]\\
      \leq & \tau_1^{2(k-1-\epsilon)}\mathbb{E}[\|\bm{w}_{it}(\tau_1)\|_2^{1+\epsilon}\|\bm{w}_{jt}(\tau_1)\|_2^{1+\epsilon}]\\
      \leq & \tau_1^{2(k-1-\epsilon)}\sqrt{\mathbb{E}[\|\bm{w}_{it}(\tau_1)\|_2^{2+2\epsilon}]\cdot\mathbb{E}[\|\bm{w}_{jt}(\tau_1)\|_2^{2+2\epsilon}]}\\
      \leq & \tau_1^{2(k-2)}(\tau_1^{2-2\epsilon}\cdot M_{1,2+2\epsilon}).
    \end{split}
  \end{equation}

  Similarly to the proof of Proposition \ref{prop:element}, by the property of mixing sequences in Lemma \ref{lemma:mixing} and the Bernstein-type inequality for the $\alpha$-mixing sequence in Lemma \ref{lemma:alpha-mixing},
  \begin{equation}
    \begin{split}
      & \mathbb{P}\left(\left|\frac{1}{T}\sum_{t=1}^T\bm{w}_{it}(\tau_1)^\top\bm{w}_{jt}(\tau_1)-\mathbb{E}[\bm{w}_{it}(\tau_1)^\top\bm{w}_{jt}(\tau_1)]\right| \geq \varepsilon \right)\\
      \leq & C\log(T)\exp\left[-C\frac{T}{\log(T)}\cdot\frac{\varepsilon^2}{\varepsilon\tau_1^2 + M_{1,2+2\epsilon}\cdot\tau_1^{2-2\epsilon}}\right]
    \end{split}
  \end{equation}
  Letting $\varepsilon=CM_{1,2+2\epsilon}^{1/(1+\epsilon)}\log(r)^{\epsilon/(1+\epsilon)}n_\text{eff}^{-\epsilon/(1+\epsilon)}$, we have
  \begin{equation}
    \begin{split}
      & \mathbb{P}\left(\left|\frac{1}{T}\sum_{t=1}^T\bm{w}_{it}(\tau_1)^\top\bm{w}_{jt}(\tau_1)-\mathbb{E}[\bm{w}_{it}(\tau_1)^\top\bm{w}_{jt}(\tau_1)]\right| \geq C\left[\frac{M_{1,2+2\epsilon}^{1/\epsilon}\log(r)}{n_\text{eff}}\right]^{\epsilon/(1+\epsilon)} \right)\\
      & ~~~~~\leq C\log(T)\exp\left[-C\frac{T}{\log(T)}\cdot\frac{\log(r)}{n_\text{eff}}\right]\\
      & ~~~~~= C\exp[\log\log(T)-C\log(r)\log(T)]\leq C\exp[-C\log(r)\log(T)].
    \end{split}
  \end{equation}
  For all $1\leq i,j\leq r$, we have
  \begin{equation}
    \begin{split}
      \mathbb{P}&\left(\max_{1\leq i,j\leq r} \left|\frac{1}{T}\sum_{t=1}^T\bm{w}_{it}(\tau_1)^\top\bm{w}_{jt}(\tau_1)-\mathbb{E}[\bm{w}_{it}(\tau_1)^\top\bm{w}_{jt}(\tau_1)]\right| \geq C\left[\frac{M_{1,2+2\epsilon}^{1/\epsilon}\log(r)}{n_\text{eff}}\right]^{\epsilon/(1+\epsilon)} \right)\\
      & \leq C\exp[\log(r)-C\log(r)\log(T)]\leq C\exp[-C\log(r)\log(T)].
    \end{split}
  \end{equation}
  Combining the above deviation bound and the bound of bias, we can obtain the convergence rate of $\|\widetilde{\bm{\Omega}}(\tau_1)-\bm{\Omega}\|_{1,\infty}$.

  Next, we prove the upper bound of $\|\widetilde{\bbm{\omega}}(\tau_1,\tau_2)-\bbm{\omega}\|_\infty$. Note that $\textup{vec}(\bm{y}_t\bm{x}_t^\top)^\top\bm{c}_i=\bm{y}_t^\top(\bm{I}_p\otimes\bm{x}_t^\top)\bm{c}_i=\bm{z}_{it}^\top\bm{w}_{it}$, for all $1\leq i\leq r$. Then, for $1\leq i\leq r$, the bias of data truncation can be bounded by
  \begin{equation}
    \begin{split}
      &|\mathbb{E}[\bm{z}_{it}^\top\bm{w}_{it}]-\mathbb{E}[\bm{z}_{it}(\tau_2)^\top\bm{w}_{it}(\tau_1)]|\\
      \leq & \mathbb{E}\left[|\bm{z}_{it}^\top\bm{w}_{it}|\cdot(1\{\|\bm{z}_{it}\|_2\geq \tau_2\}+1\{\|\bm{w}_{it}\|_2\geq \tau_1\})\right]\\
      \leq & \mathbb{E}[|\bm{z}_{it}^\top\bm{w}_{it}|^{1+\epsilon}]^{1/(1+\epsilon)}\cdot[\mathbb{P}(\|\bm{z}_{it}\|_2\geq\tau_2)^{\epsilon/(1+\epsilon)}+\mathbb{P}(\|\bm{w}_{it}\|_2\geq\tau_1)^{\epsilon/(1+\epsilon)}]\\
      \leq & \mathbb{E}[\|\bm{z}_{it}\|_2^{2+2\epsilon}]^{1/(2+2\epsilon)}\mathbb{E}[\|\bm{w}_{it}\|_2^{2+2\epsilon}]^{1/(2+2\epsilon)}\left[\left(\frac{\mathbb{E}\|\bm{y}_{t}\|_2^{2+2\epsilon}}{\tau_2^{2+2\epsilon}}\right)^{\frac{\epsilon}{1+\epsilon}}+\left(\frac{\mathbb{E}\|\bm{w}_{it}\|_2^{2+2\epsilon}}{\tau_1^{2+2\epsilon}}\right)^{\frac{\epsilon}{1+\epsilon}}\right]\\
      \lesssim & M_{1,2+2\epsilon}^{1/(2+2\epsilon)}M_{2,2+2\epsilon}^{1/(2+2\epsilon)}\left(M_{1,2+2\epsilon}^{\epsilon/(1+\epsilon)}\tau_1^{-2\epsilon}+M_{2,2+2\epsilon}^{\epsilon/(1+\epsilon)}\tau_2^{-2\epsilon}\right)\\
      \asymp & M_{1,2+2\epsilon}^{1/(2+2\epsilon)}M_{2,2+2\epsilon}^{1/(2+2\epsilon)}(\log(r)/n_\text{eff})^{\epsilon/(1+\epsilon)},
    \end{split}
  \end{equation}
  where
  \begin{equation}
    \tau_1\asymp\left[\frac{M_{1,2+2\epsilon}n_\text{eff}}{\log(r)}\right]^{1/(2+2\epsilon)}~~\text{and}~~\tau_2\asymp\left[\frac{M_{2,2+2\epsilon}n_\text{eff}}{\log(r)}\right]^{1/(2+2\epsilon)}.
  \end{equation}

  For the truncated $\bm{w}_{it}(\tau_1)$ and $\bm{z}_{it}(\tau_2)$, for any $k=3,4,\dots$
  \begin{equation}
    \begin{split}
      & \mathbb{E}[|\bm{w}_{it}(\tau_1)^\top\bm{z}_{it}(\tau_2)|^k] \leq \mathbb{E}[\|\bm{w}_{it}(\tau_1)\|_2^k\cdot\|\bm{z}_{it}(\tau_2)\|_2^k]\\
      \leq & (\tau_1\tau_2)^{k-1-\epsilon}\mathbb{E}[\|\bm{w}_{it}(\tau_1)\|_2^{1+\epsilon}\cdot\|\bm{y}_t(\tau_2)\|_2^{1+\epsilon}]\\
      \leq & (\tau_1\tau_2)^{k-1-\epsilon}\sqrt{\mathbb{E}[\|\bm{w}_{it}(\tau_1)\|_2^{2+2\epsilon}]\cdot\mathbb{E}[\|\bm{z}_{it}(\tau_2)\|_2^{2+2\epsilon}]}\\
      \leq & (\tau_1\tau_2)^{k-2}\left[(\tau_1\tau_2)^{1-\epsilon}\cdot M_{1,2+2\epsilon}^{1/2}\cdot M_{2,2+2\epsilon}^{1/2}\right].
    \end{split}
  \end{equation}

  By Lemma \ref{lemma:alpha-mixing},
  \begin{equation}
    \begin{split}
      &\mathbb{P}\left(\left|\frac{1}{T}\sum_{t=1}^T\bm{w}_{it}(\tau_1)^\top\bm{z}_{it}(\tau_2)-\mathbb{E}\left[\bm{w}_{it}(\tau_1)^\top\bm{z}_{it}(\tau_2)\right]\right|\geq \varepsilon\right)\\
      \leq & C\log(T)\exp\left[-C\frac{T}{\log(T)}\cdot\frac{\varepsilon^2}{\varepsilon\tau_1\tau_2+(\tau_1\tau_2)^{1-\epsilon}\cdot M_{1,2+2\epsilon}^{1/2}\cdot M_{2,2+2\epsilon}^{1/2}}\right].
    \end{split}
  \end{equation}
  Letting $\varepsilon=CM_{1,2+2\epsilon}^{1/(2+2\epsilon)}M_{2,2+2\epsilon}^{1/(2+2\epsilon)}[\log(r)/n_\text{eff}]^{\epsilon/(1+\epsilon)}$, we have
  \begin{equation}
    \begin{split}
      \mathbb{P}&\left(\left|\frac{1}{T}\sum_{t=1}^T\bm{w}_{it}(\tau_1)^\top\bm{z}_{it}(\tau_2)-\mathbb{E}[\bm{w}_{it}(\tau_1)^\top\bm{z}_{it}(\tau_2)]\right|\geq C\left[\frac{M_{1,2+2\epsilon}^{1/(2\epsilon)}M_{2,2+2\epsilon}^{1/(2\epsilon)}\log(r)}{n_\text{eff}}\right]^{\frac{\epsilon}{1+\epsilon}}\right)\\
      \leq &  C\log(T)\exp\left[-C\frac{T}{\log(T)}\cdot\frac{\log(r)\log(T)^2}{T}\right]\\
      \leq & C\exp[-C\log(r)\log(T)].
    \end{split}
  \end{equation}
  For all $1\leq i\leq r$,
  \begin{equation}
    \begin{split}
      \mathbb{P}&\left(\max_{1\leq i\leq r} \left|\frac{1}{T}\sum_{t=1}^T\bm{w}_{it}(\tau_1)^\top\bm{z}_{it}(\tau_2)-\mathbb{E}[\bm{w}_{it}(\tau_1)^\top\bm{z}_{it}(\tau_2)]\right| \geq C\left[\frac{M_{1,2+2\epsilon}^{1/(2\epsilon)}M_{2,2+2\epsilon}^{1/(2\epsilon)}\log(r)}{n_\text{eff}}\right]^{\frac{\epsilon}{1+\epsilon}} \right)\\
      & \leq Cr\exp[-C\log(r)\log(T)]= C\exp[\log(r)-C\log(r)\log(T)]\\
      & \leq C\exp[-C\log(r)\log(T)].
    \end{split}
  \end{equation}

  Finally, the convergence rate of $\|\widetilde{\bbm{\omega}}(\tau_1,\tau_2)-\bbm{\omega}\|_\infty$ is established by combining the deviation bound and the bound of the bias. 

\end{proof}

\subsection{Proofs of Theorems \ref{thm:sparseAR}--\ref{thm:network}}\label{sec:A.3}

\begin{proof}[\textbf{Proof of Theorem \ref{thm:sparseAR}}]
Denote the entries in $\bm{a}_i^*$ as $(a_{i,1}^*,\dots,a_{i,pd}^*)^\top$. First, the maximum norm bound directly follows from Propositions \ref{prop:2} and \ref{prop:element}. To derive the $\ell_1$ norm bound, for each row $\bm{a}_i^*$ in $\bm{A}^*$, define $S_{\kappa,i}=\{1\leq j\leq pd:|a_{ij}^*|\geq \kappa\}$, for $1\leq i\leq p$, with a threshold parameter $\kappa>0$. By the definition of $\mathbb{B}_q(s_q)$, we have $s_q\geq|S_{\kappa,i}|\cdot\kappa^q$, and thus $|S_{\kappa,i}|\leq s_q\cdot\kappa^{-q}$.

The approximation error on $S_{\kappa,i}^\complement=\{1,2,\dots,pd\}/S_{\kappa,i}$ can be bounded by
\begin{equation}
  \|(\bm{a}^*_i)_{S^\complement_{\kappa,i}}\|_1=\sum_{j\in S_{\kappa,i}^\complement}|a^*_{ij}|=\sum_{j\in S_{\kappa,i}^\complement}|a^*_{ij}|^q|a^*_{ij}|^{1-q}\leq s_q\cdot\kappa^{1-q}.
\end{equation}
Then, by Proposition \ref{prop:2}, the $\ell_1$ norm of the estimation error can be bounded by
\begin{equation}
  \begin{split}
    &\|\bm{\widehat{a}}_i(\lambda,\tau)-\bm{a}_i^*\|_{1}\\
    \lesssim & \|\bm{\Sigma}_0^{-1}\|_{1,\infty}\cdot|S_{\kappa,i}|\cdot\left[\frac{M_{2+2\epsilon}^{1/\epsilon}\log(p^2d)}{n_\textup{eff}}\right]^{\frac{\epsilon}{1+\epsilon}}+\|(\bm{a}_i^*)_{S_{\kappa,i}^\complement}\|_1\\
    \lesssim & \|\bm{\Sigma}_0^{-1}\|_{1,\infty}\cdot s_q\cdot\kappa^{-q}\cdot\left[\frac{M_{2+2\epsilon}^{1/\epsilon}\log(p^2d)}{n_\textup{eff}}\right]^{\frac{\epsilon}{1+\epsilon}} +  s_q\cdot\kappa^{1-q},
  \end{split}
\end{equation}
for all $1\leq i\leq p$. Setting 
\begin{equation}
  \kappa\asymp\|\bm{\Sigma}_0^{-1}\|_{1,\infty}\left[\frac{M_{2+2\epsilon}^{1/\epsilon}\log(p^2d)}{n_\textup{eff}}\right]^{\frac{\epsilon}{1+\epsilon}}
\end{equation}
we have that, with probability at least $1-C\exp[-C\log(T)\log(p^2d)]$,
\begin{equation}
  \|\bm{\widehat{a}}_i(\lambda,\tau)-\bm{a}_i^*\|_1\lesssim s_q\|\bm{\Sigma}_0^{-1}\|_{1,\infty}^{1-q}\left[\frac{M_{2+2\epsilon}^{1/\epsilon}\log(p^2d)}{n_\textup{eff}}\right]^{\frac{(1-q)\epsilon}{1+\epsilon}}.
\end{equation}

Finally, by the duality of the $\ell_1$ norm and $\ell_\infty$ norm, it can be shown that
\begin{equation}
  \begin{split}
    \|\bm{\widehat{a}}_i(\lambda,\tau)-\bm{a}_i^*\|_{2}^2 & \leq \|\bm{\widehat{a}}_i(\lambda,\tau)-\bm{a}_i^*\|_{1}
    \cdot\|\bm{\widehat{a}}_i(\lambda,\tau)-\bm{a}_i^*\|_{\infty}\\
    & \lesssim s_q\|\bm{\Sigma}_0^{-1}\|_{1,\infty}^{2-q}\left[\frac{M_{2+2\epsilon}^{1/\epsilon}\log(p^2d)}{n_\textup{eff}}\right]^{\frac{(2-q)\epsilon}{1+\epsilon}},
  \end{split}
\end{equation}
for all $1\leq i\leq p$.
\end{proof}

\begin{proof}[\textbf{Proof of Theorem \ref{thm:lowrankAR}}]

The operator norm directly follows from Propositions \ref{prop:1} and \ref{prop:vector}. Similarly to the proof of Theorem \ref{thm:sparseAR}, define the thresholded subspace $\mathcal{M}_\kappa$ corresponding to the column and row spaces spanned by the first $r_\kappa$ singular vectors where $\sigma_1(\bm{A}^*)\geq\cdots\geq\sigma_{r_\kappa}(\bm{A}^*)\geq\kappa>\sigma_{r_\kappa+1}(\bm{A}^*)$. By the definition of $\widetilde{\mathbb{B}}_q(r_q)$, we have $r_q\geq r_\kappa\cdot\kappa^q$ and thus $r_\kappa\leq r_q\cdot\kappa^{-q}$.

The approximation error in the nuclear norm can be bounded by
\begin{equation}
  \|\bm{A}^*_{\overline{\mathcal{M}}_\kappa^\perp}\|_\textup{nuc}=\sum_{r=r_\kappa+1}^{p}\sigma_r(\bm{A}^*)=\sum_{r=r_\kappa+1}^p\sigma_r^q(\bm{A}^*)\cdot\sigma_r^{1-q}(\bm{A}^*)\leq r_q\cdot\kappa^{1-q}.
\end{equation}
Then, by Proposition \ref{prop:2}, the nuclear norm estimation error can be bounded by
\begin{equation}
  \begin{split}
    &\|\bm{\widehat{A}}(\lambda,\tau_1,\tau_2)-\bm{A}^*\|_\textup{nuc}\\
    \lesssim & r_{\kappa}\|\bm{\Sigma}_0^{-1}\|_\textup{op}\left[\frac{pdM_{2+2\epsilon}^{1/\epsilon}}{n_\text{eff}}\right]^{\frac{\epsilon}{1+\epsilon}}+r_q\kappa^{1-q}\\
    \leq &  r_q\kappa^{-q}\|\bm{\Sigma}_0^{-1}\|_\textup{op}\left[\frac{pdM_{2+2\epsilon}^{1/\epsilon}}{n_\text{eff}}\right]^{\frac{\epsilon}{1+\epsilon}}+r_q\kappa^{1-q}.
  \end{split}
\end{equation}
Setting  
\begin{equation}
  \kappa\asymp\|\bm{\Sigma}_0^{-1}\|_\textup{op}\left[\frac{pdM_{2+2\epsilon}^{1/\epsilon}}{n_\text{eff}}\right]^{\frac{\epsilon}{1+\epsilon}},
\end{equation}
we have that, with probability at least $1-C\exp[-C\log(T)]$,
\begin{equation}
  \|\bm{\widehat{A}}(\lambda,\tau)-\bm{A}^*\|_\textup{nuc}\lesssim r_q\|\bm{\Sigma}_0^{-1}\|_\textup{op}^{1-q}\left[\frac{pdM_{2+2\epsilon}^{1/\epsilon}}{n_\text{eff}}\right]^{\frac{(1-q)\epsilon}{1+\epsilon}}.
\end{equation}

Finally, by the duality of the operator norm and nuclear norm, we have
\begin{equation}
  \|\bm{\widehat{A}}(\lambda,\tau)-\bm{A}^*\|_\textup{F}^2\leq \|\bm{\widehat{A}}(\lambda,\tau)-\bm{A}^*\|_\textup{nuc}\|\bm{\widehat{A}}(\lambda,\tau)-\bm{A}^*\|_\textup{op}\lesssim r_q\|\bm{\Sigma}_0^{-1}\|_\textup{op}^{2-q}\left[\frac{pdM_{2+2\epsilon}^{1/\epsilon}}{n_\text{eff}}\right]^{\frac{(2-q)\epsilon}{1+\epsilon}}.
\end{equation}
\end{proof}

\begin{proof}[\textbf{Proof of Theorem \ref{thm:banded}}]

Firstly, following Propositions \ref{prop:3}, \ref{prop:element} and \ref{prop:linear}, we can directly obtain the $\ell_\infty$ norm estimation bound of $\bbm{\widehat{\theta}}_\text{B}(\lambda,\tau)$,
\begin{equation}
  \|\bbm{\widehat{\theta}}_\text{B}(\lambda,\tau)-\bbm{\theta}^*\|_\infty \lesssim \|\bm{\Omega}^{-1}\|_{1,\infty} \left[\frac{M_{2+2\epsilon}^{1/\epsilon}\log(p)}{n_\text{eff}}\right]^{\frac{\epsilon}{1+\epsilon}}.
\end{equation}

Secondly, we derive the upper bound of $\bm{\widehat{A}}_\text{B}(\lambda,\tau)-\bm{A}^*$ in terms of operator norm. By the banded structure of $\bm{A}$ with bandwidth $k_0$ (see also the discussions in \citet{guo2016high}), it can be shown that
\begin{equation}
  \|\bm{\widehat{A}}_\text{B}(\lambda,\tau)-\bm{A}^*\|_\textup{op}\leq (2k_0+1)\|\bbm{\widehat{\theta}}_\text{B}(\lambda,\tau)-\bbm{\theta}^*\|_\infty \lesssim \|\bm{\Omega}^{-1}\|_{1,\infty} \left[\frac{M_{2+2\epsilon}^{1/\epsilon}\log(p)}{n_\text{eff}}\right]^{\frac{\epsilon}{1+\epsilon}}.
\end{equation}

Finally, as $\bbm{\theta}$ is a vector of dimension $2pk_0+p-k_0^2-k_0\asymp p$, we can directly obtain that
\begin{equation}
  \|\bbm{\widehat{\theta}}_\text{B}(\lambda,\tau)-\bbm{\theta}^*\|_2=\|\bm{\widehat{A}}_\text{B}(\lambda,\tau)-\bm{A}^*\|_\textup{F}\lesssim\sqrt{p}\|\bm{\Omega}^{-1}\|_{1,\infty}\left[\frac{M_{2+2\epsilon}^{1/\epsilon}\log(p)}{n_\text{eff}}\right]^{\frac{\epsilon}{1+\epsilon}}.
\end{equation}
\end{proof}

\begin{proof}[\textbf{Proof of Theorem \ref{thm:network}}]

Following Propositions \ref{prop:3} and \ref{prop:linear}, as $r$ is a fixed constant, we can obtain the $\ell_\infty$ norm estimation bound of $\bbm{\widehat{\theta}}_\text{N}(\lambda,\tau_1,\tau_2)$,
\begin{equation}
  \|\bbm{\widehat{\theta}}_\text{N}(\lambda,\tau_1,\tau_2)-\bbm{\theta}^*\|_\infty \lesssim \sqrt{p}\cdot\left[\frac{M_{2+2\epsilon}^{1/\epsilon}\log(p)}{n_\text{eff}}\right]^{\frac{\epsilon}{1+\epsilon}}.
\end{equation}
Then, for a fixed $r$, the $\ell_2$ norm bound of $\bbm{\widehat{\theta}}_\text{N}(\lambda,\tau_1,\tau_2)$ and the Frobenius norm bound of $\bm{\widehat{A}}_\text{N}(\lambda,\tau_1,\tau_2)$ directly follow.
\end{proof}

\subsection{Auxiliary Lemmas}\label{sec:A.4}

This subsection gives some auxiliary lemmas. We begin with a useful property of the transformation of $\alpha$- and $\beta$-mixing processes.

\begin{lemma}
  \label{lemma:mixing}
  For any $\alpha$- or $\beta$-mixing process $\{x_t\}_{t=1}^T$ and measurable function $f(\cdot)$, the sequence of the transformed observations $\{f(x_t)\}_{t=1}^T$ is also $\alpha$- or $\beta$-mixing in the same sense with its mixing coefficients bounded by those of the original sequence.
\end{lemma}

\begin{proof}[\textbf{Proof of Lemma \ref{lemma:mixing}}]
  For any measurable function $f(\cdot)$, it is clear  $\sigma(\{f(x_t)\}_{t=t_1}^{t_2})\subseteq\sigma(\{x_t\}_{t=t_1}^{t_2})$. Then, the statement can easily be verified by the definitions of the $\alpha$- and $\beta$-mixing conditions.
\end{proof}

Next, we show that if $\bm{y}_t$ satisfies some strong mixing conditions, the lagged observation $\bm{x}_t=(\bm{y}_{t-1},\dots,\bm{y}_{t-d})^\top$ is also strong mixing.

\begin{lemma}
  \label{lemma:lag_mixing}
  Suppose that $\bm{y}_t$ satisfies $\alpha$- or $\beta$-mixing condition with the mixing coefficients decayed geometrically. For any fixed positive integer $d$, the lagged values $\bm{x}_t=(\bm{y}_{t-1}^\top,\dots,\bm{y}_{t-d}^\top)^\top$ is also $\alpha$- or $\beta$-mixing with the mixing coefficients decayed geoemtrically.
\end{lemma}

\begin{proof}[\textbf{Proof of Lemma \ref{lemma:lag_mixing}}]
  In this proof, we focus on the $\alpha$-mixing condition, because the proof can also be applied to the $\beta$-mixing condition. For any $\ell\geq2d$, by the definition of  $\alpha$-mixing condition,
  \begin{equation}
    \alpha(\{\bm{x}_t\}_{t=-\infty}^s,\{\bm{x}_t\}_{t=s+2d}^\infty)=\alpha(\{\bm{y}_t\}_{t=-\infty}^{s-1},\{\bm{y}_t\}_{t=s+d-1}^\infty)\leq Cr^{-d}=(Cr^d)r^{-2d}:=C'r^{-2d},
  \end{equation}
  where $C'$ is another constant independent of the dimension $p$ and sample size $T$. For other lags smaller than $2d$, as there are only a fixed number of cases, the argument can easily be extended to all $\ell\geq1$.
\end{proof}

Next, we state a Bernstein-type concentration inequality for $\alpha$-mixing processes.

\begin{lemma}
  \label{lemma:alpha-mixing}
  Let $\{x_t\}_{t=1}^T$ be a strictly stationary $\alpha$-mixing process with mean zero and mixing coefficient $\alpha(l)\leq Cr^l$ for some $C>0$ and $r=r(p)<\bar{r}$, where $\bar{r}$ is a constant smaller than 1. Suppose that $\mathbb{E}|x_t|^k\leq Ck!A^{k-2}D^2$, where $D^2=\mathbb{E}(x_t^2)$, for $k=3,4,\dots$, then for any $\varepsilon>0$,
  \begin{equation}
    \mathbb{P}\left(\left|\sum_{t=1}^Tx_t\right|>n\varepsilon\right)\leq C[\log(T)+\mu(\varepsilon)]\exp\left[-C\frac{T\mu(\varepsilon)}{\log(T)}\right],
  \end{equation}
  where $\mu(\varepsilon)=\varepsilon^2/(5A\varepsilon+25D^2)$.
\end{lemma}

\begin{proof}[\textbf{Proof of Lemma \ref{lemma:alpha-mixing}}]

  Let $\mu(\varepsilon)=\varepsilon/(5A\varepsilon+25D^2)$. By Theorem 2.19 of \citet{fan2008nonlinear}, for any $T\geq2$, $k\geq3$, $q\in[1,T/2]$, and $\varepsilon>0$,
  \begin{equation}\label{eq:FY_upper}
    \begin{split}
      \mathbb{P}\left(\left|\sum_{t=1}^Tx_t\right|>T\varepsilon\right)& \leq C[1+T/q+\mu(\varepsilon)]\exp[-q\mu(\varepsilon)]\\
      &+CT\left[1+5\varepsilon^{-1}(\mathbb{E}x_t^k)^{1/(2k+1)}\right]\alpha\left(\left[\frac{T}{q+1}\right]\right)^{2k/(2k+1)}.
    \end{split}
  \end{equation}

  Since there exists a constant $\bar{r}<1$ such that $r(p)<\bar{r}$ for all $p$, letting  $q=T/(c_1\log(T))$ for sufficiently large $c_1>0$, we have $\alpha(T/(q+1))\asymp T^{-c_2}$, where $c_2$ can be arbitrarily large as long as $c_1$ is sufficiently large. Therefore, for any fixed $\varepsilon$ and when $T\to\infty$, the second term of the upper bound in \eqref{eq:FY_upper} can be arbitrarily small. In this case, we have
  \begin{equation}
    \mathbb{P}\left(\left|\sum_{t=1}^Tx_t\right|>T\varepsilon\right) \leq C[\log(T)+\mu(\varepsilon)]\exp\left[-\frac{T\mu(\varepsilon)}{\log(T)}\right].
  \end{equation}

\end{proof}

We state a covariance inequality for $\alpha$-mixing random variables from \citet{doukhan1994mixing}.
\begin{lemma}
  \label{lemma:covariance}
  If $\mathbb{E}[|X|^p+|Y|^q]<\infty$ for some $p,q\geq1$ and $1/p+1/q<1$, it holds that
  \begin{equation}
    |\textup{Cov}(X,Y)|\leq 8\alpha^{1/r}[\mathbb{E}|X|^p]^{1/p}[\mathbb{E}|Y|^q]^{1/q},
  \end{equation}
  where $r=(1-p^{-1}-q^{-1})^{-1}$ and $\alpha$ is the $\alpha$-mixing coefficient between $X$ and $Y$.
\end{lemma}

The definition of $\beta$-mixing condition for a time series of random vectors $\bm{y}_t$ can readily be extended to a time series of random matrices.
We then have a Bernstein-type concentration inequality for geometrically $\beta$-mixing dependent random matrices. The following lemma is Theorem 1 of \citet{banna2016bernstein}.

\begin{lemma}
  \label{lemma:beta-mixing}
  Let $\{\bm{X}_t\}_{t\geq1}$ be a family of self-adjoint stationary $\beta$-mixing random matrices of size $d$ with mean zero and mixing coefficient $\beta(\ell)\leq r^\ell$ for some $r\in(0,1)$. Suppose that $\lambda_{\max}(\bm{X}_t)\leq M$ almost surely, then for any $\varepsilon\geq0$,
  \begin{equation}
    \mathbb{P}\left(\lambda_{\max}\left(\sum_{t=1}^T\bm{X}_t\right)\geq \varepsilon\right)\leq d\exp\left(-\frac{C\varepsilon^2}{v^2T+c^{-1}M^2+\varepsilon M\gamma(c,T)}\right),
  \end{equation}
  where
  \begin{equation}
    v^2=\sup_{K\subset\{1,\dots,T\}}\frac{1}{\textup{Card}(K)}\lambda_{\max}\left(\mathbb{E}\left(\sum_{t\in K}\bm{X}_t\right)^2\right)
  \end{equation}
  and
  \begin{equation}
    \gamma(c,T)=\frac{\log(T)}{\log(2)}\max\left(2,\frac{32\log(T)}{c\log(2)}\right).
  \end{equation}
\end{lemma}

\section{Proof of Lower Bound Results}\label{append:B}

In this section, we provide proofs of the minimax lower bound results for the estimation problems of autocovariance matrices and VAR models. We prove the lower bounds for  autocovariance estimation in the $\ell_\infty$ norm and sparse VAR estimation in Section \ref{sec:B.1} and those of  the operator norm lower bound for autocovariance estimation and the Frobenius norm lower bound for reduced-rank VAR estimation in Section \ref{sec:B.2}. Some auxiliary lemmas are presented in Section \ref{sec:B.3}.

We start with some notations. For any real number $x$, the ceiling function $\lceil x\rceil$ is defined as the smallest integer not smaller than $x$ and the floor function $\lfloor x\rfloor$ is the largest integer not larger than $x$. For any binary vectors $\bm{v}_1,\bm{v}_2\in\{0,1\}^p$, the Hamming distance of $\bm{v}_1$ and $\bm{v}_2$, denoted by $d_H(\bm{v}_1,\bm{v}_2)$, is defined as the number of different entries in the two vectors.

\subsection{Maximum Norm Lower Bounds and Sparse VAR}\label{sec:B.1}

We start with constructing minimax lower bounds of autocovariance matrix estimation in the $\ell_\infty$ norm.

\begin{proof}[\textbf{Proof of Proposition \ref{prop:LB_maximum}}]
  
  The proof consists of four steps. The first three steps present the minimax lower bound for $\bm{\Sigma}_0$ with $d=1$. Specifically, in the first step, a pair of discrete distributions are constructed. In the second step, we consider a pair of $\alpha$-mixing time series $\{\bm{y}_1,\dots,\bm{y}_T\}$ and $\{\bm{y}_1',\dots,\bm{y}_T'\}$ with the mixing coefficients decayed geometrically. In the third step, we show that with probability at a constant level, it is impossible to obtain an estimation error smaller than the minimax lower bound for $\{\bm{y}_1,\dots,\bm{y}_T\}$ and $\{\bm{y}_1',\dots,\bm{y}_T'\}$ simultaneously.
  In the last step, the minimax lower bound is extended to the case of $d>1$ and $\bm{\Sigma}_1$.
  Note that we establish the lower bounds in a non-asymptotic sense, where all parameters including $p$, $T$ and $M$ are assumed to be finite.\\

  \noindent\textit{Step 1. Specify a pair of bivariate distributions}

  First, motivated by \citet{devroye2016sub}, we construct a pair of discrete bivariate distributions for $(y_1,y_2)^\top$. Specifically, for the given moment bound $M$, we consider the bivariate distribution $\mathcal{P}_{c,\gamma}=\{P_+,P_-\}$ with some positive constants $c>0$ and $0<\gamma<1$ such that $c^{2+2\epsilon}\gamma=M$, in which
  \begin{equation}
    \begin{split}
      &P_+(y_1=c,y_2=c)=P_+(y_1=-c,y_2=-c)=\frac{\gamma}{2},\quad P_+(y_1=y_2=0)=1-\gamma,\\
      &P_-(y_1=c,y_2=-c)=P_-(y_1=-c,y_2=c)=\frac{\gamma}{2},~~\text{and}~~P_-(y_1=y_2=0)=1-\gamma.
    \end{split}
  \end{equation}
  It is simple to check that under both distributions $P_+$ and $P_-$, $\mathbb{E}(y_1)=\mathbb{E}(y_2)=0$ and $\mathbb{E}(|y_1|^{2+2\epsilon})=\mathbb{E}(|y_2|^{2+2\epsilon})=c^{2+2\epsilon}\gamma=M$. In addition, $P_+(y_1y_2=c^2)=\gamma$, $P_+(y_1y_2=0)=1-\gamma$, $P_-(y_1y_2=-c^2)=\gamma$, and $P_-(y_1y_2=0)=1-\gamma$. Therefore, $\mathbb{E}_+(y_1y_2)=c^2\gamma$ and $\mathbb{E}_-(y_1y_2)=-c^2\gamma$.\\

  \noindent\textit{Step 2. Construct an $\alpha$-mixing sequence}

  Let $v=r^{-2}>1$. For $t=1,2,\dots,T$, consider $\widetilde{T}:=\lceil (T/\lfloor\log_vT\rfloor)\rceil$ time index groups $I_1=\{1,\dots,\lfloor\log_{v}T\rfloor\}$, $I_2=\{\lfloor\log_{v}T\rfloor+1,\dots,2\lfloor\log_{v}T\rfloor\}$, $\cdots$, $I_{\widetilde{T}}=\{(\widetilde{T}-1)\lfloor\log_{v}T\rfloor+1,\dots,\widetilde{T}\lfloor\log_{v}T\rfloor\}$. For each $1\leq k\leq\widetilde{T}$, let $\{y_{1t},y_{2t},y_{1t}',y_{2t}'\}_{t\in I_k}$ be random variables with the joint distribution
  \begin{equation}
    \begin{split}
      & \mathbb{P}\left(\cap_{t\in I_k}\{y_{1t}=y_{2t}=y_{1t}'=y_{2t}'=0\}\right)=1-\gamma,\\
      & \mathbb{P}\left(\cap_{t\in I_k}\{y_{1t}=c,y_{2t}=c,y_{1t}'=c,y_{2t}'=-c\}\right)=\frac{\gamma}{2},\\
      \text{and }& \mathbb{P}\left(\cap_{t\in I_k}\{y_{1t}=-c,y_{2t}=-c,y_{1t}'=-c,y_{2t}'=c\}\right)=\frac{\gamma}{2}.
    \end{split}
  \end{equation}
  Marginally, it is easy to check that $(y_{1t},y_{2t})\sim P_+$ and $(y_{1t}',y_{2t}')\sim P_-$. For the variables in different time index groups, $\{y_{1t},y_{2t},y_{1t}',y_{2t}'\}$ are independent. 

  For the rest of variables other than the first and second ones, we assume that $y_{it}=y_{it}'$ almost surely for $i\in\{3,4,\dots,p\}$ and they are independent with $y_{1t}$ and $y_{2t}$. In addition, they are independent at different time points. Though we only consider the data from $t=1$ to $t=T$, we assume that the above distribution applies to all time indexes $t$ so that we can study the mixing property of the stochastic process.

  Let $\gamma=(2\widetilde{T})^{-1}\log(1/(2\delta))\leq 1/2$ with some constant $\delta\in[\exp(-\widetilde{T})/2,1/2)$. Now we verify that both $\bm{y}_t=(y_{1t},y_{2t},y_{3t},\dots,y_{pt})^\top$ and $\bm{y}_t'=(y_{1t}',y_{2t}',y_{3t}',\dots,y_{pt}')^\top$ are $\alpha$-mixing with mixing coefficients geometrically decayed. As $\{(y_{3t},\dots,y_{pt})^\top\}_{t=1}^T$ and $\{(y_{3t}',\dots,y_{pt}')^\top\}_{t=1}^T$ are independent series and are independent with $\{y_{1t},y_{2t},y_{1t}',y_{2t}'\}$, it suffices to show that $\{(y_{1t},y_{2t})\}_{t=1}^T$ and $\{(y_{1t}',y_{2t}')\}_{t=1}^T$ satisfy the $\alpha$-mixing condition.

  Without loss of generality, we first consider $t_1,t_2\in I_1$ with $t_1<t_2$. By definition, it can easily be checked that,
  \begin{equation}
    \begin{split}
      &\alpha(\{\bm{y}_t\}_{t=-\infty}^{t_1},\{\bm{y}_t\}_{t=t_2}^{\infty})=\alpha\left(\{(y_{1t},y_{2t})^\top\}_{t=1}^{t_1},\{(y_{1t},y_{2t})^\top\}_{t=t_2}^{\lfloor\log_vT\rfloor}\right)\\
      =&\sup|\mathbb{P}(A\cap B)-\mathbb{P}(A)\mathbb{P}(B)|,~~A\in\sigma(\{(y_{1t},y_{2t})^\top\}_{t=1}^{t_1})~\text{and}~B\in\sigma(\{(y_{1t},y_{2t})^\top\}_{t=t_2}^{\lfloor\log_vT\rfloor})\\
      =&\max\{(1-\gamma)-(1-\gamma)^2,\gamma/2-\gamma^2/4,\gamma(1-\gamma)\}\\
      \leq&\gamma=\frac{\log(1/(2\delta))}{2\widetilde{T}}=Cr^{2\log_v\widetilde{T}}\leq Cr^{\log_vT}\leq Cr^{\lfloor\log_vT\rfloor}\leq Cr^{t_2-t_1}.
    \end{split}
  \end{equation}
  In addition, for $t_1<t_2$ in different time index groups,  it is obvious that
  \begin{equation}
    \alpha(\{\bm{y}_t\}_{t=-\infty}^{t_1},\{\bm{y}_t\}_{t=t_2}^{\infty})=0,
  \end{equation}
  because the variables in different groups are independent. 
  Therefore, the constructed $p$-dimensional series $\bm{y}_t$ and $\bm{y}'_t$ are both $\alpha$-mixing with the mixing coefficients decayed geometrically.\\

  \noindent\textit{Step 3. Establish the lower bound for $\bm{\Sigma}_0$ with $d=1$}

  As $1-\gamma\geq\exp[-\gamma/(1-\gamma)]$, we have
  \begin{equation}
    \mathbb{P}(\{\bm{y}_t\}_{t=1}^T=\{\bm{y}_t'\}_{t=1}^T)=(1-\gamma)^{\widetilde{T}}\geq\exp\left(\frac{-\gamma\widetilde{T}}{1-\gamma}\right)\geq\exp(-2\gamma\widetilde{T})=2\delta.
  \end{equation}

  Let $\widehat{\mathbb{E}}_{1,2}(\cdot)$ be any mean estimator of $(\bm{\Sigma}_0)_{1,2}$. Then,
  \begin{equation}
    \begin{split}
      & \max\left[\mathbb{P}\left(\left|\widehat{\mathbb{E}}_{1,2}(\{\bm{y}_t\}_{t=1}^T)-c^2\gamma\right|>c^2\gamma\right),\mathbb{P}\left(\left|\widehat{\mathbb{E}}_{1,2}(\{\bm{y}_{t}'\}_{t=1}^T)+c^2\gamma\right|>c^2\gamma\right)\right]\\
      \geq & \frac{1}{2}\mathbb{P}\left[\left|\widehat{\mathbb{E}}_{1,2}(\{\bm{y}_{t}\}_{t=1}^T)-c^2\gamma\right|>c^2\gamma,\text{ or }\left|\widehat{\mathbb{E}}_{1,2}(\{\bm{y}_{t}'\}_{t=1}^T)+c^2\gamma\right|>c^2\gamma\right]\\
      \geq & \frac{1}{2}\mathbb{P}\left[\widehat{\mathbb{E}}_{1,2}(\{\bm{y}_{t}\}_{t=1}^T)=\widehat{\mathbb{E}}_{1,2}(\{\bm{y}_{t}'\}_{t=1}^T)\right]\\
      \geq & \frac{1}{2}\mathbb{P}\left[\{\bm{y}_t\}_{t=1}^T=\{\bm{y}_{t}'\}_{t=1}^T\right]\geq \delta,
    \end{split}
  \end{equation}
  where $\widehat{\mathbb{E}}_{1,2}(\{\bm{y}_t\}_{t=1}^T)=\widehat{\mathbb{E}}_{1,2}(\{\bm{y}_t'\}_{t=1}^T)$, following  directly from $\{\bm{y}_t\}_{t=1}^T=\{\bm{y}_t'\}_{t=1}^T$.

  Taking $\delta=1/3$, we have that
  \begin{equation}
    c^2\gamma = M^{1/(1+\epsilon)}\gamma^{\epsilon/(1+\epsilon)}=\left[\frac{M^{1/\epsilon}\log(1/(2\delta))}{2\widetilde{T}}\right]^{\frac{\epsilon}{1+\epsilon}}\asymp\left[\frac{M^{1/\epsilon}\log(T)}{T}\right]^{\frac{\epsilon}{1+\epsilon}}.
  \end{equation}

  Hence, we have
  \begin{equation}
    \begin{split}
      \max&\left\{\mathbb{P}\left[\left|\widehat{\mathbb{E}}_{1,2}(\{\bm{y}_{t}\}_{t=1}^T)-\mathbb{E}[y_{1t}y_{2t}]\right|\gtrsim\left[\frac{M^{1/\epsilon}\log(T)}{T}\right]^{\frac{\epsilon}{1+\epsilon}}\right],\right.\\
      &\quad\left.\mathbb{P}\left[\left|\widehat{\mathbb{E}}_{1,2}(\{\bm{y}_{t}'\}_{t=1}^T)-\mathbb{E}[y_{1t}'y_{2t}']\right|\gtrsim\left[\frac{M^{1/\epsilon}\log(T)}{T}\right]^{\frac{\epsilon}{1+\epsilon}}\right]\right\}\geq 1/3.
    \end{split}
  \end{equation}
  Therefore, as $\mathcal{P}_{c,\gamma}\in\mathcal{P}_\textup{E}(M,\epsilon,r)$, the minimax lower bound for $\bm{\Sigma}_0$ in the $\ell_\infty$ norm can be established.\\

  \textit{Step 4. Extension to the case of $d>1$ and $\bm{\Sigma}_1$}

  Since the analysis in the first three steps is independent of $p$, the above minimax lower bound result can be extended directly to the $pd$-dimensional vector $\bm{x}_t=(\bm{y}_{t-1}^\top,\bm{y}_{t-2}^\top,\dots,\bm{y}_{t-d}^\top)^\top$ and the $p(d+1)$-dimensional vector $\bm{z}_t=(\bm{y}_t^\top,\bm{x}_t^\top)^\top$. Hence, the minimax lower bound results for the case of $d>1$ and $\bm{\Sigma}_1$ can be obtained in a similar fashion, so the details are omitted for brevity.

\end{proof}

Next, we apply the techniques developed above to the estimation of sparse VAR models.

\begin{proof}[\textbf{Proof of Theorem \ref{thm:LB_sparse_VAR}}]
  
  Similarly to the proof of Proposition \ref{prop:LB_maximum}, we focus on the case of $d=1$, as the proof can readily be extended to the general case of $d>1$. The proof consists of three steps.
  In the first step, a pair of joint distributions of $\{\bm{y}_1,\dots,\bm{y}_T\}$ and $\{\bm{y}_1',\dots,\bm{y}_T'\}$ are constructed. In the second step, we verify that the constructed time series are $\alpha$-mixing with mixing coefficients decayed geometrically. In the final step, we show that with probability at a constant level, it is impossible to obtain an estimation error smaller than the minimax lower bound for $\{\bm{y}_1,\dots,\bm{y}_T\}$ and $\{\bm{y}_1',\dots,\bm{y}_T'\}$ simultaneously.\\

  \noindent\textit{Step 1. Specify a class of distribution}

  For simplicity, for given $0<\epsilon\leq 1$,  sparsity level $s_0=s$, $c>0$, and $0<\gamma<1/s$, we consider the bivariate distribution $\mathcal{P}_{c,\gamma,s,\epsilon}=\{P_+,P_-\}$ for random variables $\{x,y\}$, where
  \begin{equation}
    \begin{split}
      & P_+(x=cs^{\epsilon/(1+\epsilon)},y=c)=\frac{\gamma}{2},\quad\quad P_+(x=-cs^{\epsilon/(1+\epsilon)},y=-c)=\frac{\gamma}{2},\\
      & P_+(x=0,y=c)=\frac{(s-1)\gamma}{2},\quad\quad P_+(x=0,y=-c)=\frac{(s-1)\gamma}{2},\\
      & P_+(x=0,y=0)=1-s\gamma,
    \end{split}
  \end{equation}
  and
  \begin{equation}
    \begin{split}
      & P_-(x=cs^{\epsilon/(1+\epsilon)},y=-c)=\frac{\gamma}{2},\quad\quad P_-(x=-cs^{\epsilon/(1+\epsilon)},y=c)=\frac{\gamma}{2},\\
      & P_-(x=0,y=c)=\frac{(s-1)\gamma}{2},\quad\quad P_-(x=0,y=-c)=\frac{(s-1)\gamma}{2},\\
      & P_-(x=0,y=0)=1-s\gamma.
    \end{split}
  \end{equation}
  We have that under both distributions, $\mathbb{E}(x)=\mathbb{E}(y)=0$, $\mathbb{E}(x^2)=c^2s^{2\epsilon/(1+\epsilon)}\gamma$, $\mathbb{E}(y^2)=c^2s\gamma$, $\mathbb{E}|x|^{2+2\epsilon}=c^{2+2\epsilon}s^{2\epsilon}\gamma:=M$, $\mathbb{E}|y|^{2+2\epsilon}=c^{2+2\epsilon}s\gamma$. Also, $\mathbb{E}_+(xy)=c^2s^{\epsilon/(1+\epsilon)}\gamma$ and $\mathbb{E}_-(xy)=-c^2s^{\epsilon/(1+\epsilon)}\gamma$.\\

  \noindent\textit{Step 2. Construct an $\alpha$-mixing sequence}

  Note that $y_{s+1,t+1}$ is only correlated with $\bm{x}_t=(y_{1t},\dots,y_{st})^\top$, and $y_{s+1,t+1}'$ is only correlated with $\bm{x}_t'=(y_{1t}',\dots,y_{st}')^\top$. Let $v=r^{-2}>1$. For $t=1,2,\dots,T$, consider $\widetilde{T}:=\lceil(T/\log_v T)\rceil$ index groups $I_1=\{1,\dots,\lceil\log_vT\rceil\}$, $I_2=\{\lceil\log_vT\rceil+1,\dots,2\lceil\log_vT\rceil\}$, $\dots$, $I_{\widetilde{T}}=\{(\widetilde{T}-1)\lceil\log_vT\rceil+1,\dots,\widetilde{T}\lceil\log_vT\rceil\}$. For each $1\leq k\leq \widetilde{T}$, let $\{\bm{x}_t,y_{s+1,t+1},\bm{x}_t',y_{s+1,t+1}'\}_{t\in I_k}$ be random variables with the joint distribution
  \begin{equation}
    \begin{split}
      & \mathbb{P}\left(\cap_{t\in I_k}\left\{\bm{x}_t=cs^{\frac{\epsilon}{1+\epsilon}}\bm{e}_i,y_{s+1,t+1}=c,\bm{x}_t'=cs^{\frac{\epsilon}{1+\epsilon}}\bm{e}_i,y_{s+1,t+1}'=-c\right\}\right)=\frac{\gamma}{2},~~1\leq i\leq s,\\
      & \mathbb{P}\left(\cap_{t\in I_k}\left\{\bm{x}_t=-cs^{\frac{\epsilon}{1+\epsilon}}\bm{e}_i,y_{s+1,t+1}=-c,\bm{x}_t'=-cs^{\frac{\epsilon}{1+\epsilon}}\bm{e}_i,y_{s+1,t+1}'=c\right\}\right)=\frac{\gamma}{2},~~1\leq i\leq s,\\
      & \text{and}~~\mathbb{P}\left(\cap_{t\in I_k}\left\{\bm{x}_t=\bm{x}_t'=\bm{0},y_{s+1,t+1}=y_{s+1,t+1}'=0\right\}\right)=1-s\gamma.
    \end{split}
  \end{equation}
  Marginally, it is easy to check that each pair of $(y_{it},y_{s+1,t+1})\sim P_+$ and $(y_{it}',y_{s+1,t+1}')\sim P_-$, for $i=1,\dots,s$, respectively. In addition, for any $1\leq i<j\leq s$, $\text{Cov}(y_{it},y_{jt})=0$.

  For the rest of variables other than the first $s+1$ variables, we assume that $y_{it}=y_{it}'$ almost surely for $i\in\{s+2,\dots,p\}$ and they are independent with $\bm{x}_t$ and $y_{s+1,t}$. Let $\gamma=(2s\widetilde{T})^{-1}\log(1/(2\delta))\leq 1/2$ with some constant $\delta\in[\exp(-\widetilde{T}/2,1/2))$. Since both $\{(y_{s+2,t},\dots,y_{pt})^\top\}_{t=1}^T$ and $\{(y_{s+2,t}',\dots,y_{pt}')^\top\}_{t=1}^T$ are independent series and they are independent with the first $s+1$ variables, it suffices to show that $\widetilde{\bm{y}}_t=\{(\bm{x}_t^\top,y_{s,t})^\top\}_{t=1}^T$ and $\widetilde{\bm{y}}_t'=\{(\bm{x}_t^{'\top},y_{s,t}')^\top\}_{t=1}^T$ satisfy the $\alpha$-mixing condition.

  Without loss of generality, we first consider $t_1,t_2\in I_1$ with $t_1<t_2$. By definition, it can be checked that
  \begin{equation}
    \begin{split}
      & \alpha(\{\widetilde{\bm{y}}_t\}_{t=-\infty}^{t_1},\{\widetilde{\bm{y}}_t\}_{t=t_2}^{\infty})=\alpha(\{\widetilde{\bm{y}}_t\}_{t=1}^{t_1},\{\widetilde{\bm{y}}_t\}_{t=t_2}^{\lfloor\log_vT\rfloor})\\
      = & \sup|\mathbb{P}(A\cap B)-\mathbb{P}(A)\mathbb{P}(B)|,~~A\in\sigma(\{\widetilde{\bm{y}}_t\}_{t=1}^{t_1}),~~B\in\sigma(\{\widetilde{\bm{y}}_t\}_{t=t_2}^{\lfloor\log_vT\rfloor})\\
      = & \max\{(1-s\gamma)-(1-s\gamma)^2,s\gamma/2-s^2\gamma^2/4,s\gamma(1-s\gamma)\}\\
      \leq & s\gamma = \frac{\log(1/2\delta)}{2\widetilde{T}}=Cr^{2\log_v\widetilde{T}}\leq Cr^{\log_vT}\leq Cr^{\lfloor\log_v T\rfloor}\leq Cr^{t_2-t_1}.
    \end{split}
  \end{equation}
  In addition, for $t_1<t_2$ in different index groups, as the variables in different index groups are independent, we have
  \begin{equation}
    \alpha(\{\bm{y}_t\}_{t=-\infty}^{t_1},\{\bm{y}_t\}_{t=t_2}^{\infty})=0.
  \end{equation}
  Therefore, the constructed $p$-dimensional series $\bm{y}_t$ and $\bm{y}_t'$ are both $\alpha$-mixing with the mixing coefficients decayed geometrically.\\

  \noindent\textit{Step 3. Establish the lower bound}

  For the constructed sequence $\bm{y}_t$ and $\bm{y}_t'$, we have
  \begin{equation}
    \begin{split}
      & \mathbb{P}(\{y_{1t},\dots,y_{st},y_{s+1,t+1}\}_{t=1}^T=\{y_{1t}',\dots,y_{st}',y_{s+1,t+1}'\}_{t=1}^T) = (1-s\gamma)^{\widetilde{T}}\\
      \geq & \exp\left(\frac{s\gamma\widetilde{T}}{1-s\gamma}\right)\geq\exp(-2s\gamma\widetilde{T})=2\delta.
    \end{split}
  \end{equation}

  We consider the estimation task of coefficients of the predictors $\bm{x}_t=(y_{1t},\dots,y_{st})^\top$ and the response $y_{s+1,t}$. By the Yule--Walker equation, the true values of the coefficient vector are defined as $\bm{a}^*(\{\bm{y}_t\})=\mathbb{E}[\bm{x}_t\bm{x}_{t}^\top]^{-1}\mathbb{E}[\bm{x}_{t}y_{s+1,t}]$ and $\bm{a}^*(\{\bm{y}'_t\})=\mathbb{E}[\bm{x}_{t}'\bm{x}_{t}^{'\top}]^{-1}\mathbb{E}[\bm{x}'_{t}y'_{s+1,t}]$, respectively. As $\bm{\Sigma}_{\bm{x}}:=\mathbb{E}[\bm{x}_{t}\bm{x}_{t}^\top]=\mathbb{E}[\bm{x}_{t}'\bm{x}_{t}'^{\top}]$ is a diagonal matrix, it can be obtained that $\|\bm{a}^*(\{\bm{y}_t\})\|_2 = \|\bm{a}^*(\{\bm{y}_t'\})\|_2 =\|\bm{\Sigma}_{\bm{x}}^{-1}\|_{1,\infty}c^2s^{\frac{\epsilon}{1+\epsilon}+\frac{1}{2}}\gamma$.

  Let $\widehat{\mathbb{E}}_{s+1}(\cdot)$ be any mean estimator of $\bm{a}_{s+1}$. Then,
  \begin{equation}
    \begin{split}
      & \max\left\{\mathbb{P}\left[\left\|\widehat{\mathbb{E}}_{s+1}(\{\bm{y}_t\}_{t=1}^T)-\bm{a}^*(\{y_t\})\right\|_2>\|\bm{a}^*(\{\bm{y}_t\})\|_2\right],\right.\\
      &\quad\quad\quad\left.\mathbb{P}\left[\left\|\widehat{\mathbb{E}}_{s+1}(\{\bm{y}_{t}'\}_{t=1}^T)+\bm{a}^*(\{\bm{y}_t'\})\right\|_2>\|\bm{a}^*(\{\bm{y}_t'\})\|_2\right]\right\}\\
      \geq & \frac{1}{2}\mathbb{P}\left[\left\|\widehat{\mathbb{E}}_{s+1}(\{\bm{y}_{t}\}_{t=1}^T)-\bm{a}^*(\{\bm{y}_t\})\right\|_2>\|\bm{a}^*(\{\bm{y}_t\})\|_2,\right.\\
      &\left.\quad\quad\quad\text{ or }\left\|\widehat{\mathbb{E}}_{s+1}(\{\bm{y}_{t}'\}_{t=1}^T)+\bm{a}^*(\{\bm{y}_t'\})\right\|_2>\|\bm{a}^*(\{\bm{y}_t'\})\|_2\right]\\
      \geq & \frac{1}{2}\mathbb{P}\left[\widehat{\mathbb{E}}_{s+1}(\{\bm{y}_{t}\}_{t=1}^T)=\widehat{\mathbb{E}}_{s+1}(\{\bm{y}_{t}'\}_{t=1}^T)\right]\\
      \geq & \frac{1}{2}\mathbb{P}\left[\{\bm{y}_t\}_{t=1}^T=\{\bm{y}_{t}'\}_{t=1}^T\right]\\
      =& \frac{1}{2}\mathbb{P}\left[\{y_{1t},\dots,y_{st},y_{s+1,t+1}\}_{t=1}^T=\{y_{1t}',\dots,y_{st}',y_{s+1,t+1}'\}_{t=1}^{T}\right]\geq \delta,
    \end{split}
  \end{equation}

  Taking $\delta=1/3$, we have that
  \begin{equation}
    \begin{split}
      \|\bm{a}^*(\{\bm{y}_t\})\|_2=&\|\bm{a}^*(\{\bm{y}_t'\})\|_2 =  \|\bm{\Sigma}_{\bm{x}}^{-1}\|_{1,\infty}(c^2s^{\frac{\epsilon}{1+\epsilon}}\gamma^{\frac{1}{1+\epsilon}})\sqrt{s}\gamma^{\frac{\epsilon}{1+\epsilon}} =  \|\bm{\Sigma}_{\bm{x}}^{-1}\|_{1,\infty}M^{\frac{1}{1+\epsilon}}\sqrt{s}\gamma^{\frac{\epsilon}{1+\epsilon}}\\
      = & \|\bm{\Sigma}_{\bm{x}}^{-1}\|_{1,\infty}M^{\frac{1}{1+\epsilon}}\sqrt{s}\left[\frac{\log(1/(2\delta))}{2\widetilde{T}}\right]^{\frac{\epsilon}{1+\epsilon}} \asymp \|\bm{\Sigma}_{\bm{x}}^{-1}\|_{1,\infty}\sqrt{s}\left[\frac{M^{1/\epsilon}\log(T)}{T}\right]^{\frac{\epsilon}{1+\epsilon}}.
    \end{split}
  \end{equation}
  Therefore, the minimax lower bound for $\bm{a}_{s+1}$ in terms of $\ell_2$ norm can be derived, and the $\|\cdot\|_{2,\infty}$ lower bound can be developed accordingly.

  Finally, for $d>1$, we consider the same data generating process and use the VAR($d$) model as the running model. As the coefficient vector in the VAR(1) model is a sub-vector of that in the VAR($d$) model, the lower bound result for the VAR(1) model can be extended directly to the VAR($d$) model.

\end{proof}

\subsection{Operator Norm Lower Bounds and Reduced-Rank VAR}\label{sec:B.2}

In this subsection, we establish the minimax lower bounds of the autocovariance estimation in the operator norm and estimation of the reduced-rank VAR.

\begin{proof}[\textbf{Proof of Proposition \ref{prop:LB_operator}}]
  
  The proof consists of three steps. The operator norm minimax lower bound for $\bm{\Sigma}_0$ with $d=1$ is developed in the first three steps. In the first step, we construct a class of multivariate distributions with some moment conditions. In the second step, we apply Fano's inequality (Lemma \ref{lemma:Fano}) to show that with probability at a constant level, it is impossible to obtain an estimation error in the operator norm smaller than the minimax lower bound. Finally, the lower bound result is extended to the case of $d>1$ and $\bm{\Sigma}_1$ in the third step.\\

  \noindent\textit{Step 1. Construct a class of discrete multivariate distributions}

  Without loss of generality, we assume that the dimension $p$ is a multiple of 2, i.e., $p=2h$, where $h$ is a positive integer. Consider $2^p$ vectors $\bm{v}_1,\bm{v}_2,\dots,\bm{v}_{2^p}$ where each $\bm{v}_i\in\{p^{-1/2},-p^{-1/2}\}^p$, for $i=1,\dots,2^p$. 

  For any $\bm{v}_i$, $1\leq i\leq 2^p$, it can be obtained that
  $\sum_{j=1}^{2^p}\langle\bm{v}_i,\bm{v}_j\rangle^2=2^p/p$. Based on $\bm{v}_i$, define the corresponding multivariate distribution $P_{\bm{v}_i}$, where
  \begin{equation}
    \begin{split}
      &P_{\bm{v}_i}(\bm{y}=c\bm{v}_j)=\frac{\gamma}{p}+\langle\bm{v}_i,\bm{v}_j\rangle^2\gamma,\quad j=1,\dots,2^p\\
      \text{and}~&P_{\bm{v}_i}(\bm{y}=\bm{0}_p)=1-\frac{2^p}{p}\gamma-\sum_{j=1}^{2^p}\langle\bm{v}_i,\bm{v}_j\rangle^2\gamma=1-\frac{2^{p+1}}{p}\gamma,
    \end{split}
  \end{equation}
  with the parameters $c$ and $\gamma$ satisfying that $c^{2+2\epsilon}\gamma p^{-(2+\epsilon)}2^{p+3}=M$ and $\gamma<2^{-(p+1)}p$.

  Under each distribution $P_{\bm{v}_i}$, it can be checked that $\mathbb{E}_{\bm{v}_i}[\bm{y}]=\bm{0}_p$,
  \begin{equation}
    \bm{\Sigma}_{\bm{v}_i}:=\mathbb{E}_{\bm{v}_i}[\bm{y}\bm{y}^\top]= \left[\frac{2^p}{p^2}+\frac{2^p(p-2)}{p^3}\right]c^2\gamma\bm{I}_p +  \frac{2^{p+1}}{p^2}c^2\gamma\bm{v}_i\bm{v}_i^\top,
  \end{equation}
  and
  \begin{equation}
    \begin{split}
      &\mathbb{E}_{\bm{v}_i}\left[|\bm{v}_i^\top\bm{y}|^{2+2\epsilon}\right]=\sum_{j=1}^{2^p}|\langle c\bm{v}_j,\bm{v}_i\rangle|^{2+2\epsilon}P_{\bm{v}_i}(\bm{y}=c\bm{v}_j)\\
      =&c^{2+2\epsilon}\frac{\gamma}{p}\sum_{j=1}^{2^p}|\langle\bm{v}_j,\bm{v}_i\rangle|^{2+2\epsilon} + c^{2+2\epsilon}\gamma\sum_{j=1}^{2^p}|\langle\bm{v}_j,\bm{v}_i\rangle|^{4+2\epsilon}\\
      \leq & c^{2+2\epsilon}\gamma p^{-2-\epsilon}2^{p+3}=M.
    \end{split}
  \end{equation}

  For $h=p/2$, by Lemma \ref{lemma:GV}, there exist $N\geq\exp(h/8)$ binary vectors $\bm{z}_1,\dots,\bm{z}_N\in\{0,1\}^p$ such that $d_\text{H}(\bm{z}_j,\bm{z}_k)\geq h/4$ for all $1\leq j\neq k\leq h$, where $d_\text{H}(\cdot,\cdot)$ is the Hamming distance that measures the number of different entries in two binary vectors.

  For each $\bm{z}_j=(z_{j1},\dots,z_{jh})^\top$, $j=1,\dots,N$, let
  \begin{equation}
    \bm{w}_j=p^{-1/2}\left(z_{j1}(1,1)+(1-z_{j1})(1,-1),\dots,z_{jh}(1,1)+(1-z_{jh})(1,-1)\right)^\top
  \end{equation}
  and note that $\bm{w}_j\in\{p^{-1/2},-p^{-1/2}\}^p$.

  Consider a class of $N$ multivariate distributions $\mathcal{P}_{c,\gamma}=\{P_{\bm{w}_1},\dots,P_{\bm{w}_N}\}$. For the operator norm lower bound, we consider the independent setting, which is also a special case of strong mixing. Specifically, denote the distribution of $\{\bm{y}_t\}_{t=1}^\infty$ as $\mathbb{P}_{\bm{w}_j}^T$ with \textit{i.i.d.} data, where each $\bm{y}_t$ follows the distribution $P_{\bm{w}_j}$.\\

  \noindent\textit{Step 2. Establish the lower bound for $\bm{\Sigma}_0$ with $d=1$}

  For the multivariate distributions $P_{\bm{w}_j}$ and $P_{\bm{w}_k}$, the Kullback-Leibler (KL) divergence of these two distributions is defined as
  \begin{equation}
    \begin{split}
      \text{KL}(P_{\bm{w}_j},P_{\bm{w}_k})=&\sum_{i=1}^{2^p}P_{\bm{w}_j}(\bm{y}=c\bm{v}_i)\log\left(\frac{P_{\bm{w}_j}(\bm{y}=c\bm{v}_i)}{P_{\bm{w}_k}(\bm{y}=c\bm{v}_i)}\right)\\
      =&\frac{\gamma}{p}\sum_{i=1}^{2^p}\log\left(\frac{P_{\bm{w}_j}(\bm{y}=c\bm{v}_i)}{P_{\bm{w}_k}(\bm{y}=c\bm{v}_i)}\right)+\gamma\sum_{i=1}^{2^p}\langle\bm{v}_i,\bm{w}_j\rangle^2\log\left(\frac{P_{\bm{w}_j}(\bm{y}=c\bm{v}_i)}{P_{\bm{w}_k}(\bm{y}=c\bm{v}_i)}\right).
    \end{split}
  \end{equation}
  Note that
  \begin{equation}
    \sum_{i=1}^{2^p}\log\left(\frac{P_{\bm{w}_j}(\bm{y}=c\bm{v}_i)}{P_{\bm{w}_k}(\bm{y}=c\bm{v}_i)}\right)=\sum_{i=1}^{2^p}\log(P_{\bm{w}_j}(\bm{y}=c\bm{v}_i))-\sum_{i=1}^{2^p}\log(P_{\bm{w}_k}(\bm{y}=c\bm{v}_i))=0.
  \end{equation}
  In addition, as $\log(1+t)\leq t$ for $t\in(0,\infty)$, we have
  \begin{equation}
    \begin{split}
      & \sum_{i=1}^{2^p}\langle\bm{v}_i,\bm{w}_j\rangle^2\log\left(\frac{P_{\bm{w}_j}(\bm{y}-c\bm{v}_i)}{P_{\bm{w}_k}(\bm{y}-c\bm{v}_i)}\right)\\
      = & \sum_{i=1}^{2^p}\langle\bm{v}_i,\bm{w}_j\rangle^2[\log(P_{\bm{w}_j}(\bm{y}-c\bm{v}_i))-\log(P_{\bm{w}_k}(\bm{y}-c\bm{v}_i))]\\
      = & \sum_{i=1}^{2^p}\langle\bm{v}_i,\bm{w}_j\rangle^2\log(P_{\bm{w}_j}(\bm{y}-c\bm{v}_i)) - \sum_{i=1}^{2^p}\langle\bm{v}_i,\bm{w}_k\rangle^2\log(P_{\bm{w}_k}(\bm{y}-c\bm{v}_i))\\
      + & \sum_{i=1}^{2^p}(\langle\bm{v}_i,\bm{w}_k\rangle^2-\langle\bm{v}_i,\bm{w}_j\rangle^2)\log(P_{\bm{w}_k}(\bm{y}-c\bm{v}_i))\\
      = & \sum_{i=1}^{2^p}(\langle\bm{v}_i,\bm{w}_k\rangle^2-\langle\bm{v}_i,\bm{w}_j\rangle^2)\log(1+p\langle\bm{v}_i,\bm{w}_k\rangle^2)\\
      \leq & p\sum_{i=1}^{2^p}(\langle\bm{v}_i,\bm{w}_k\rangle^2-\langle\bm{v}_i,\bm{w}_j\rangle^2)\langle\bm{v}_i,\bm{w}_k\rangle^2\\
      \leq & p\sum_{i=1}^{2^p}\langle\bm{v}_i,\bm{w}_k\rangle^4 = \frac{2^{p+1}}{p^2}(1.5p-1) \leq 2^{p+2}p^{-1}.
    \end{split}
  \end{equation}
  Hence, we have that $\text{KL}(P_{\bm{w}_j},P_{\bm{w}_k})\leq \gamma p^{-1}2^{p+2}$.

  For $1\leq j\neq k\leq N$, $d_H(\bm{z}_j,\bm{z}_k)\geq h/4=p/8$. Hence, we have
  \begin{equation}
    \begin{split}
      & \|\bm{\Sigma}_{\bm{w}_j}-\bm{\Sigma}_{\bm{w}_k}\|_\textup{op}^2 = \frac{2^{2p+2}}{p^4}c^4\gamma^2\|\bm{w}_j\bm{w}_j^\top-\bm{w}_k\bm{w}_k^\top\|_\textup{op}^2 \\
      \geq&\frac{2^{2p+1}}{p^4}c^4\gamma^2\|\bm{w}_j\bm{w}_j^\top-\bm{w}_k\bm{w}_k^\top\|_\textup{F}^2\\
      \geq&\frac{2^{2p+1}}{p^4}c^4\gamma^2\frac{1}{p^2}\|(\sqrt{p}\bm{w}_j)(\sqrt{p}\bm{w}_j)^\top-(\sqrt{p}\bm{w}_k)(\sqrt{p}\bm{w}_k)^\top\|_\textup{F}^2\\
      \geq&\frac{2^{2p+1}}{p^4}c^4\gamma^2\frac{32}{p^2}\sum_{1\leq s,r\leq h}1\{z_{js}\neq z_{ks},z_{jr}\neq z_{kr}\}\\
      \geq&\frac{2^{2p+1}}{p^4}c^4\gamma^2\frac{32}{p^2}\left(\frac{p}{8}\right)^2=\frac{2^{2p-4}}{p^4}c^4\gamma^2.
    \end{split}
  \end{equation}

  Since $p<128T$, we take $\gamma=p^22^{-(p+8)}T^{-1}$ and have that 
  \begin{equation}
    \frac{2^{p-2}}{p^2}c^2\gamma=\frac{1}{32}M^{\frac{1}{1+\epsilon}}(p^{-1}2^{p+3}\gamma)^{\frac{\epsilon}{1+\epsilon}}=\frac{1}{32}M^{\frac{1}{1+\epsilon}}\left(\frac{p}{32T}\right)^{\frac{\epsilon}{1+\epsilon}}\geq\frac{1}{1024}M^{\frac{1}{1+\epsilon}}\left(\frac{p}{T}\right)^{\frac{\epsilon}{1+\epsilon}}.
  \end{equation}
  Also, for any $1\leq j,k\leq N$,
  \begin{equation}
    \text{KL}(\mathbb{P}^T_{\bm{w}_j},\mathbb{P}^T_{\bm{w}_k})=T\cdot\text{KL}(P_{\bm{w}_j},P_{\bm{w}_k})\leq T\gamma p^{-1}2^{p+2}=\frac{p}{64}.
  \end{equation}

  By the Fano's inequality in Lemma \ref{lemma:Fano}, when $p\geq20>64\log2$, we have
  \begin{equation}
    \begin{split}
      & \inf_{\widehat{\bm{\Sigma}}}\max_{j\in\{1,\dots,N\}}\mathbb{P}\left[\|\widehat{\bm{\Sigma}}-\bm{\Sigma}_{\bm{w}_j}\|_\textup{op}\geq\frac{1}{1024}\left(\frac{M^{1/\epsilon}p}{T}\right)^{\frac{\epsilon}{1+\epsilon}}\right]\\
      \geq & 1-\frac{N^{-2}\sum_{1\leq j,k\leq N}\text{KL}(\mathbb{P}^T_{\bm{w}_j},\mathbb{P}^T_{\bm{w}_k})+\log2}{\log N}\\
      \geq & 1-\frac{(p/64)+\log2}{h/8}\geq1-\frac{(p/64)+(p/64)}{p/16}=\frac{1}{2}.
    \end{split}
  \end{equation}

  Therefore, as $\mathbb{P}_{\bm{w}_1}^T,\mathbb{P}_{\bm{w}_2}^T,\dots,\mathbb{P}_{\bm{w}_N}^T\in\mathcal{P}_\text{V}(M,\epsilon,r)$, the required minimax lower bound for $\bm{\Sigma}_0$ with $d=1$ is obtained.\\

  \noindent\textit{Step 3. Extension to the case of $d>1$ and $\bm{\Sigma}_1$}

  For the case of $d>1$, consider the proposed distribution $\mathbb{P}_{\bm{w}_1}^T,\mathbb{P}_{\bm{w}_2}^T,\dots,\mathbb{P}_{\bm{w}_N}^T$ with \textit{i.i.d.} samples. Let $\bm{x}_t=(\bm{y}_{t-1}^\top,\dots,\bm{y}_{t-d}^\top)^\top$, then $\mathbb{E}_{\bm{w}_j}[\bm{x}_t\bm{x}_t^\top]$ is a block diagonal matrix with $d$ blocks $\bm{\Sigma}_{\bm{w}_j}$. Hence, the minimax lower in the operator norm can be obtained accordingly.

  For $\bm{\Sigma}_1$, without loss of generality, assume that $p$ is a multiple of 4, i.e., $p=4h$, where $h$ is a positive integer. Let $\bm{y}_t=(\bm{y}_{1t}^\top,\bm{y}_{2t}^\top)^\top$, where each $\bm{y}_{it}$ is a $2h$-dimensional vector. Assume that $\bm{y}_{1t}=\bm{y}_{2,t-1}$ almost surely, and $\bm{y}_{1t}$ is a sequence of \textit{i.i.d.} random vectors. Then,
  \begin{equation}
    \mathbb{E}[\bm{y}_t\bm{y}_{t-1}^\top]=\begin{bmatrix}
      \mathbb{E}[\bm{y}_{1t}\bm{y}_{1,t-1}^\top] & \mathbb{E}[\bm{y}_{1t}\bm{y}_{2,t-1}^\top]\\
      \mathbb{E}[\bm{y}_{2t}\bm{y}_{1,t-1}^\top] & \mathbb{E}[\bm{y}_{2t}\bm{y}_{2,t-1}]
    \end{bmatrix}=
    \begin{bmatrix}
      \bm{0}_{2h\times 2h} & \mathbb{E}[\bm{y}_{1t}\bm{y}_{1t}^\top]\\
      \bm{0}_{2h\times 2h} & \bm{0}_{2h\times 2h}
    \end{bmatrix}.
  \end{equation}
  Hence, we can apply the constructed distributions $\mathbb{P}_{\bm{w}_1}^T,\mathbb{P}_{\bm{w}_2}^T,\dots,\mathbb{P}_{\bm{w}_N}^T$ to $\bm{y}_{1t}$, where each $\bm{w}_j\in\{(2h)^{-1/2},(2h)^{-1/2}\}^{2h}$. The required minimax lower bound can be obtained similarly and the proof is omitted for brevity.

\end{proof}

\begin{proof}[\textbf{Proof of Theorem \ref{thm:LB_lowrank_VAR}}]

  The proof consists of three steps. In the first step, for the special case of $r_0=1$ and $d=1$, we apply the proposed multivariate discrete distributions to construct a special time series process. In the second step, we apply Fano's inequailty to obtain the minimax lower bound. In the last step, the lower bound result is extended to the general case of $r_0>1$ and $d>1$.\\

  \noindent\textit{Step 1. Construct a class of discrete time series}

  For simplicity, we first focus on the case with $r_0=1$ and $d=1$. Without loss of generality, we assume that $p$ is a multiple of $4$, that is, $p=4h$, where $h$ is a positive integer. For $\bm{y}_t\in\mathbb{R}^p$, we split it as  $\bm{y}_t=(\bm{y}_{1t}^\top,\bm{y}_{2t}^\top)^\top$, where each $\bm{y}_{it}$ is a $2h$-dimensional vector.

  Consider the following data generating process of $\bm{y}_t$: $\bm{y}_{1t}$ is independent with $\bm{y}_{2t}$ and all historical information $\{\bm{y}_s\}_{s<t}$, and $\bm{y}_{2t}$ is only related to $\bm{y}_{1,t-1}$. Let $\bm{u}_t=(\bm{y}_{1,t-1}^\top,\bm{y}_{2t}^\top)^\top\in\mathbb{R}^p$.

  Similarly to the step 1 in the proof of Proposition \ref{prop:LB_operator}, we consider the multivariate distribution $P_{\bm{v}_1},\dots,P_{\bm{v}_{2^p}}$. As discussed in the proof of Proposition \ref{prop:LB_operator}, $P_{\bm{w}_j}\in\mathcal{P}_{\textup{V}}(M,\epsilon,r)$, for $j=1,\dots,N$. For $h=p/4$, by Lemma \ref{lemma:GV}, there exist $N\geq\exp(h/8)$ binary vectors $z_1,\dots,z_N\in\{0,1\}^p$ such that $d_\text{H}(\bm{z}_j,\bm{z}_k)\geq h/4$ for all $1\leq j\neq k\leq N$, where $d_\text{H}(\cdot,\cdot)$ is the Hamming distance that measures the number of different entries in two binary vectors.

  For each $\bm{z}_j=(z_{j1},\dots,z_{jh})^\top$, $j=1,\dots,N$, let
  \begin{equation}
    \bar{\bm{w}}_j=p^{-1/2}(z_{j1}(1,1)+(1-z_{j1})(1,-1),\dots,z_{jh}(1,1)+(1-z_{jh})(1,-1))^\top
  \end{equation}
  and $\bm{w}_j=(\bar{\bm{w}}^\top_j,\bar{\bm{w}}^\top_j)^\top\in\{p^{-1/2},-p^{-1/2}\}^p$.

  Based on the vectors $\bm{w}_1,\dots,\bm{w}_N$, consider $N$ distributions $P_{\bm{w}_1},\dots,P_{\bm{w}_N}$ for the vector $\bm{u}_t$. According to the proof of Proposition \ref{prop:LB_operator}, under the distribution $P_{\bm{w}_j}$
  \begin{equation}
    \begin{split}
      \mathbb{E}_{\bm{w}_j}[\bm{u}_t\bm{u}_t^\top]=&
      \begin{bmatrix}
        \mathbb{E}_{\bm{w}_j}[\bm{y}_{1t}\bm{y}_{1t}^\top] & \mathbb{E}_{\bm{w}_j}[\bm{y}_{1,t-1}\bm{y}_{2t}^\top]\\
        \mathbb{E}_{\bm{w}_j}[\bm{y}_{2t}\bm{y}_{1,t-1}^\top] & \mathbb{E}_{\bm{w}_j}[\bm{y}_{2t}\bm{y}_{2t}^\top]
      \end{bmatrix}\\
      =&\left[\frac{2^p}{p^2}+\frac{2^p(p-2)}{p^3}\right]c^2\gamma\bm{I}_p + \frac{2^{p+1}}{p^2}c^2\gamma\bm{w}_j\bm{w}_j^\top\\
      =&\left[\frac{2^p}{p^2}+\frac{2^p(p-2)}{p^3}\right]c^2\gamma\bm{I}_p + \frac{2^{p+1}}{p^2}c^2\gamma\begin{bmatrix}
        \bar{\bm{w}}_j\bar{\bm{w}}_j^\top & \bar{\bm{w}}_j\bar{\bm{w}}_j^\top\\
        \bar{\bm{w}}_j\bar{\bm{w}}_j^\top & \bar{\bm{w}}_j\bar{\bm{w}}_j^\top
      \end{bmatrix}.
    \end{split}
  \end{equation}
  Hence, we can further obtain
  \begin{equation}
    \begin{split}
      \mathbb{E}_{\bm{w}_j}[\bm{y}_t\bm{y}_t^\top]=&\begin{bmatrix}
        \mathbb{E}_{\bm{w}_j}[\bm{y}_{1t}\bm{y}_{1t}^\top] & \mathbb{E}_{\bm{w}_j}[\bm{y}_{1t}\bm{y}_{2t}^\top]\\
        \mathbb{E}_{\bm{w}_j}[\bm{y}_{2t}\bm{y}_{1t}^\top] & \mathbb{E}_{\bm{w}_j}[\bm{y}_{2t}\bm{y}_{2t}^\top]
      \end{bmatrix}
      =\begin{bmatrix}
        \bm{M}_{\bar{\bm{w}}_j} & \bm{0}_{2h\times 2h} \\
        \bm{0}_{2h\times 2h} & \bm{M}_{\bar{\bm{w}}_j}
        \end{bmatrix}
    \end{split}
  \end{equation}
  and
  \begin{equation}
    \begin{split}
      \mathbb{E}_{\bm{w}_j}[\bm{y}_t\bm{y}_{t-1}^\top]=&\begin{bmatrix}
        \mathbb{E}_{\bm{w}_j}[\bm{y}_{1t}\bm{y}_{1,t-1}^\top] & \mathbb{E}_{\bm{w}_j}[\bm{y}_{1t}\bm{y}_{2,t-1}^\top]\\
        \mathbb{E}_{\bm{w}_j}[\bm{y}_{2t}\bm{y}_{1,t-1}^\top] & \mathbb{E}_{\bm{w}_j}[\bm{y}_{2t}\bm{y}_{2,t-1}^\top]
      \end{bmatrix}
      =\begin{bmatrix}
        \bm{0}_{2h\times 2h} & c^2\gamma 2^{p+1}p^{-2}\bar{\bm{w}}_j\bar{\bm{w}}_j^\top \\
        \bm{0}_{2h\times 2h} & \bm{0}_{2h\times 2h}
        \end{bmatrix},
    \end{split}
  \end{equation}
  where
  \begin{equation}
    \begin{split}
      \bm{M}_{\bar{\bm{w}}_j}=\left[\frac{2^p}{p^2}+\frac{2^p(p-2)}{p^3}\right]c^2\gamma\bm{I}_{2h} + \frac{2^{p+1}}{p^2}c^2\gamma\bar{\bm{w}}_j\bar{\bm{w}}_j^\top:=c^2\gamma (C_1\bm{I}_{2h}+C_2\bar{\bm{w}}_j\bar{\bm{w}}_j^\top)
    \end{split}
  \end{equation}
  and
  \begin{equation}
    \bm{M}_{\bar{\bm{w}}_j}^{-1}=c^{-2}\gamma^{-1}C_1^{-1}[\bm{I}_{2h}-2C_2/(2C_1+C_2)\bar{\bm{w}}_j\bar{\bm{w}}_j^\top].
  \end{equation}

  Denote by $\mathbb{P}_{\bm{w}_j}^T$ the distribution of $\{\bm{y}_t\}_{t=1}^T$, where $\bm{u}_t$ follows the distribution $P_{\bm{w}_j}$.
  By Yule--Walker equation,
  \begin{equation}
    \bm{A}^*(P_{\bm{w}_j})=\mathbb{E}_{\bm{w}_j}[\bm{y}_t\bm{y}_{t-1}^\top]\mathbb{E}_{\bm{w}_j}[\bm{y}_t\bm{y}_t^\top]^{-1}=
    \begin{bmatrix}
      \bm{0}_{2h\times 2h} & \bm{A}^*_{1,2}(P_{\bm{w}_j})\\
      \bm{0}_{2h\times 2h} & \bm{0}_{2h\times 2h}
    \end{bmatrix}
  \end{equation}
  is a rank-1 matrix, where
  \begin{equation}
    \begin{split}
      & \bm{A}^*_{1,2}(P_{\bm{w}_j})=c^2\gamma 2^{p+1}p^{-2}\bar{\bm{w}}_j\bar{\bm{w}}_j^\top\bm{M}_{\bar{\bm{w}}_j}^{-1} = \frac{2C_2}{2C_1+C_2}\bar{\bm{w}}_j\bar{\bm{w}}_j^\top\\
      = & \left[c^{-2}\gamma^{-1}C_1^{-1}\left(1-\frac{2C_2}{2C_1+C_2}\right)\right]c^2\gamma C_1\left(1-\frac{2C_2}{2C_1+C_2}\right)^{-1}\frac{2C_2}{2C_1+C_2}\bar{\bm{w}}_j\bar{\bm{w}}_j^\top\\
      = & \left\|\mathbb{E}_{\bm{w}_j}[\bm{y}_t\bm{y}_t^\top]^{-1}\right\|_\textup{op} C_1C_3c^2\gamma\bar{\bm{w}}_j\bar{\bm{w}}_j^\top
    \end{split}
  \end{equation}
  and $C_3=2C_2/(2C_1-C_2)=p/(p-2)$. In addition, note that $\|\mathbb{E}_{\bm{w}_j}[\bm{y}_t\bm{y}_t^\top]^{-1}\|_\textup{op}:=C_4$ is a constant independent of the choice of $\bm{w}_j$.\\

  \noindent\textit{Step 2. Establish the lower bound with $r_0=1$ and $d=1$}

  For $1\leq j\neq k\leq N$, as $d_\text{H}(\bm{w}_j,\bm{w}_k)\geq h/4=p/16$,
  \begin{equation}
    \begin{split}
      &\|\bm{A}^*(P_{\bm{w}_j})-\bm{A}^*(P_{\bm{w}_k})\|_\textup{op}^2\\
      =&C_1^2C_3^2C_4^2c^4\gamma^2\|\bar{\bm{w}}_j\bar{\bm{w}}_j^\top-\bar{\bm{w}}_k\bar{\bm{w}}_k^\top\|_\textup{F}^2\\
      \geq&C_1^2C_3^2C_4^2c^4\gamma^2\|\bar{\bm{w}}_j\bar{\bm{w}}_j^\top-\bar{\bm{w}}_k\bar{\bm{w}}_k^\top\|_\textup{F}^2\\
      \geq&C_4^2\frac{2^{2p}}{p^4} c^4\gamma^2\frac{1}{p^2}\|(\sqrt{p}\bar{\bm{w}}_j)(\sqrt{p}\bar{\bm{w}}_j)^\top-(\sqrt{p}\bar{\bm{w}}_k){\sqrt{p}\bar{\bm{w}}_k^\top}\|_\textup{F}^2 \\
      \geq&C_4^2\frac{2^{2p}}{p^4}c^4\gamma^2.
    \end{split}
  \end{equation}
  
  Similarly to the proof of Proposition \ref{prop:LB_operator}, note that
  \begin{equation}
    \frac{2^p}{p^2}c^2\gamma=\frac{1}{8}M^{\frac{1}{1+\epsilon}}(p^{-1}2^{p+3}\gamma)^{\frac{\epsilon}{1+\epsilon}}.
  \end{equation}
  Taking $\gamma=p^22^{-(p+8)}T^{-1}$, we have that for any $1\leq j,k\leq N$,
  \begin{equation}
    \text{KL}(\mathbb{P}_{\bm{w}_j}^T,\mathbb{P}_{\bm{w}_k}^T) = T\gamma p^{-1}2^{p+2}=\frac{p}{128}.
  \end{equation}
  By Fano's inequality in Lemma \ref{lemma:Fano}, when $p\geq39>128\log(2)$,
  \begin{equation}
    \begin{split}
      & \inf_{\widehat{\bm{A}}}\max_{j\in\{1,\dots,N\}}\mathbb{P}\left[\|\widehat{\bm{A}}-\bm{A}^*(P_{\bm{w}_j})\|_\textup{F}\geq\frac{1}{256}\|\bm{\Sigma}_0(P_{\bm{w}_j})^{-1}\|_\textup{op}M^{\frac{1}{1+\epsilon}}(p/T)^{\frac{\epsilon}{1+\epsilon}}\right]\\
      \geq & 1-\frac{N^{-2}\sum_{1\leq j,k\leq N}\text{KL}(\mathbb{P}_{\bm{w}_j}^T,\mathbb{P}_{\bm{w}_k}^T)+\log2}{\log N}\\
      \geq & 1-\frac{p/128+\log(2)}{p/32}\geq\frac{1}{2},
    \end{split}
  \end{equation}
  which concludes the proof of the case with $r_0=1$ and $d=1$.\\

  \noindent\textit{Step 3. Extension to the general case of $r_0>1$}

  For the case of $r_0>1$, we consider that $\bm{y}_{t}$ can be split into $2r_0$ components $\bm{y}_t=[\bm{y}_{1t}^{(1)\top},\bm{y}_{2t}^{(1)\top},\bm{y}_{1t}^{(2)\top},\bm{y}_{2t}^{(2)\top},\dots,\bm{y}_{1t}^{(r_0)\top},\bm{y}_{2t}^{(r_0)\top}]$, where $\{(\bm{y}_{1t}^{(j)\top},\bm{y}_{2t}^{(j)\top})^\top\}_{t=1}^T$ and $\{(\bm{y}_{1t}^{(k)\top},\bm{y}_{2t}^{(k)\top})^\top\}_{t=1}^T$ are independent for $1\leq j,k\leq r_0$. Then, we can apply the same technique in the steps 1 and 2 to each $\{(\bm{y}_{1t}^{(j)\top},\bm{y}_{2t}^{(j)\top})^\top\}_{t=1}^T$ and obtain the required lower bound. In addition, the lower bound result for $d=1$ directly implies that for $d>1$, so the proof of $d>1$ is omitted for brevity.

\end{proof}

\subsection{Auxiliary Lemmas}\label{sec:B.3}

In this section, we present some auxiliary lemmas used in the proofs of lower bound results.
The first lemma presents the number of vectors satisfying a certain Hamming distance, which is known as the Gilbert--Varshamov Lemma in \citet{massart2007concentration}.

\begin{lemma}\label{lemma:GV}
  There exist binary vectors $\bm{z}_1,\dots,\bm{z}_m\in\{0,1\}^p$ such that
  \begin{itemize}
    \item[1.] $d_H(\bm{z}_j,\bm{z}_k)\geq p/4$ for all $j\neq k$;
    \item[2.] $m\geq\exp(p/8)$.
  \end{itemize}
\end{lemma}

The second lemma is a well known result from information theory called Fano's inequality \citep{wainwright2019high}.

\begin{lemma}\label{lemma:Fano}
  Let $P_1,\dots,P_M$ be $M$ probability distributions, $M\geq2$. Then,
  \begin{equation}
    \inf_{\psi}\max_{1\leq j\leq M}P_j[\psi(X)\neq j]\geq 1-\frac{M^{-2}\sum_{j,k=1}^M\textup{KL}(P_j,P_k)+\log2}{\log M}
  \end{equation}
  where $\textup{KL}(\cdot,\cdot)$ is the Kullback--Leibler (KL) divergence of two distributions and the infimum is taken over all tests with values in $\{1,2,\dots,M\}$.
\end{lemma}

\section{ADMM Algorithms}\label{append:C}

In this appendix, we present the detailed ADMM algorithm for $\ell_1$ regularized sparse VAR model and nuclear norm regularized reduced-rank VAR model.

\subsection{$L_1$-regularized sparse VAR model}

For a matrix $\bm{A}=(A_{ij})$, define the element-wise soft thresholding operator $S_{\kappa}(\bm{A})=(S_{\kappa}(A_{ij}))_{ij}$, where
\begin{equation}
  S_\kappa(x)=\begin{cases}
    x-\kappa, & x>\kappa;\\
    0, & -\kappa\leq x\leq \kappa;\\
    x+\kappa, & x<-\kappa.
  \end{cases}
\end{equation}
The ADMM algorithm proposed for $\ell_1$ regularized sparse VAR model has closed-form iterative updates summarized in Algorithm \ref{alg:sparse}.

\begin{algorithm}
  \begin{flushleft}
    \textbf{Input}: $\bm{\widetilde{\Sigma}}_0$, $\bm{\widetilde{\Sigma}}_1$, $\lambda$, $\rho$, $\mu$, $J$\\
    \textbf{Initialize}: $\bm{A}^{(0)}$, $\bm{D}^{(0)}$, $\bm{W}^{(0)}$\\
    \textbf{for} $j=1,\dots,J$\\
    \hspace*{1cm} $\bm{A}^{(j+1)}=S_{2/(\rho\mu)}(\bm{A}^{(j)}-2/\mu(\bm{A}^{(j)}\bm{\widetilde{\Sigma}}_0-\bm{D}^{(j)}-\bm{\widetilde{\Sigma}}_1-\bm{W}^{(j)}/\rho)\bm{\widetilde{\Sigma}}_0)$\\
    \hspace*{1cm} $\bm{D}^{(j+1)}=\text{sign}(\bm{A}^{(j+1)}\bm{\widetilde{\Sigma}}_0-\bm{\widetilde{\Sigma}}_1-\bm{W}^{(j)}/\rho)\cdot(|\bm{A}^{(j+1)}\bm{\widetilde{\Sigma}}_0-\bm{\widetilde{\Sigma}}_1-\bm{W}^{(j)}/\rho|\wedge \lambda)$\\
    \hspace*{1cm} $\bm{W}^{(j+1)}=\bm{W}^{(j)}+\rho(\bm{\widetilde{\Sigma}}_1-\bm{A}^{(j+1)}\bm{\widetilde{\Sigma}}_0+\bm{D}^{(j+1)})$\\
    \textbf{end for}\\
    \textbf{Return}: $\bm{A}^{(J)}$
  \end{flushleft}
  \caption{ADMM algorithm for $\ell_1$ regularized sparse VAR model}
  \label{alg:sparse}
\end{algorithm}

\subsection{Nuclear norm regularized reduced-rank VAR model}

For a matrix $\bm{A}$ with singular value decomposition $\bm{A}=\bm{U}\bm{D}\bm{V}^\top$, define the singular value soft thresholding operator $\text{SSV}_\kappa(\bm{A})=\bm{U}S_\kappa(\bm{D})\bm{V}^\top$. For the $\bm{A}$-update, we can apply the singular value soft thresholding operator to obtain the closed-form solution, while the $\bm{D}$-update can be achieved by truncating the singular values of $\bm{\widetilde{\Sigma}}_1-\bm{A}^{(j+1)}\bm{\widetilde{\Sigma}}_0+\bm{W}^{(j)}$. Thus, the ADMM algorithm proposed for nuclear norm regularized reduced-rank VAR model has closed-form iterative updates summarized in Algorithm \ref{alg:reduced-rank}.

\begin{algorithm}
  \begin{flushleft}
    \textbf{Input}: $\bm{\widetilde{\Sigma}}_0$, $\bm{\widetilde{\Sigma}}_1$, $\lambda$, $\rho$, $\mu$, $J$\\
    \textbf{Initialize}: $\bm{A}^{(0)}$, $\bm{D}^{(0)}$, $\bm{W}^{(0)}$\\
    \textbf{for} $j=1,\dots,J$\\
    \hspace*{1cm} $\bm{A}^{(j+1)}=\text{SSV}_{2/(\rho\mu)}(\bm{A}^{(j)}-2/\mu(\bm{A}^{(j)}\bm{\widetilde{\Sigma}}_0-\bm{D}^{(j)}-\bm{\widetilde{\Sigma}}_1-\bm{W}^{(j)}/\rho)\bm{\widetilde{\Sigma}}_0)$\\
    \hspace*{1cm} $\bm{U},\bm{S},\bm{V}=\text{SVD}(\bm{A}^{(j+1)}\bm{\widetilde{\Sigma}}_0-\bm{\widetilde{\Sigma}}_1-\bm{W}^{(j)}/\rho)$\\
    \hspace*{1cm} $\bm{D}^{(j+1)}=\bm{U}(\bm{S}\wedge\lambda)\bm{V}^\top$\\
    \hspace*{1cm} $\bm{W}^{(j+1)}=\bm{W}^{(j)}+\rho(\bm{\widetilde{\Sigma}}_1-\bm{A}^{(j+1)}\bm{\widetilde{\Sigma}}_0+\bm{D}^{(j+1)})$\\
    \textbf{end for}\\
    \textbf{Return}: $\bm{A}^{(J)}$
  \end{flushleft}
  \caption{ADMM algorithm for nuclear norm regularized reduced-rank VAR model}
  \label{alg:reduced-rank}
\end{algorithm}

\end{appendix}

\end{document}